%% file: main.tex
\PassOptionsToPackage{letterpaper,margin=0cm,bottom=0.1cm}{geometry}
\documentclass[prx,11pt,letterpaper,twocolumn,accepted=2025-07-16]{quantumarticle}
\pdfoutput=1
\usepackage[letterpaper,margin=0cm,bottom=0.1cm]{geometry}

\usepackage[utf8]{inputenc}
\usepackage[english]{babel}
\usepackage[T1]{fontenc}
\usepackage{amsmath}
\usepackage{amsfonts}

\usepackage{microtype}

\microtypecontext{spacing=nonfrench}

\usepackage[all=normal,floats=tight,mathspacing=tight,wordspacing=tight,paragraphs=normal,tracking=tight,charwidths=tight,mathdisplays=normal,sections=normal,margins=normal]{savetrees}

\usepackage{tikz}
\usepackage{lipsum}
\usepackage{graphicx}

\usepackage{mathtools}

\usepackage{amsthm}
\newtheorem{theorem}{Theorem}

\newtheorem{assump}{Assumption} 

\usepackage{subcaption}
\usepackage{thmtools}
\usepackage{thm-restate}
\usepackage{amsfonts}
\usepackage{amsmath}
\usepackage{standalone}
\usepackage{physics}
\usepackage{float}

\usepackage[numbers,compress]{natbib}

\usepackage{url}
\usepackage{hyperref}

\usepackage{cleveref}

\begin{document}
\title{Connecting extended Wigner’s friend arguments and noncontextuality}

\author{Laurens Walleghem}
\affiliation{Department of Mathematics, University of York, Heslington, York YO10 5DD, United Kingdom}
\affiliation{International Iberian Nanotechnology Laboratory (INL), Av. Mestre Jos\'{e} Veiga, 4715-330 Braga, Portugal}
\author{Y{\`i}l{\`e} Y{\=\i}ng}
\affiliation{Perimeter Institute for Theoretical Physics, Waterloo, Ontario, Canada, N2L 2Y5}
\affiliation{Department of Physics and Astronomy, University of Waterloo, Waterloo, Ontario, Canada, N2L 3G1}
\author{Rafael Wagner}
\affiliation{Institute of Theoretical Physics, Ulm University, Albert-Einstein-Allee 11 89081, Ulm, Germany}
\affiliation{International Iberian Nanotechnology Laboratory (INL), Av. Mestre Jos\'{e} Veiga, 4715-330 Braga, Portugal}
\affiliation{Centro de F\'{i}sica, Universidade do Minho, Braga 4710-057, Portugal}
\affiliation{Department of Physics ``E. Fermi'', University of Pisa, Largo B. Pontecorvo 3, 56127 Pisa, Italy}
\author{David Schmid}
\affiliation{International Centre for Theory of Quantum Technologies, University of Gda{\'n}sk, 80-309 Gda\'nsk, Poland}

\begin{abstract}
The Local Friendliness argument is an extended Wigner's friend no-go theorem that provides strong constraints on the nature of reality---stronger even than those imposed by Bell's theorem or by noncontextuality arguments. In this work, we prove a variety of connections between Local Friendliness scenarios and Kochen-Specker noncontextuality. Specifically, we first show how one can derive new Local Friendliness inequalities using known tools and results from the literature on Kochen-Specker noncontextuality. In doing so, we provide a new derivation for some of the facets of the Local Friendliness polytope, and we prove that this polytope is equal to the Bell polytope in a wide range of extended Wigner's friend scenarios with multipartite agents and sequential measurements.  We then show how any possibilistic Kochen-Specker argument can be mathematically translated into a related proof of the Local Friendliness no-go theorem. In particular, we construct a novel kind of Local Friendliness scenario where a friend implements several compatible measurements (or joint measurements of these) in between the superobserver's operations on them.  We illustrate this with the well-known 5-cycle and Peres-Mermin contextuality arguments.
\end{abstract}

\maketitle
\begingroup
\renewcommand\thefootnote{}  
\footnotetext{LW and YY are co-first authors of this paper. \newline
Email: \href{mailto:laurens.walleghem@york.ac.uk}{laurens.walleghem@york.ac.uk}, 
\href{mailto:yying@pitp.ca}{yying@pitp.ca}}
\addtocounter{footnote}{-1}  
\endgroup

\tableofcontents

\section{Introduction} \label{sec:intro}

Recently, extensions~\cite{schmid2023review,brukner2017quantum,frauchiger2018quantum,bong2020strong} 
of the Wigner's friend thought experiment~\cite{Wigner1995} have led to novel no-go theorems for the foundations of quantum theory. These theorems illustrate that quantum theory does not provide a clear, consistent, and interpretation-free description of observers as quantum systems.

One of the most striking of these no-go theorems is the Local Friendliness (LF) no-go theorem~\cite{bong2020strong,cavalcanti2021implications,wiseman2023thoughtful}, which demonstrates an inconsistency between the assumptions of (i) absoluteness of observed outcomes, (ii) a notion of locality (that is weaker than Bell's local causality), and (iii) that closed physical systems (even those with macroscopic observers) constitute valid quantum systems that evolve unitarily.
Because the notion of locality used in this argument is strictly weaker than local causality, the resulting no-go theorem can be viewed as a stronger constraint on physical theories than Bell's theorem~\cite{bong2020strong,cavalcanti2021implications,yile2023}. This cements its place as one of the most significant known no-go theorems in quantum foundations. 

Relative to other important no-go theorems like those due to Bell~\cite{bell1964einstein,brunner2014bell}, Spekkens~\cite{spekkens2005contextuality}, Kochen-Specker~\cite{kochen1990onthe,budroni2022kochenspecker}, and Leggett-Garg~\cite{legget1985quantum,emary2013leggett,vitagliano2023leggett,schmid2024reviewreformulation}, the Local Friendliness no-go theorem has not yet been analyzed in significant detail. How wide a variety of scenarios support such a no-go argument? Given any such scenario, what are the correlations consistent with the Local Friendliness assumptions, and can quantum theory violate these? How do such arguments relate to other fundamental notions such as locality, noncontextuality, macroscopic realism, and so on?
Although some work has been done on these questions~\cite{bong2020strong,utreras2022extended,utreras2023allowing,yile2023,wiseman2023thoughtful,haddara2022possibilistic,cavalcanti2021implications,schmid2023review,haddara2024local}, much remains to be discovered.

In this work, we make progress on all three of these sorts of questions, with the unifying thread that all of our results are motivated by elucidating the connections between Local Friendliness and Kochen-Specker noncontextuality.\footnote{The relevance of noncontextuality in some other extended Wigner's friend scenarios has been studied, e.g., in Refs.~\cite{walleghem2023extended,nurgalieva2023multi,walleghem2024strong,szangolies2020quantum,montanhano2023contextuality}.} First, after reviewing the necessary background on KS noncontextuality and Local Friendliness in \cref{sec:background}, we provide methods for deriving new Local Friendliness inequalities in \cref{sec:LFasKSNC} using tools and results from the literature on Kochen-Specker noncontextuality. This allows us to provide a new derivation of and perspective on some (although not all) Local Friendliness inequalities. 
We then use these ideas to prove that the set of correlations consistent with Bell's notion of local causality coincides with the set of correlations consistent with Local Friendliness in a wide range of scenarios. Finally, in \cref{sec:KSNC_to_LF}, we show how one can translate any possibilistic Kochen-Specker argument into a corresponding proof of the failure of Local Friendliness. Two examples of this are worked out in detail, beginning with the 5-cycle~\cite{cabello2013simple,santos2021conditions} and Peres-Mermin~\cite{peres1997quantum} proofs of the failure of Kochen-Specker noncontextuality.
Sections \ref{sec:LFasKSNC} and \ref{sec:KSNC_to_LF} can be read independently of each other.

\section{Background} \label{sec:background}
\subsection{Review of Kochen-Specker noncontextuality}
\label{sec:recap_KSNC}

We first briefly review some concepts regarding Kochen-Specker noncontextuality (KSNC) that will be useful later. For a comprehensive introduction, we refer the reader to Refs.~\cite{amaral2018graph,budroni2022kochenspecker}.

Let $\mathcal{M}$ be a finite set of measurements, of which some subsets are compatible, that is, jointly measurable. 
These subsets are called \emph{contexts}, and a given measurement may appear in more than one context. A context is said to be maximal in a scenario if it is not a subset of any other context. For example, if we have 3 measurements labeled by $M_1,M_2,M_3$ in the same context (denoted as a set $\{M_1,M_2,M_3\}$), we can jointly measure them, obtaining a joint outcome $(m_1,m_2,m_3)$. A \emph{measurement scenario} 
$(\mathcal{M},\mathcal{C},O)$ is defined by a finite set of measurements $\mathcal{M}$, a finite set of maximal contexts $\mathcal{C}$, and a set of possible outcomes $O$ for each measurement.  

One can represent such a measurement scenario by a compatibility hypergraph---a hypergraph which associates a measurement to each node of the graph and a hyperedge to each context~\cite{amaral2018graph}. In specific measurement scenarios, it is sufficient to consider only on the 2-skeleton of the compatibility hypergraph, known as the \emph{compatibility graph}~\cite[Def. 2.2, pg. 15]{amaral2018graph} (see also Ref.~\cite{xu2019necessary}), in which case contexts are denoted by cliques in the graph. 
One can find an example of this in Figure~\ref{fig:NaBa} of Section~\ref{sec:compatibility_graph}. 

A behavior in such a measurement scenario is a collection of \emph{empirical correlations}, one for each context $C$, over the outcomes for the measurements in that context. An empirical correlation is said to satisfy \emph{no-disturbance} if, for any measurement appearing in more than one context, the marginal probabilities over its outcomes are independent of context.
For example, if $M_1$ is in the context with $C_1\coloneqq\{M_1,M_2\}$ and the context $C_2\coloneqq\{M_1,M_3\}$, namely, $M_1$ is compatible with $M_2$ and $M_3$, then the empirical correlation of the first context, $\wp_{C_1}(m_1,m_2)$, and the empirical correlation of the second context, $\wp_{C_2}(m_1,m_3)$, agree on their overlap\footnote{We say that distributions \emph{agree on their overlap} if they have the same marginals on the intersection of their domains.}, namely the probability for $m_1$:
\begin{align}
\sum_{m_2}\wp_{C_1}(m_1,m_2)=\sum_{C_2}\wp(m_1,m_3).
\end{align}
We use $\wp$ to denote empirical correlations---i.e., those that are in principle experimentally accessible, with respect to some physical theory. We use $P$ to denote generic probability distributions that are not necessarily experimentally accessible.

A behavior is said to be Kochen-Specker noncontextual (KSNC) if there is a single probability distribution $P_{NC}(\cdot)$ over all measurements in $\mathcal{M}$ whose marginals reproduce the empirical correlations for each context. Note that any KSNC behavior automatically satisfies no-disturbance, so it only makes sense to consider KSNC for no-disturbing behaviors.

If for some particular compatibility hypergraph, there is a behavior (or correlations) with no such joint probability distribution, then one calls such a behavior KS contextual, and it provides a proof of the KSNC no-go theorem.
One moreover says that the correlations are logically contextual or possibilistically contextual if the proof that no such joint probability distribution exists can be given while only keeping track of whether or not a given combination of measurement outcomes is possible or not, i.e., by looking at the support of the probability distributions in the argument~\cite{abramsky2016possibilities,abramsky2011sheaf}.

The set of all empirical correlations satisfying these assumptions in a given scenario is known as the Kochen-Specker noncontextuality polytope and is a function of the compatibility hypergraph for the scenario. A special case of a measurement scenario is a Bell scenario, in which the measurement contexts arise from a party structure where each party can choose from a certain set of measurements to perform on the system they own. In this case, the Kochen-Specker noncontextuality polytope is well known to coincide with the Bell polytope~\cite{abramsky2016possibilities,choudhary2024lifting,abramsky2011sheaf,barbosa2023closing,abramsky2012cohomology,barbosa2022continuous,araujo2013all,wagner2024inequalities}---that is, the polytope of correlations consistent with Bell's local causality.

As a final remark, we note that KSNC applies only to sharp measurements and requires an assumption of outcome determinism for such measurements~\cite{spekkens2014status,KunjwalKS2015}. (In quantum theory, these correspond to projective measurements.)
One can motivate the assumption of outcome determinism for sharp measurements---and more importantly, drop the assumption of sharpness---by appealing to the generalized notion of noncontextuality introduced by Spekkens~\cite{spekkens2005contextuality}, which is in any case motivated by similar considerations, namely, an appeal to a form of Leibniz's principle~\cite{spekkens2019ontologicalidentityempiricalindiscernibles,sep-identity-indiscernible}. KSNC can also be motivated (also without assuming outcome determinism) from no-fine-tuning of a classical causal model~\cite{cavalcantiClassical2018,Pearl_2021}.

\subsection{Local Friendliness scenarios} \label{sec:recap_LF_scenario}

We now introduce the Local Friendliness scenario, sketched in \Cref{fig:EWFS}, following Ref.~\cite{bong2020strong} but with minor modifications that will be convenient later on.  We briefly comment on the difference between our presentation and the usual presentation, as well as some other extended Wigner's friend scenarios in Appendix~\ref{app:compare}.

Consider four agents named Alice, Bob, Charlie, and Debbie. 
At the beginning of the experiment, Charlie and Debbie share a bipartite system $ST$, and each resides in a closed lab. Charlie performs a measurement (denoted $A_1$) on his share of the bipartite system $S$, obtaining outcome $a_1$. Debbie performs a measurement (denoted $B_1$) on her share $T$, obtaining outcome $b_1$. Charlie and Debbie do not have measurement choices.

Alice and Bob are so-called \emph{superobservers}, since Alice is assumed to be able to perform any operations (including coherent quantum transformations) on Charlie's lab, including on Charlie and $S$, and Bob is assumed to be able to perform any operations on Debbie's lab, including on Debbie and $T$. After Charlie and Debbie have performed their measurements, Alice and Bob make choices.

Alice's choices are labeled by $x\in \{1,2,\dots,N_A\}$, and Bob's by $y\in\{1,2,\dots,N_B\}$. Depending on the cardinalities of their choices, such a specific LF scenario is also referred to as the \emph{$N_A,N_B$-setting LF scenario}. 
If Alice chooses $x=1$, Alice opens Charlie's lab and reveals his outcome $a_1$ to the world. 
Otherwise, i.e., for $x \neq 1$, Alice performs a different operation on Charlie’s lab. One particular operation that she could perform (that is often used in proving the no-go theorem~\cite{bong2020strong}) is to first reverse the evolution of Charlie’s lab due to Charlie's measurement (restoring the state of $S$ to its initial state) and to then instruct Charlie to perform another measurement on $S$, whose outcome will be revealed subsequently. We will sometimes refer to such reversal operation as `Alice \emph{undoes} Charlie's measurement'. Charlie's new measurement is denoted $A_x$, and his respective outcome is denoted $a_x$. For instance, if $x=2$, the new measurement performed by Charlie is $A_2$ with outcome $a_2$. After Charlie has performed his new measurement, Alice reveals its outcome to the world. 

Similarly, when Bob chooses $y=1$, Bob opens Debbie's lab and reveals her outcome $b_1$ to the world. For $y \neq 1$, Bob performs a different operation on Charlie's lab. For instance, one can consider the case where Bob reverses the evolution of Debbie's lab due to Debbie's measurement, restoring the state of $T$ to its initial state, and then instructs Debbie to perform a new measurement $B_y$ on $T$, obtaining outcome $b_y$; this again is often used in proving the no-go theorem. After Debbie has performed her new measurement, Bob reveals her outcome to the world.

At the end of the experiment, we collect statistics for $a_x$, $b_y$, $x$, and $y$, which constitute the empirical probabilities $\wp(a_x,b_y|x,y)$. 
These probabilities constitute the empirical correlations (or behaviors) in this LF scenario. For simplicity, we also denote these empirical probabilities as $\wp(a_x,b_y)$ without the explicit conditionals since the conditionals are actually redundant --- the empirical correlation between $a_x$ and $b_y$ are only defined in the runs with the corresponding choices $x$ and $y$.

\begin{figure}[t]
         \centering
\includegraphics[width=0.9\columnwidth]{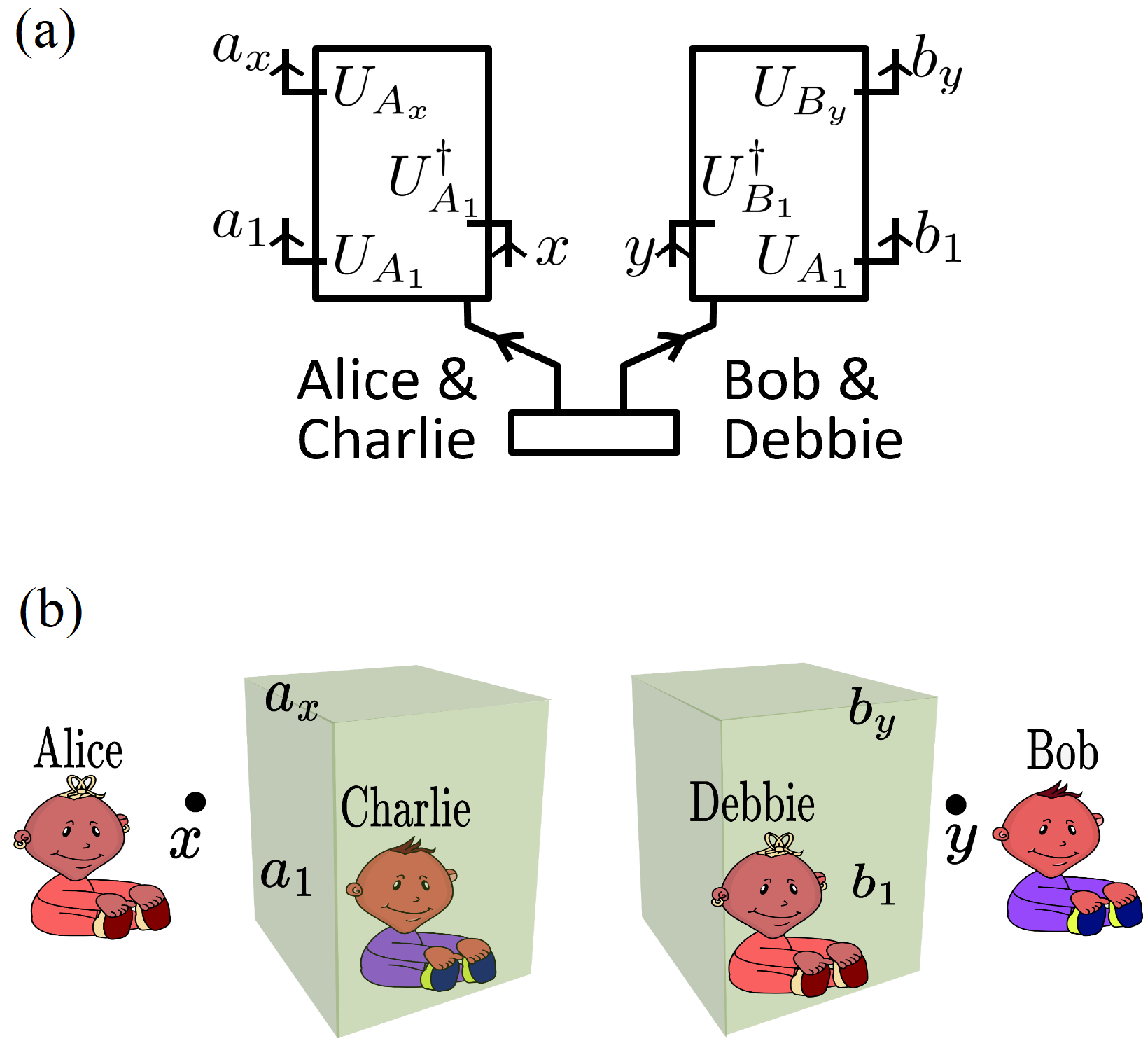}
         \caption{(a) a schematic representation of the LF scenario introduced in \cref{sec:recap_LF_scenario}; (b) a corresponding cartoon representation. Charlie and Debbie share an entangled state and each performs a measurement on their part, obtaining outcomes $a_1$ and $b_1$, respectively. Then, the superobservers Alice and Bob make their respective choices $x$ and $y$. For $x=1$ Alice reveals $a_1$, whereas for $x \in \{2,3,\ldots,N_A\}$ she reverses Charlie's first measurement, instructs Charlie to perform a measurement $A_x$ and reveals the outcome $a_x$. The protocol for Bob's choice $y$ is analogous.}
         \label{fig:EWFS}
\end{figure}

\subsection{The Local Friendliness no-go theorem}
\label{sec:LFnogo}

Within this scenario, one can prove the Local Friendliness no-go theorem. It contains two crucial assumptions called Absoluteness of Observed Events (AOE) and Local Agency (LA)~\cite{bong2020strong,cavalcanti2021implications,yile2023,schmid2023review},\footnote{Instead of Local Agency, one could also use the assumptions, No-Superdeterminism and Locality, as was done in Ref.~\cite{bong2020strong}.} which taken together are termed the \emph{Local Friendliness} assumptions. 

We now formally state these two assumptions and spell out their consequences for the LF scenario.  Due to the modifications we made to the LF scenario, these consequences are described slightly differently than usual, but in a way that better sets the stage for our purposes.  

\begin{assump}[Absoluteness of Observed Events (AOE)] \label{assumption:AOE}
An observed event is an absolute single event, not relative to anything or anyone. 
\end{assump}
In each individual run of the LF scenario, either two, three, or four measurements are performed, depending on Alice's and Bob's choices. Specifically, when $x=1,y=1$, only two measurements are performed while when  $x=2,y=2$, all four measurements are performed in a single run. Thus, depending on the choices, in each run, either two, three, or four outcomes are observed by the agents. AOE implies that for a given value of $x$ and a given value of $y$, there exists a joint distribution over the outcomes observed in the corresponding runs. For example, for $x=i,y=1$, measurements $A_1,A_i,B_1$ are performed and their respective outcomes exist in a single run. By AOE, for these runs, one can assign a well-defined joint distribution over their outcomes $a_1,a_i,b_1$, which we denote as $P(a_1,a_i,b_1|x{=}i,y{=}1)$.\footnote{Here we follow the common abuse of notation in the literature that the same letter such as $a_1$ is used to denote both the random variable representing the outcome and the value of the random variable. If one wishes to write the distribution in a more explicit way, one could use $\mathcal{A}_1, \mathcal{A}_i,\mathcal{B}_1$ for the random variables denoting the respective outcomes and use the lowercase letters for their values; then, the joint distribution can be written as $P(\mathcal{A}_1=a_1,\mathcal{A}_i=a_i,\mathcal{B}_1=b_1|x{=}i,y{=}1)$. } By the experimental setup, the empirical correlation $\wp(a_i,b_1)$ is its marginal, i.e., $P(a_1,a_i,b_1|x{=}i,y{=}1) = \wp(a_i,b_1)$. 
For all given pairs of values $x,y$ of choices, we have the following respective joint distributions:

\begin{subequations}
\label{eq:AOEdis}
\begin{align}
    P(a_1,b_1| x=1,y=1),\\
    P(a_1,a_i,b_1|x=i,y=1),\\
    P(a_1,b_1,b_j|x=1,y=j), \\
    P(a_1,a_i,b_1,b_j|x=i,y=j), \label{eq:AOEij}
\end{align}
\end{subequations}
for $i\in\{2,3,\dots,N_A\}$ and $j\in\{2,3,\dots,N_B\}$.

While these joint probability distributions are not empirically accessible, according to the experimental setup they must satisfy
\begin{subequations}
\label{eq_emp}
\begin{align}
    P(a_1,b_1|x=1,y=1)=\wp(a_1,b_1), \label{eq:AOEem} \\
  \sum_{a_1}P(a_1,a_i,b_1|x=i,y=1)=\wp(a_i,b_1),  \\
  \sum_{b_1}P(a_1,b_1,b_j|x=1,y=j)=\wp(a_1,b_j), \label{eq:empa1bj}\\
 \sum_{a_1,b_1}P(a_1,a_i,b_1,b_j|x=i,y=j)=\wp(a_i,b_j). \label{eq:marg11}
\end{align}
\end{subequations}
\begin{assump}[Local Agency] \label{assumption:Local_Agency}
    If a measurement setting is freely chosen, then it is uncorrelated with any set of relevant events not in its future-light-cone.
\end{assump}
Since in LF scenarios, we assume that none of $a_i$, $a_1$ or $b_1$ is in the future light cone of Bob's measurement setting $y$, and none of $b_j$, $a_1$ or $b_1$ is in the future light cone of Alice's measurement setting $x$, Local Agency implies that 
\begin{subequations}
    \label{eq:LA}
\begin{align}
P(a_1,b_1|x,y)&=P(a_1,b_1) \\
P(a_1,a_i,b_1|x=i,y) &= P(a_1,a_i,b_1|x=i), \label{eq:LAaib1}\\
P(a_1,b_1,b_j|x,y=j)&=P(a_1,b_1,b_j|y=j).\label{eq:LAa1bj}
\end{align}
\end{subequations}

The existence of the distributions in \cref{eq:AOEdis} and the conditions in \cref{eq_emp,eq:LA} constitute the \emph{LF constraints} on the
empirical correlations in the LF scenario, leading to the so-called \emph{LF inequalities}.
The set of empirical correlations consistent with all LF inequalities is termed the \emph{LF polytope}. In \cref{sec:LFasKSNC}, we will explore the details of LF inequalities for various LF scenarios, including generalizations to the ones defined in \cref{sec:recap_LF_scenario}. 

In Refs.~\cite{bong2020strong,wiseman2023thoughtful}, it is shown that some LF inequalities can be violated in a proposed quantum realization of the LF scenario;  hence, one has the following no-go theorem~\cite{bong2020strong}. 
\begin{theorem}[Local Friendliness no-go theorem]
If a superobserver can perform arbitrary quantum operations on an observer and its environment, then no physical theory can satisfy Local Friendliness.
\end{theorem}
 See Appendix~\ref{app:compare} for more details on those proposed quantum realizations and the LF theorem.

\section{Deriving Local Friendliness inequalities from KSNC constraints on subgraphs}
\label{sec:LFasKSNC}

In this section, we show how techniques from the study of Kochen-Specker noncontextuality can be used to derive many Local Friendliness inequalities.
We do so by considering Kochen-Specker inequalities relative to {\em subgraphs} of the compatibility graph naturally associated to the Local Friendliness scenario in question. Then, we show how the LF polytope coincides with a Bell polytope in a wide range of Local Friendliness scenarios (including multipartite systems and sequential measurements).

\subsection{The compatibility graph for a $N_A,N_B$-setting LF scenario}
\label{sec:compatibility_graph}

We start our results by applying the notion of a compatibility graph introduced in \cref{sec:recap_KSNC} to the $N_A,N_B$-setting LF scenario defined in \cref{sec:recap_LF_scenario}. In a $N_A,N_B$-setting LF scenario, all measurements done on Alice's side, i.e., $A_1,\dots,A_{N_{A}}$, are compatible with all measurements done on Bob's side, i.e., $B_1,\dots,B_{N_{A}}$. However, $A_1,\dots,A_{N_{A}}$ measurements are not necessarily compatible with each other, and similarly, $B_1,\dots,B_{N_{B}}$ measurements are not necessarily compatible with each other. Thus, the corresponding compatibility graph for the $N_A,N_B$-setting LF scenario is as shown in \cref{fig:NaBa} where each of $A_1,\dots,A_{N_{A}}$ is connected with each of $B_1,\dots,B_{N_{B}}$, but no connections among  $A_1,\dots,A_{N_{A}}$ or among $B_1,\dots,B_{N_{B}}$.

\begin{figure}[h!]
\includegraphics[width=0.15\textwidth]{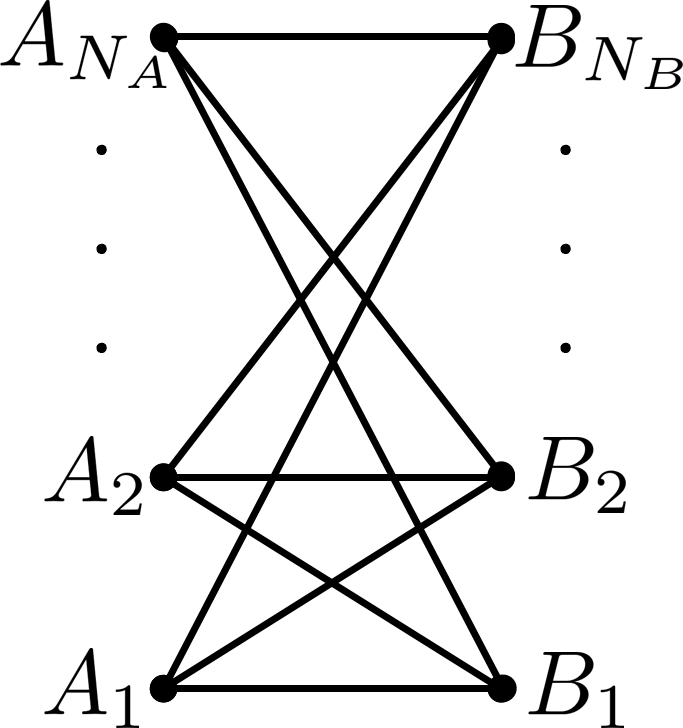}  
\caption{Compatibility graph for a $N_A,N_B$-setting LF scenario. Each node is associated with a measurement. The measurements connected by an edge are compatible and form a context. Every node on one side is connected by an edge with every node on the other side, forming a complete bipartite graph.} 
\label{fig:NaBa} 
\end{figure}

Such a compatibility graph is equivalent to a compatibility graph of a bipartite Bell scenario where in each run, Alice implements one of the measurements $A_1,\dots,A_{N_{A}}$ while Bob implements one of the measurements $B_1,\dots,B_{N_{B}}$.  We refer to such a bipartite Bell scenario as \emph{the corresponding Bell scenario for the $N_A,N_B$-setting LF scenario}. 

Note that the constraints coming from Local Friendliness, presented in \cref{sec:LFnogo}, imply the no-disturbance constraints (cf.~\cref{sec:recap_KSNC}) on the empirical  correlations in a $N_A,N_B$-setting LF scenario. These no-disturbance constraints are
\begin{subequations}
\label{eq:NS}
\begin{align}
    \sum_{b_l}\wp(a_k,b_l) = \sum_{b_{l'}}\wp(a_k,b_{l'}) \\
    \sum_{a_k}\wp(a_k,b_l) = \sum_{a_{k'}}\wp(a_{k'},b_{l}),
\end{align}
\end{subequations}
for all $k,k'\in\{1,2,\dots,N_A\}$ and $l,l'\in\{1,2,\dots,N_B\}$.
They are obtained by taking the appropriate marginals of distributions in \cref{eq:LA} and then connecting them with empirical correlations via \cref{eq_emp}. 
Furthermore, all (expected) empirical correlations in a quantum experimental realization of an LF scenario satisfy these no-disturbance relations (which can also be seen as no-superluminal-signaling here).

As the LF constraints imply no-disturbance, it makes sense to further ask if LF constraints can be understood as KSNC constraints in some sense, which is what we do next in this section. In particular, we will show that the LF constraints on the empirical correlations in the LF scenario sometimes (and we will see when) coincide with the KSNC constraints on the corresponding compatibility graph for some $N_A,N_B$ and always imply KSNC constraints on subgraphs of that graph. In the cases where they coincide (cf. \cref{sec:22,sec:32}), the KSNC constraints on the corresponding compatibility graph provide all LF inequalities, and the LF polytope coincides with the Bell polytope of the corresponding Bell scenario. In other cases (cf. \cref{sec:33}), the LF constraints on the empirical correlations in the LF scenario are \emph{strictly weaker} than the KSNC constraints on the corresponding compatibility graph. In such scenarios, some (not all) LF inequalities can be derived via demanding KSNC constraints on \emph{subgraphs} of the corresponding graph and the LF polytope strictly contains the Bell polytope. Finally, we generalize the LF scenario defined in \cref{sec:recap_LF_scenario} to multipartite sequential LF scenarios and we characterize a wide class of such scenarios where all the LF inequalities are KSNC constraints on the corresponding graph, meaning that the LF polytope coincides with the Bell polytope for the corresponding multipartite Bell scenario. 

Before deriving these results, we will first provide an alternative expression in \cref{sec:equi} to the LF constraints defined in \cref{sec:LFnogo}, which will aid in our investigation of the relationship between LF constraints and KSNC constraints.

\subsection{An alternative expression of the LF constraints}
\label{sec:equi}

We saw in the \cref{sec:LFnogo} that the constraints implied by LF in the $N_A,N_B$-setting LF scenario are the existence of probability distributions as in \cref{eq:AOEdis} and satisfying \cref{eq_emp,eq:LA}.
Next, we introduce an alternative expression of these constraints; one which will later be useful for making the connection to KSNC. Two of its immediate consequences are shown in \cref{sec:1stLFasKSNC} and \cref{sec:Bell_inside_LF}.

This alternative expression is the following. There must exist a joint distribution $P_{ij}(a_1, b_1, a_i, b_j)$ for any $i \in \{2,\ldots,N_A\},$ and $j \in \{2,\ldots,N_B\}$ with the following two properties: 1) each $P_{ij}(a_1,a_i,b_1, b_j)$ has as its marginals the four pairwise empirical correlations $\wp(a_i,b_j)$, $\wp(a_1,b_j)$, $\wp(a_i,b_1)$ and $\wp(a_1,b_1)$; 2) they agree on their overlaps, i.e., 
\begin{subequations}
\label{eq:agree}
\begin{align}
\sum_{a_i,b_j}P_{ij}(a_1,a_i,b_1,b_j)&=\sum_{a_{i'},b_{j'}}P_{ij}(a_1,a_{i'},b_1,b_{j'}), \label{eq:agreeab}\\
\sum_{b_j}P_{ij}(a_1,a_i,b_1,b_j)&= \sum_{b_{j'}}P_{ij'}(a_1,a_i,b_1,b_{j'}), \label{eq:agreeb}\\
\sum_{a_i}P_{ij}(a_1,a_i,b_1,b_j)&=\sum_{a_{i'}}P_{i'j}(a_1,a_{i'},b_1,b_j). \label{eq:agreea}
\end{align}
\end{subequations}

We will now prove that this is equivalent to the LF constraints defined in \cref{sec:LFnogo}. 

First, we will show that the LF constraints imply the existence of the distributions $\{P_{ij}(a_1, b_1, a_i, b_j)\}_{i \in \{2,\ldots,N_A\},j \in \{2,\ldots,N_B\}}$ with the two properties mentioned above.

The LF constraints include the existence of $P(a_1,a_i,b_1,b_j|x=i,y=j)$ for any $i\in\{2,3,\dots,N_A\}$ and $j\in\{2,3,\dots,N_B\}$ as per \cref{eq:AOEij}, a consequence of the assumption of AOE in the $N_A,N_B$-setting LF scenario. Now, let 
\begin{align}
\label{eq:defPij}
    P_{ij}(a_1,a_i,b_1,b_j) = P(a_1,a_i,b_1,b_j|x=i,y=j).
\end{align}
Then, \cref{eq_emp,eq:LA} together imply that
\begin{subequations}
    \label{eq:LAij}
\begin{align}
     \sum_{a_1,b_1} P_{ij}(a_1,a_i,b_1,b_j) 
    =\wp(a_i,b_j),  \label{eq:LA0}\\
     \sum_{a_i,b_1} P_{ij}(a_1,a_i,b_1,b_j) 
    =\wp(a_1,b_j),  \label{eq:LA1}\\
     \sum_{a_1,b_j} P_{ij}(a_1,a_i,b_1,b_j) 
    =\wp(a_i,b_1), \label{eq:LA2} \\
     \sum_{a_i,b_j} P_{ij}(a_1,a_i,b_1,b_j) 
    =\wp(a_1,b_1). \label{eq:LA3}
\end{align}
\end{subequations}
Here, the first equation \cref{eq:LA0} is equivalent to \cref{eq:marg11} using \cref{eq:defPij}. The second equation \cref{eq:LA1} comes from the fact that
\begin{align}
    &\sum_{a_i,b_1} P_{ij}(a_1,a_i,b_1,b_j) \nonumber \\
    = &\sum_{a_i,b_1} P(a_1,a_i,b_1,b_j|x=i,y=j) \nonumber \\
    =& P(a_1,b_j|x=i,y=j) \nonumber \\
    = & P(a_1,b_j|x=1,y=j) \nonumber \\
    = & \wp(a_1,b_j), 
\end{align}
where we used \cref{eq:LAa1bj} for the third equality and \cref{eq:empa1bj} for the fourth equality; the proof for \cref{eq:LA2} or \cref{eq:LA3} is similar.

Eqs.~\eqref{eq:LAij} mean that the joint distribution $P_{ij}(a_1,a_i,b_1,b_j)$ has the four pairwise empirical correlations $\wp(a_i,b_j)$, $\wp(a_1,b_j)$, $\wp(a_i,b_1)$ and $\wp(a_1,b_1)$ as its marginals, proving the first property.

Furthermore, the LF constraints due to the assumption of Local Agency in Eqs.~\eqref{eq:LA} directly imply Eqs.~\eqref{eq:agree} by \cref{eq:defPij}, meaning that the distributions in $\{P_{ij}(a_1, b_1, a_i, b_j)\}_{i \in \{2,\ldots,N_A\},j \in \{2,\ldots,N_B\}}$ agree on their overlaps, proving the second property

\vspace{0.5cm}

Now, we will prove the converse direction: that the existence of $\{P_{ij}(a_1, b_1, a_i, b_j)\}_{i \in \{2,\ldots,N_A\},j \in \{2,\ldots,N_B\}}$ with the above mentioned two properties implies all LF constraints.

For any $i \in \{2,\ldots,N_A\}$ and $j \in \{2,\ldots,N_B\}$, define 
\begin{subequations}
    \label{eq:AOEconverse}
\begin{align}
    &P(a_1,b_1|x=1,y=1)\coloneqq \sum_{a_2,b_2} P_{22}(a_1, a_2,b_1,b_2) \\
    &P(a_1, a_i,b_1|x=i,y=1)\coloneqq \sum_{b_2} P_{i2}(a_1, a_i,b_1,b_2) \\
    &P(a_1, b_1,b_j|x=1,y=j) \coloneqq  \sum_{a_2} P_{2j}(a_1, a_2,b_1,b_j)\\
    &P(a_1, a_i,b_1,b_j|x=i,y=j) \coloneqq P_{ij}(a_1, a_i,b_1,b_j)
\end{align}
\end{subequations}

These are the distributions required to exist by AOE in Eqs.~\eqref{eq:AOEdis}. They reproduce the empirical correlations as in Eqs.~\eqref{eq_emp} due to the first property of $\{P_{ij}(a_1, b_1, a_i, b_j)\}_{i \in \{2,\ldots,N_A\},j \in \{2,\ldots,N_B\}}$ as in Eqs.~\eqref{eq:LAij}. 

Furthermore, the LF conditions in \cref{eq:LA} are also satisfied due to the second property of $\{P_{ij}(a_1, b_1, a_i, b_j)\}_{i \in \{2,\ldots,N_A\},j \in \{2,\ldots,N_B\}}$ that these distributions agree on their overlaps. For example, to prove \cref{eq:LAaib1}, it suffices to note that for any $i \in \{2,\ldots,N_A\}$, $j \in \{2,\ldots,N_B\}$ and $j'\in\{1,2,\dots,N_B\}$, we have
\begin{align}
    &P(a_1,a_i,b_1 |x=i,y=j) \nonumber\\
    = &\sum_{b_j} P_{ij}(a_1,a_i,b_1,b_j) \nonumber\\
    = &\sum_{b_{j'}} P_{ij'}(a_1,a_i,b_1,b_j') \nonumber\\
    = &P(a_1,a_i,b_1|x=i,y=j'), 
\end{align}
where the second equality is due to the fact that $P_{ij}(a_1,a_i,b_1,b_j)$ and $P_{ij'}(a_1,a_i,b_1,b_j')$ agree on their overlap.

By expressing the LF constraints in terms of the existence of distributions in $\{P_{ij}(a_1, b_1, a_i, b_j)\}_{i \in \{2,\ldots,N_A\},j \in \{2,\ldots,N_B\}}$ and the two properties they have, we can easily obtain two insights into the features of the LF inequalities, as we now show in \cref{sec:1stLFasKSNC} and \cref{sec:Bell_inside_LF}.

\subsubsection{A class of tight LF inequalities that are KSNC constraints on subgraphs}
\label{sec:1stLFasKSNC}

The first insight that immediately follows from the alternative expression just derived is the existence of a class of LF inequalities that are KSNC constraints on subgraphs of the compatibility graph for a $N_A,N_B$-setting LF scenario, i.e.,  \cref{fig:NaBa}. Specifically, all the LF inequalities in this class are KSNC constraints on graphs of the form shown in \cref{fig:a1aib1bj}.

\begin{figure}[h!]
\includegraphics[width=0.15\textwidth]{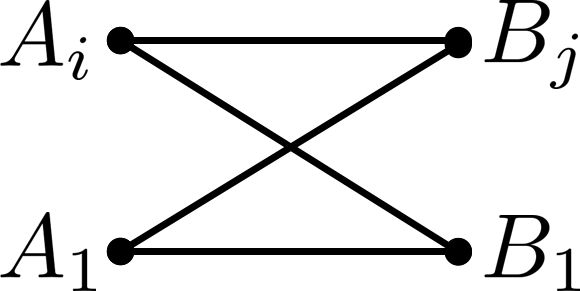}  
\caption{The compatibility graph for $A_1$, $A_i$, $B_1$ and $B_j$ with $i\in \{2,3,\dots,N_A\}$ and $j\in\{2,3,\dots,N_B\}$. It is a subgraph of \cref{fig:NaBa}.} 
\label{fig:a1aib1bj} 
\end{figure}

To see this, first recall from \cref{sec:recap_KSNC} that demanding KSNC constraints on a compatibility graph is equivalent to demanding the existence of a joint probability distribution over the outcomes of all measurements associated with vertices of the graph, such that the distribution has the empirical correlation associated with each edge as its marginals.
For the compatibility graph in \cref{fig:a1aib1bj}, from \cref{eq:LAij}, we see that $P_{ij}(a_1,a_i,b_1,b_j)$ is a joint distribution over the outcomes of the four measurements in the graph and has as its marginals the four empirical correlations for the four respective edges, namely, $\wp(a_i,b_j)$, $\wp(a_1,b_j)$, $\wp(a_i,b_1)$ and $\wp(a_1,b_1)$. 

Thus, the LF constraints require the KSNC constraints on graphs in the form of the one in \cref{fig:a1aib1bj}, which are subgraphs of \cref{fig:NaBa}, the full compatibility graph for a $N_A,N_B$-setting LF scenario. Furthermore, since \cref{fig:NaBa} corresponds to a bipartite Bell scenario where both Alice and Bob have binary measurement settings, the KSNC constraints on 
\cref{fig:NaBa} are CHSH inequalities. Therefore, for each $i\in \{2,3,\dots,N_A\}$ and $j\in\{2,3,\dots,N_B\}$, we can derive LF inequalities that are CHSH inequalities on $\wp(a_i,b_j)$, $\wp(a_1,b_j)$, $\wp(a_i,b_1)$ and $\wp(a_1,b_1)$. 

Furthermore, these LF inequalities are in fact \emph{tight}, meaning that they can be saturated by correlations that satisfy the LF assumptions.  This is because of the second immediate insight we gain from the alternative expression of the LF constraints, which, as we will explain next, is that the LF inequalities are not stricter than the Bell inequalities for the corresponding Bell scenario, or in other words, the LF polytope always contains the corresponding Bell polytope. This fact has been previously shown in Refs.~\cite{bong2020strong,cavalcanti2021implications,yile2023,haddara2024local} (and indirectly also in Ref.~\cite{woodhead2014}.)

\subsubsection{The LF polytope contains the Bell polytope} \label{sec:Bell_inside_LF}

Here we give an alternative proof of the fact the LF polytope contains the Bell polytope in the corresponding Bell scenario to the ones in previous literature such as Refs.~\cite{bong2020strong,yile2023}.

The Bell polytope in the Bell scenario corresponding to our LF scenario is precisely characterized by the KSNC constraints on the compatibility graph of the LF scenario as in \cref{sec:compatibility_graph}. That is, it comes from the requirement that there must exist a global distribution $P(a_1,b_1,a_2,b_2,\ldots,a_{N_A},b_{N_B})$ over outcome of all measurements that reproduces all empirical correlations, namely, all $\wp(a_x,a_y)$ with $ x\in\{1,2,\dots,N_A\}$ and $y\in\{1,2,\dots,N_A\}$.

When such a $P(a_1,b_1,a_2,b_2,\ldots,a_{N_A},b_{N_B})$ exists, it further implies the existence all distributions in $\{P_{ij}(a_1, b_1, a_i, b_j)\}_{i \in \{2,\ldots,N_A\},j \in \{2,\ldots,N_B\}}$ with the two properties described at the beginning of \cref{sec:equi}. Specifically, we can define $P_{ij}(a_1, b_1, a_i, b_j)$ simply by marginalising $P(a_1,b_1,a_2,b_2,\ldots,a_{N_A},b_{N_B})$. These distributions evidently reproduce the four empirical correlations $\wp(a_i,b_j)$, $\wp(a_1,b_j)$, $\wp(a_i,b_1)$ and $\wp(a_1,b_1)$; they also agree on their overlap by construction (as they all stem from the same global distribution). 

Thus, the KSNC constraints (which are also Bell inequalities) for the corresponding Bell scenario imply the respective LF constraints, which also means that the LF polytope always contains the corresponding Bell polytope.

\subsection{The 2,2-setting LF scenario}
\label{sec:22}
In this scenario, Alice and Bob have binary measurement settings. Thus, in total, we have four measurements $A_1$, $A_2$, $B_1$ and $B_2$ with measurement outcomes $a_1$, $a_2$, $b_1$ and $b_2$. The compatibility graph for the 2,2-setting LF scenario is shown in \cref{fig:22}.

\begin{figure}[h!]
\includegraphics[width=0.15\textwidth]{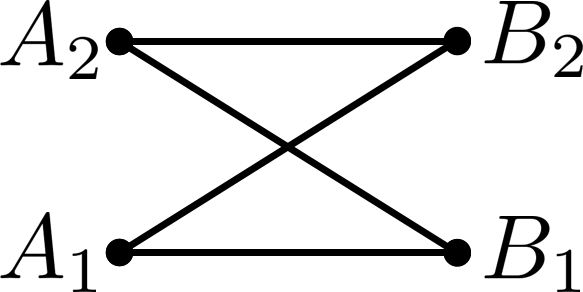}  
\caption{The compatibility graph for the 2,2-setting LF scenario, which is a special case of the graph in \cref{fig:NaBa}.} 
\label{fig:22} 
\end{figure}

From the alternative expression described in \cref{sec:22} of the LF constraints, we know that there exists a joint distribution $P_{22}(a_1,a_2,b_1,b_2)$ such that it has as its marginals $\wp(a_2 ,b_2)$, $\wp(a_1,b_2)$, $\wp(a_2,b_1)$ and $\wp(a_1,b_1)$, which are all the empirical correlations in the 2,2-setting LF scenario.
Thus, as discussed in \cref{sec:1stLFasKSNC}, the four pairwise correlations $\wp(a_x, b_y)$ for $x=1,2$ and $y=1,2$ must obey KSNC constraints on the compatibility graph in \cref{fig:22}, which give rise to the CHSH inequalities on these four correlations. Furthermore, since the LF polytope must contain the corresponding Bell polytope, as was proven in \Cref{sec:Bell_inside_LF}, in this 2,2-setting scenario, the LF polytope coincides with the Bell polytope whose nontrivial facets are characterized by these CHSH inequalities.

\subsection{The 3,2-setting LF scenario}
\label{sec:32}

Now consider the case where Alice has 3 settings and Bob has 2 settings. There are now six empirical pairwise correlations, namely $\wp(a_x, b_y)$ for $x=1,2,3$ and $y=1,2$.

Assuming Local Friendliness, following \cref{sec:equi}, we have that
\begin{itemize}
    \item $P_{22}(a_1,a_2,b_1,b_2)$ has as its marginals $\wp(a_2, b_2)$, $\wp(a_1,b_2)$, $\wp(a_2,b_1)$ and $\wp(a_1,b_1)$. Thus, the four pairwise correlations $\wp(a_x, b_y)$ for $x=1,2$ and $y=1,2$ must obey KSNC constraints on the compatibility graph in \cref{fig:a12b120}. 
    \item  $P_{32}(a_1,a_3,b_1,b_2)$ has as its marginals $\wp(a_3,b_2)$, $\wp(a_1,b_2)$, $\wp(a_3,b_1)$ and $\wp(a_1,b_1)$. Thus, the four pairwise correlations $\wp(a_x, b_y)$ for $x=1,3$ and $y=1,2$ must obey KSNC constraints on the compatibility graph in \cref{fig:a13b120}. 
\end{itemize}

\begin{figure}[h!]
    \centering
        \begin{subfigure}[b]{0.11\textwidth}
         \centering
         \includegraphics[width=\textwidth]{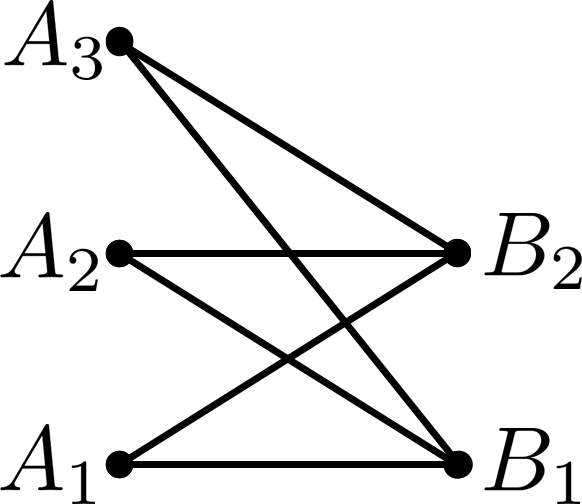}
         \caption{}
         \label{fig:a123b12m}
     \end{subfigure}
    \hspace{0.4mm}
    \begin{subfigure}[b]{0.11\textwidth}
         \centering
         \includegraphics[width=\textwidth]{figures/a12b12.png}
         \caption{}
    \label{fig:a12b120}
     \end{subfigure}
    \hspace{0.4mm}
    \begin{subfigure}[b]{0.11\textwidth}
         \centering
         \includegraphics[width=\textwidth]{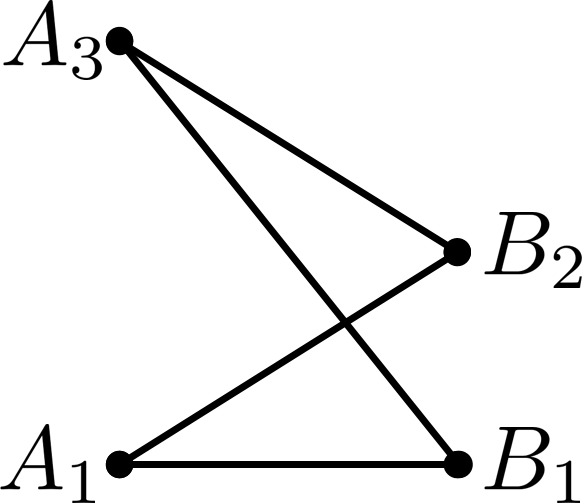}
         \caption{}
         \label{fig:a13b120}
     \end{subfigure}
\caption{(a): The compatibility graph of the 3,2-setting LF scenario. (b) and (c): Subgraphs of the graph in (a) that are special cases of the graph in \cref{fig:a1aib1bj}. In \cref{sec:glue} we explain how LF assumptions allow us to `glue' the distributions for (b) and (c) together to form a joint distribution for (a), using the fact that LF requires the distributions for (b) and (c) to agree on their overlap, i.e., have the same marginals on the intersection of their domain.}
\end{figure}

So far, we have \emph{not} shown that the four pairwise correlations $\wp(a_x,b_y)$ for $x=2,3$ and $y=1,2$ must satisfy KSNC constraints on a corresponding compatibility graph, and in fact one cannot do so by an argument analogous to the above. 

Nevertheless, we will now show that the Local Friendliness assumptions do require KSNC constraints for these correlations as well, and in fact, that the Local Friendliness assumptions imply KSNC constraints on the compatibility graph for the whole 3,2-LF scenario as shown in \cref{fig:a123b12m}, meaning that all six pairwise correlations $\wp(a_x,b_y)$ for $x=1,2,3$ and $y=1,2,3$ are marginals of a single joint probability distribution.

\subsubsection{The `gluing' trick}
\label{sec:glue}

The aforementioned joint distribution of which all six pairwise correlations $\wp(a_x,b_y)$ for $x=1,2,3$ and $y=1,2,3$ are marginals can be constructed using $P_{22}(a_1,a_2,b_1,b_2)$ and $P_{32}(a_1,a_3,b_1,b_2)$ via 
\begin{align}
\label{eq:2232}
    & P_{22,32} \\
    \coloneqq & \frac{P_{22}(a_1,a_2,b_1,b_2)P_{32}(a_1,a_3,b_1,b_2)}{\sum_{a_3} P_{32}(a_1,a_3,b_1,b_2)}.\nonumber
\end{align}
We call the trick used in \cref{eq:2232} for constructing $P_{22,32}$ the \emph{gluing} operation\footnote{A similar technique was used in, e.g., Refs.~\cite{fine1982joint,haddara2024local}.}, which constructs a larger joint probability distribution (i.e., $P_{22,32}$)  from smaller distributions (i.e., $P_{22}(a_1,a_2,b_1,b_2)$ and $P_{32}(a_1,a_3,b_1,b_2)$) in such a way that (as we will prove soon) the larger one has as its marginals all empirical pairwise correlations that are marginals of the smaller ones. We call it `gluing' since in terms of the compatibility graphs upon which the KSNC constraints are demanded by the respective joint distributions, such a procedure `glues' together two smaller graphs (i.e., \cref{fig:a12b120} and \cref{fig:a13b120}) to a bigger one (i.e, \cref{fig:a123b12m}); in terms of the probability distributions, it `glues' together two marginal distributions to create a bigger joint distribution.

 The proof for $P_{22,32}$ to recover all six pairwise correlations uses the fact that 
$P_{22}(a_1,a_2,b_1,b_2)$ and $P_{32}(a_1,a_3,b_1,b_2)$ agree on their overlap:
\begin{align}
\label{eq:bcd2232}
\sum_{a_2} P_{22}(a_1,a_2,b_1,b_2)
    =  \sum_{a_3} P_{32}(a_1,a_3,b_1,b_2), 
\end{align}
which is a special case of \cref{eq:agreea}.  See Appendix~\ref{app:glu} for the explicit proof.

As such, all six pairwise empirical correlations in the 3,2-setting LF scenario can be recovered as marginals of a single joint distribution $P_{22,32}$. Thus, these six pairwise correlations must satisfy the KSNC constraints on the compatibility graph in \cref{fig:a123b12m}. As this compatibility graph also represents a Bell scenario, it follows that these constraints give rise to Bell inequalities. Furthermore, these Bell inequalities are tight LF inequalities since the LF polytope must contain the Bell polytope as proven in \cref{sec:Bell_inside_LF}. Since these six pairwise correlations are all the observable correlations in the 3,2-setting LF scenario, then, the Local Friendliness polytope coincides with the Bell polytope in this case.

A trivial corollary of the above is that the four pairwise correlations $\wp(a_x,b_y)$ for $x=2,3$ and $y=1,2$ can be recovered as marginals of a single joint distribution, namely, $P_{22,32}$, and so must satisfy CHSH inequalities, and equivalently, the constraints implied by KSNC for any four projective measurements satisfying the compatibility relations in \cref{fig:a1aib1bj} (with the appropriate relabeling). 

Following an analogous proof to the one we just did here for the 3,2-setting LF scenario, one can further show that a $(N_A,2)$-setting or a $(2,N_B)$-setting LF scenario for any integer $N_A,N_B\geq 2$ has its LF constraints being equivalent to the KSNC constraints on the compatibility graph for that LF scenario, which give rise to the Bell inequalities for the corresponding Bell scenario. 

\subsubsection{The minimal LF scenario}
\label{sec:minimal}
An analogous argument can be constructed to explain why the LF polytope for the so-called `minimal LF scenario' studied in Ref.~\cite{wiseman2023thoughtful,yile2023} also coincides with the corresponding Bell polytope. 

In the minimal LF scenario, there is no Debbie. Alice has a binary measurement setting. She and Charlie proceed as usual. Bob, however, now chooses between two possible measurements $B_y$ on the other half of the bipartite system depending on his binary setting choice $y$. The empirical correlations in this scenario are the four pairwise correlations: $\wp(a_2, b_2)$, $\wp(a_1,b_2)$, $\wp(a_2,b_1)$ and $\wp(a_1,b_1)$.

Here, AOE demands the existence of $P(a_1,b_1|x=1,y=1)$, $P(a_1a_2b_1|x=2,y=1)$ and $P(a_1a_2b_2|x=2,y=2)$ for the corresponding runs of the experiment.  As shown in Appendix~\ref{app:glu}, by applying the gluing trick to $P(a_1a_2b_1|x=2,y=1)$ and $P(a_1a_2b_2|x=2,y=2)$, we obtain a joint probability distributions that recovers all four pairwise correlations, and thus explains why the LF polytope coincides with the corresponding Bell polytope in the minimal scenario.

Furthermore, one can generalize the minimal LF scenario and show that if Bob's measurement choice is not binary but instead with higher cardinality for any integer $N_B\geq 2$, it is still true that all LF constraints in this scenario are equivalent to the KSNC constraints (or equivalently the Bell inequalities) on a corresponding Bell scenario. 

\subsection{The 3,3-setting LF scenario}
\label{sec:33}
Now consider the case where both Alice and Bob have 3 settings. In this case, the Local Friendliness polytope does {\em not} coincide with the Bell polytope, as was shown in Ref.~\cite{bong2020strong}, which used a software to derive the LF polytope for binary outcomes explicitly. We will give a more algebraic derivation and explanation of this fact (and without the restriction to binary outcomes).

Assuming Local Friendliness, following \cref{sec:equi}, we have that $P_{22}(a_1,a_2,b_1,b_2)$ has as its marginals $\wp(a_2,b_2)$, $\wp(a_1,b_2)$, $\wp(a_2,b_1)$ and $\wp(a_1,b_1)$. Thus, the four pairwise correlations $\wp(a_x,b_y)$ for $x=1,2$ and $y=1,2$ must obey the KSNC constraints on the compatibility relations in \cref{fig:a12b12}. Since this compatibility graph in \cref{fig:a12b12} corresponds to a Bell scenario, these KSNC constraints give rise to Bell inequalities. Similar reasoning holds for $P_{32}(a_1,a_3,b_1,b_2)$, $P_{23}(a_1,a_2,b_1,b_3)$ and $P_{33}(a_1,a_3,b_1,b_3)$, depicted in \Cref{fig:a13b12,fig:a12b13,fig:a13b13}.

\begin{figure}[!h]
\captionsetup[subfigure]{aboveskip=-2pt,belowskip=-1pt}
    \centering
    \begin{subfigure}[b]{0.11\textwidth}
         \centering
         \includegraphics[width=\textwidth]{figures/a12b12.png}
         \caption{}
    \label{fig:a12b12}
     \end{subfigure}
    \hspace{0.4mm}
    \begin{subfigure}[b]{0.11\textwidth}
         \centering
         \includegraphics[width=\textwidth]{figures/a13b12.png}
         \caption{}
         \label{fig:a13b12}
     \end{subfigure}
     \hspace{0.4mm}
    \begin{subfigure}[b]{0.11\textwidth}
         \centering
         \includegraphics[width=\textwidth]{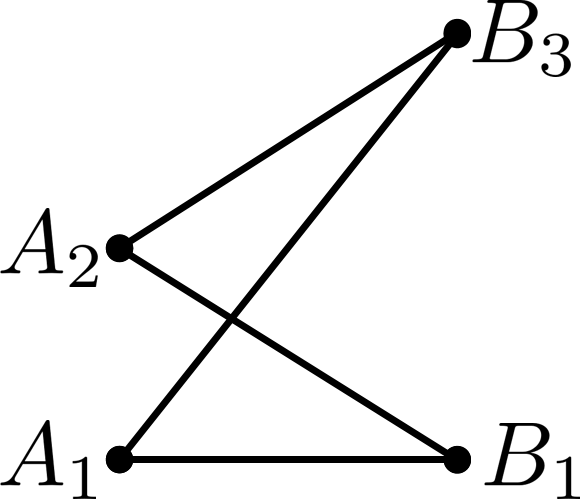}
         \caption{}
         \label{fig:a12b13}
     \end{subfigure}
    \hspace{0.4mm}
    \begin{subfigure}[b]{0.11\textwidth}
         \centering
         \includegraphics[width=\textwidth]{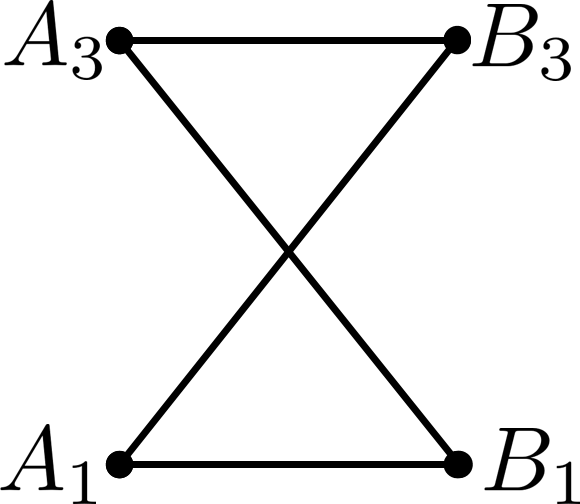}
         \caption{}
         \label{fig:a13b13}
     \end{subfigure}
    \caption{Special cases of the compatibility graph in \cref{fig:NaBa} that are subgraphs of the compatibility graph of the 3,3-setting LF scenario. }
    \label{fig:4pair}
\end{figure}

Furthermore, as proven in \cref{sec:equi}, these four distributions, $P_{22},P_{23},P_{32},P_{33}$ (where we omit their arguments for compactness), agree on their overlaps. 
Then, by gluing various three out of the four distributions together, following an argument like that for \cref{eq:2232}, we obtain the following. Gluing together $P_{22},P_{32},P_{23}$ we find that $P_{22,32,23} \! \coloneqq \! 
    \frac{P_{22}P_{32}P_{23}}{\sum_{a_3} P_{32} \sum_{b_3} P_{23}}$ has as its marginals all pairwise correlations $\wp(a_x,b_y)$ for $x=1,2,3$ and $y=1,2,3$ \emph{except} for $\wp(a_3,b_3)$, and thus, we also denote $P_{22,32,23}$ as $P_{\lnot a_3b_3}$. As such, these pairwise correlations must obey the KSNC constraints on the compatibility graph in \cref{fig:noa3b3}. The same reasoning holds for $P_{\lnot a_2b_2}, P_{\lnot a_2b_3}, P_{\lnot a_3b_2}$ with compatibility graphs in \Cref{fig:noa2b2,fig:noa2b3,fig:noa3b2}.

\begin{figure}[!h]
\captionsetup[subfigure]{aboveskip=-2pt,belowskip=-1pt}
    \centering
    \begin{subfigure}[b]{0.11\textwidth}
         \centering
         \includegraphics[width=\textwidth]{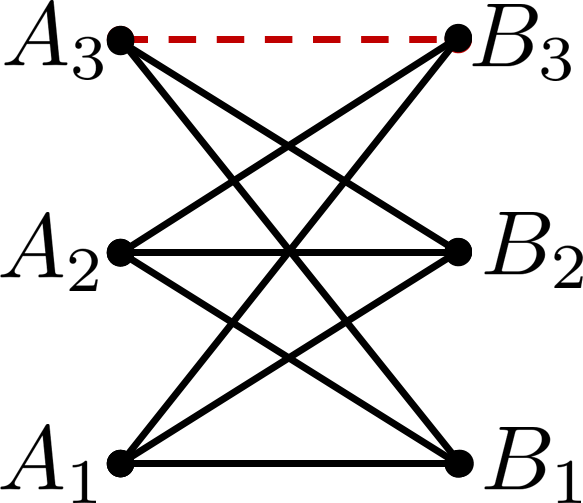}
         \caption{}
    \label{fig:noa3b3}
     \end{subfigure}
    \hspace{0.4mm}
    \begin{subfigure}[b]{0.11\textwidth}
         \centering
         \includegraphics[width=\textwidth]{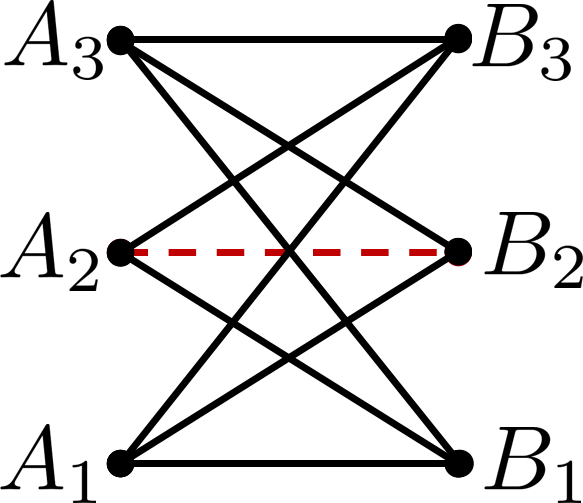}
         \caption{}
         \label{fig:noa2b2}
     \end{subfigure}
     \hspace{0.4mm}
    \begin{subfigure}[b]{0.11\textwidth}
         \centering
         \includegraphics[width=\textwidth]{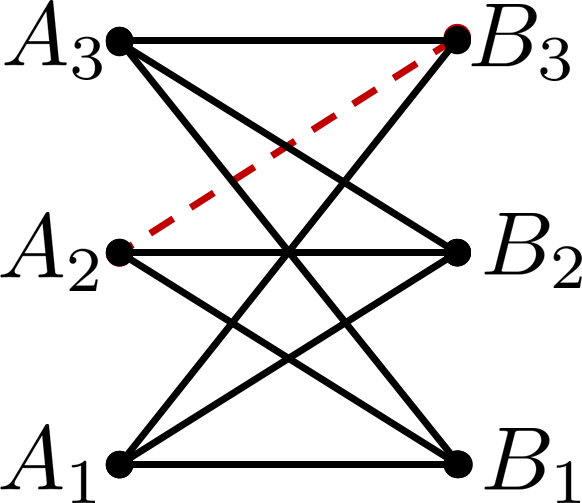}
         \caption{}
         \label{fig:noa2b3}
     \end{subfigure}
    \hspace{0.4mm}
    \begin{subfigure}[b]{0.11\textwidth}
         \centering
         \includegraphics[width=\textwidth]{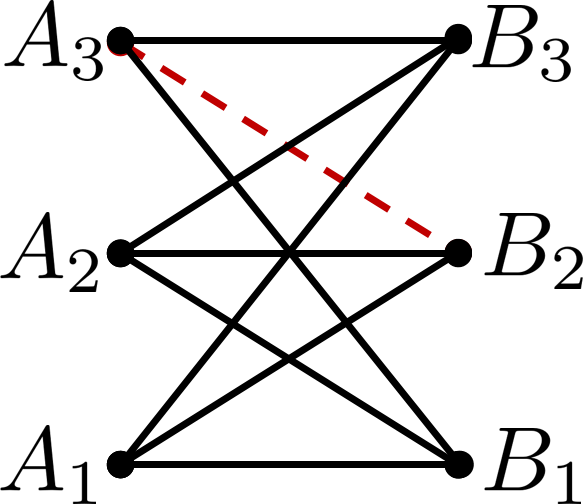}
         \caption{}
         \label{fig:noa3b2}
     \end{subfigure}
    \caption{Subgraphs of the compatibility graph of the 3,3-setting LF scenario. Each of them has 8 edges, i.e., 8 pairwise compatibility relations. The red dashed line indicates the missing compatibility relation for deriving Bell inequalities on all six measurements in each graph.}
    \label{fig:8pair}
\end{figure}

Unlike all of the previous cases, these compatibility graphs in \cref{fig:8pair} are {\em not} the complete compatibility graphs of any Bell scenario, because of the missing edges indicated by the red dashed lines in \cref{fig:8pair}; rather, they are subgraphs. 
Consequently,  the inequalities derived from KSNC constraints in this case will not generally be tight Bell inequalities. 
Nevertheless, some of these KSNC inequalities coincide with tight Bell inequalities because the compatibility graph in \cref{fig:8pair} includes subgraphs that are the complete compatibility graphs for some other Bell scenarios. 
For example, \cref{fig:noa3b3} includes the graphs in \cref{fig:a12b12,fig:a13b12,fig:a12b13} and \cref{fig:a123b12m} as subgraphs, and they are complete compatibility graphs for some Bell scenarios; in addition, it includes a compatibility subgraph for $\{A_2,A_3,B_1,B_3\}$, one for $\{A_1,A_3,B_2,B_3\}$, one for $\{A_2,A_3,B_1,B_1\}$, one for $\{A_1,A_2,B_2,B_3\}$ and one for $\{A_2,A_3,B_1,B_2,B_3\}$, of which each represents a complete compatibility graph for a Bell scenario. 
Thus, the KSNC constraints on the compatibility graphs in  \cref{fig:8pair} include Bell inequalities for these compatibility subgraphs.

The compatibility graphs in \cref{fig:8pair} are also \emph{not} the compatibility graph of the 3,3-setting LF scenario because of the missing edges depicted as red dashed lines. In fact, the Local Friendliness assumptions do not imply the existence of a joint probability distribution that can have as its marginals all the 9 pairwise empirical correlations $\wp(a_x,b_y)$ for $x=1,2,3$ and $y=1,2,3$ in the 3,3-setting LF scenario, and hence, the LF polytope here is strictly bigger than the Bell polytope, as pointed out in Ref.~\cite{bong2020strong}.  In Appendix~\ref{app:glu}, we provide further intuitions on why this is the case by considering how the gluing trick introduced in \cref{sec:glue} works.

Since any $(N_A,N_B)$-setting LF scenario with $N_A,N_B\geq 3$ has the 3,3-setting LF scenario as a subscenario, we can further see that a $(N_A,N_B)$-setting LF scenario for any integer $N_A,N_B\geq 3$ will have its LF constraints being strictly weaker than the KSNC constraints on the compatibility graph for that LF scenario, meaning that the LF polytope is strictly bigger than the corresponding Bell polytope. 

Note that the KSNC constraints on the four graphs in \cref{fig:8pair} do not have all constraints given by Local Agency. In particular, Local Agency further demands certain constraints on the relationship between  $P_{\lnot a_3b_3}$, $P_{\lnot a_2b_2}$, $P_{\lnot a_2b_3}$ and $P_{\lnot a_3b_2}$, due to the fact that any two of the building blocks of the gluing trick (namely, $P_{22}$, $P_{23}$, $P_{32}$ and $P_{33}$) must agree on their overlap.
Thus, the inequalities we have derived here are not necessarily tight.
This can also be seen by comparing to the results in Ref.~\cite{bong2020strong} as we explain more in Appendix~\ref{app:connect}.

Although the inequalities we have just derived are not necessarily tight (such as some of the ones in \cref{sec:33}), and so do not fully characterize the set of correlations satisfying LF assumptions, they are nonetheless tight enough to admit violations within unitary quantum theory. Every LF inequality (whether it is tight or not) is also a valid Bell inequality, as a consequence of the fact that the LF polytope always contains the Bell polytope for the same setting and outcome variable cardinalities. 
If such a Bell inequality can be violated by quantum theory in a Bell scenario, then it, as a LF inequality, can also be violated by quantum theory.
This is because one can translate any quantum proposal for a $N_A,N_B$-setting Bell scenario violating a Bell inequality to a quantum proposal for the $N_A,N_B$-setting Local Friendliness scenario that violates the corresponding LF inequality. We show this in detail in Appendix~\ref{app:violate}. 

\vspace{0.5cm}

\subsection{Generalizations to multipartite sequential LF scenarios}
\label{sec:generalisation_LF_Bell}

So far, we have considered bipartite LF scenarios and characterized the relationship between the LF polytope in these scenarios and their corresponding KSNC or Bell polytopes. This analysis can be extended to more general LF scenarios involving multiple superobserver–observer pairs performing measurements sequentially (with arbitrarily long sequences). Furthermore, we can characterize the conditions under which all LF constraints in such a scenario can also be interpreted as KSNC constraints, and in particular, as Bell inequalities. We refer interested readers to Appendix~\ref{app:generalisation_LF_Bell} for more details.

\section{Constructing new Local Friendliness scenarios from KS noncontextuality arguments}
\label{sec:KSNC_to_LF}

In the first half of this paper, we focused on showing how, given a specific LF scenario, one can derive many LF inequalities using known tools and results from the literature on KSNC, and on clarifying the relation between Bell local causality and LF. In the rest of the paper, we show how given a possibilistic proof of the failure of KSNC, one can construct an extended Wigner's friend argument---a no-go theorem against AOE together with either Local Agency (in this section) or Commutation Irrelevance~\cite{walleghem2023extended} (in Appendix~\ref{app:CI}).

\subsection{The 5-cycle example} \label{sec:5-cycle}
 In Ref.~\cite{walleghem2023extended}, it was shown that the 5-cycle noncontextuality scenario~\cite{cabello2013simple,klyachko2008simple,santos2021conditions} can be used to construct an extended Wigner's friend argument with the assumption of Commutation Irrelevance, which is not one of the LF assumptions. Here, we show that the protocol in Ref.~\cite{walleghem2023extended} can in fact be adapted to construct a novel scenario where one can prove the LF no-go theorem, that is, a no-go theorem based on the LF assumptions without using Commutation Irrelevance. 

\subsubsection{The 5-cycle noncontextuality argument} \label{sec:5-cycle_recap}

We first recap the 5-cycle noncontextuality argument from Ref.~\cite{cabello2013simple,klyachko2008simple,santos2021conditions}, following the presentation in Ref.~\cite{walleghem2023extended}. 
Consider five binary-outcome measurements $\{A_i\}_{i\in\{1,2,...,5\}}$,  forming the compatibility graph in \cref{fig:NC_5cycle}. 
Imagine that the system is prepared such that the observations for these five joint measurements satisfy
\begin{subequations} \label{eq:5_cycle_12+23_34+45}
\begin{align}
 P(1,1|1,2)=0, \label{eq:5_cycle_12} \\
 P(0,0|2,3)=0, \label{eq:5_cycle_23} \\
    P(1,1|3,4) =0, \label{eq:5_cycle_34} \\
    P(0,0|4,5)=0, \label{eq:5_cycle_45}
    \end{align}
\end{subequations}
and
\begin{align}   
\label{eq:01_box_not0}
            P(0,1|5,1) \neq 0.
    \end{align}

Denoting the deterministic assignment of the outcome of measurement $A_i$ as $a_i$, Eq.~\eqref{eq:5_cycle_12+23_34+45} implies that 
\begin{equation}
\label{eq:mcontr}
   a_1=1 \Rightarrow a_2=0 \Rightarrow a_3=1 \Rightarrow a_4=0 \Rightarrow a_5=1.
\end{equation}
Therefore, in every run of the experiment where the outcome of $A_1$ is $1$, the outcome of $A_5$ must be $0$, contradicting \cref{eq:01_box_not0}. Thus, there is no Kochen-Specker noncontextual assignment for these runs. 

The set of correlations in Eq.~\eqref{eq:5_cycle_12+23_34+45} and Eq.~\eqref{eq:01_box_not0} has a quantum realization with a qutrit prepared in the state \begin{equation}\label{eq:eta}
    |\eta \rangle := \sqrt{\frac{1}{3}} (|0\rangle+\vert 1\rangle+\vert 2\rangle)\equiv \sqrt{\frac{1}{3}} (1,1,1)^T,
\end{equation} where $T$ denotes transposition, and measurements 
\begin{equation}
\label{eq:M}
A_i:=\{\vert v_i\rangle \langle v_i \vert ,\mathbf{1}-\vert v_i\rangle \langle v_i \vert \},
\end{equation}
with the states $|v_i \rangle$ defined as \begin{align}\label{eq: KCBS measurements}
       & |v_1 \rangle = \sqrt{\frac{1}{3}} (1,-1,1)^T, \,\, 
        |v_2 \rangle = \sqrt{\frac{1}{2}} (1,1,0)^T, \\
      &   |v_3 \rangle = (0,0,1)^T, 
        \,\, |v_4 \rangle = (1,0,0)^T, \,\,
        |v_5 \rangle = \sqrt{\frac{1}{2}} (0,1,1)^T. \nonumber
\end{align} 
We take outcomes $0,1$ of a given measurement to correspond to $\mathbf{1}-\vert v_i\rangle \langle v_i \vert, \vert v_i \rangle \langle v_i \vert $, respectively. In this case, Eqs.~\eqref{eq:5_cycle_12}-\eqref{eq:5_cycle_45} and Eq.~\eqref{eq:01_box_not0} are satisfied. Hence, quantum theory is contextual.

\begin{flushleft}  
\begin{figure}[]
\includestandalone[width=0.2\textwidth]{figures/5-cycle_LF_vs_noncontextuality}  
\caption{The compatibility graph of the 5-cycle scenario. The pairs $\{(1,2),(2,3),(3,4),$ $(4,5),(5, 1 )\}$ are jointly measurable.
} \label{fig:NC_5cycle} 
\end{figure}
\end{flushleft}

\subsubsection{The LF 5-cycle protocol}
\label{sec:5cycleprotocol}

\begin{figure}[htb!]
\centering
\includegraphics[width=0.35\textwidth]{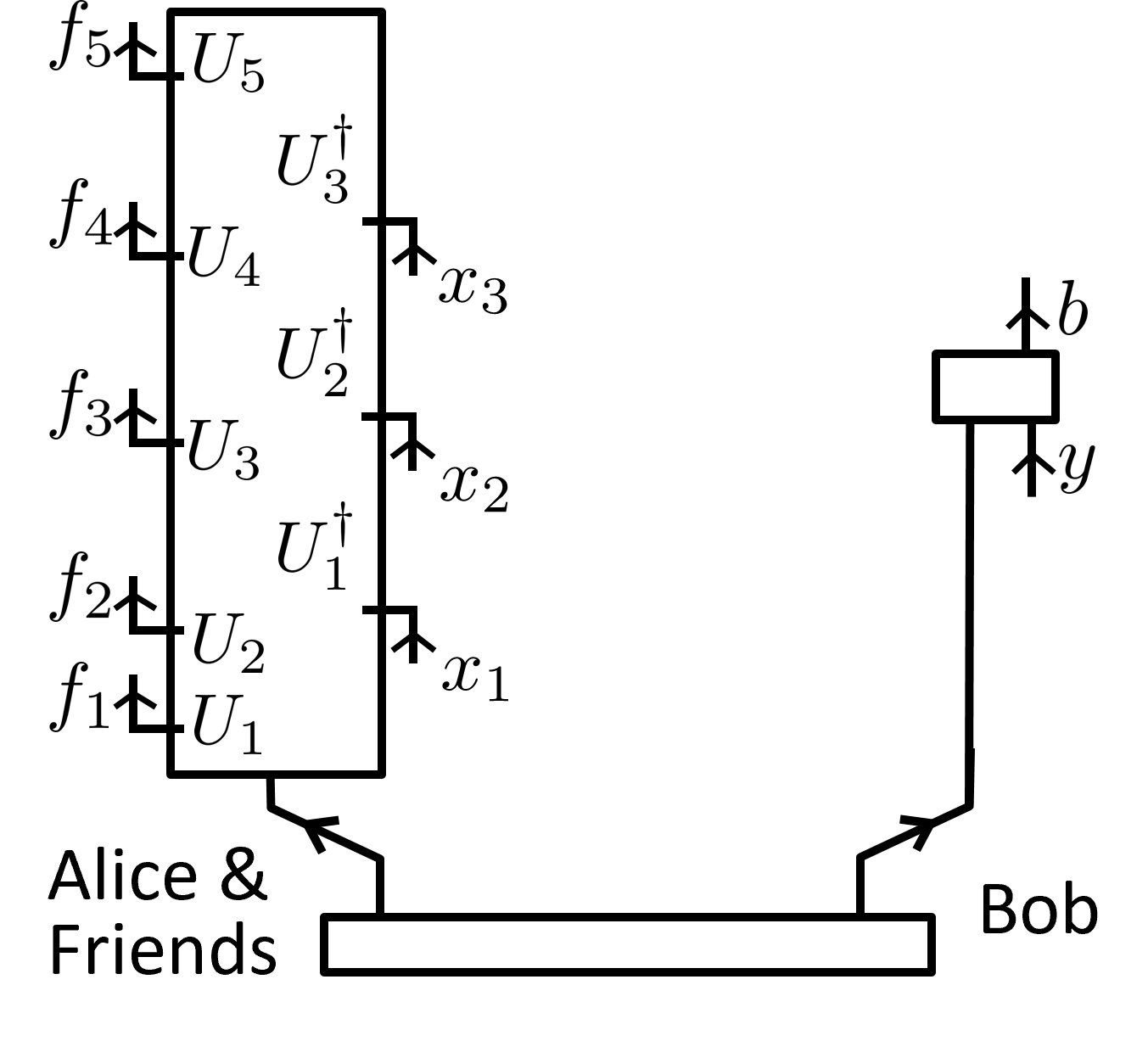}
\caption{The schematic representation of the LF 5-cycle scenario, see also \Cref{fig_cartoon5cycle}.}
\label{fig_setusp5cycle}
\end{figure}

\begin{figure}[htb!]
\centering
\includegraphics[width=0.45\textwidth]{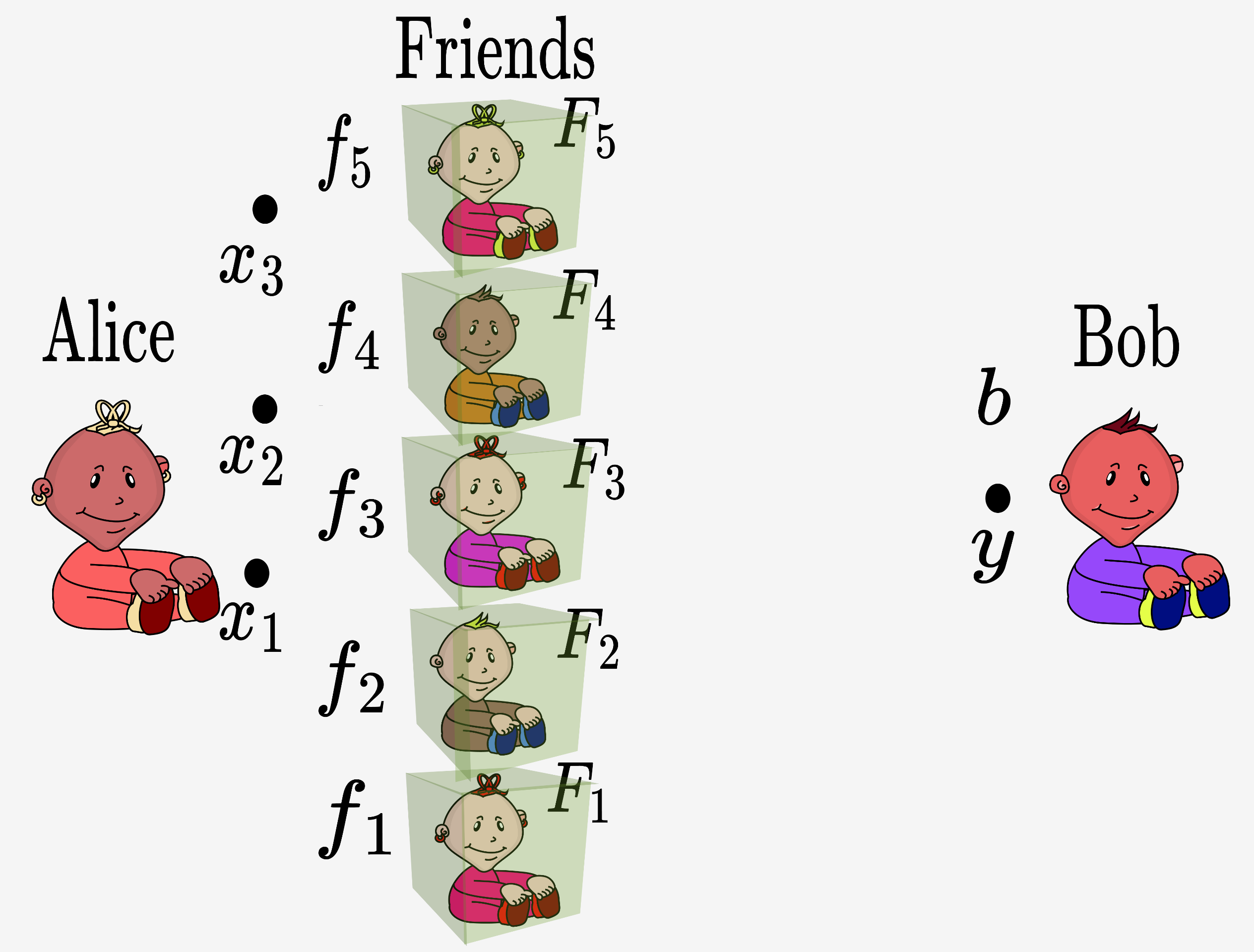}
\caption{Cartoon representation of the 5-cycle LF setup, with Alice a superobserver for $F_1,F_2,F_3,F_4
,F_5$, who perform measurements. Alice makes choices $x_1,x_2,x_3$ whether to reveal measurement outcomes and Bob, spacelike separated, makes measurement choice $y$ and obtains outcome $b$.} 
\label{fig_cartoon5cycle}
\end{figure}

We now leverage this 5-cycle KSNC argument to construct a novel scenario in which one can prove the LF no-go theorem.  Our presentation of this LF 5-cycle protocol highlights the differences from the protocol in Ref.~\cite{walleghem2023extended}.

This new LF scenario involves two groups of agents. The first group consists of Alice (denoted $A$) and five friends (denoted $F_1$, ..., $F_5$). The second group consists of a single agent Bob.  See \Cref{fig_cartoon5cycle} for a cartoon representation of the set-up.  Compared to the protocol in Ref.~\cite{walleghem2023extended}, we have one more agent here, namely, Bob; furthermore, the superobserver Alice here has measurement choices. These additions are crucial for deriving a no-go theorem based on LF assumptions without using the assumption of Commutation Irrelevance as in \cite{walleghem2023extended}.

In the following, we will describe the quantum protocol that will be used to prove the no-go theorem. 
 However, the scenario could also be described in a theory-independent way as discussed in Appendix~\ref{app:compare}.  

At the beginning of the experiment, $F_1$ and Bob share a bipartite system $ST$ where $S$ is a qutrit and $T$ is a qubit. The joint system $ST$ is prepared in the following entangled state:\footnote{\label{ft:mp}One way to construct this entangled state \eqref{eq_5cycleini} is to let $F_1$ perform measurement $A_1$ on $\ket{\eta}_S$ and prepare $\ket{0}_T$ if her measurement outcome corresponds to $\ket{v_1^{\perp}}$ or $\ket{1}_T$ if her measurement outcome corresponds to $\ket{v_1}$, and then $F_1$ sends system $T$ to Bob. In fact, instead of starting with $F_1$ and $B$ sharing the entangled state \eqref{eq_5cycleini}, the protocol could start with $F_1$ having system $S$ prepared in $\ket{\eta}_S$, followed by the above measure-prepare procedure. All subsequent operations by Alice, the friends $F_i$ and Bob remain the same. In this alternative protocol, $F_1$'s operation is in the past-light-cone of everyone else's operations, while Bob's operations are still space-like separated from all operations on Alice's side except for $F_1$'s. This is similar to how the measure-prepare version of the Frauchiger-Renner argument~\cite{frauchiger2018quantum} as it was originally presented, is equivalent to the simpler entanglement-based scenario (see, e.g., \cite{vilasini2022general,schmid2023review}).
}
\begin{align}
\label{eq_5cycleini}
    \ket{\Psi}_{ ST}
    = &\langle{v_1^{\perp}}|{\eta}\rangle\ket{v_1^{\perp}}_{ S}\otimes\ket{0}_{ T}+\langle{v_1}|{\eta}\rangle\ket{v_1}_{ S}\otimes\ket{1}_{ T} 
\end{align}
where $\ketbra{\textstyle v_1^{\perp}}:=\mathbf{1}-\ketbra{v_1}$. System $S$ is sent to Alice's side while system $T$ is sent to Bob. 

On Alice's side, the five friends, $F_1$, ..., $F_5$, perform the five measurements $A_i$ defined in \cref{eq:M} sequentially on the system $S$, with the superobserver's reversal of their measurements in between, as shown in Fig.~\ref{fig_setusp5cycle}. From AOE, each friend $F_i$ observes an absolute outcome, denoted $f_i$, during their measurement $A_i$. We label the $\ketbra{v_i}$ outcome by $f_i=1$ and the other outcome by $f_i=0$. In addition, the superobserver Alice has a binary choice $x_i$ right before she implements the inverse unitary $U^{\dagger}_i$ for $i=1,2,3$. If $x_i=1$, Alice opens the door to $F_i$'s and $F_{i+1}$'s labs and reveals their respective measurement outcome $f_i$ and $f_{i+1}$ to the world; if $x_i=2$, she simply continues the protocol.\footnote{Note that in the $x_i=1$ case, all records of the outcome $f_i$ are erased by the inverse unitaries applied by Alice; in contrast, in the $x_i=1$ case, the act of revealing the outcome $f_i$ to the world implies making and distributing many classical copies of it. In that case, the inverse unitaries performed by Alice do not have the effect of reverting the state of $S$ and the Friends to their initial states. They moreover will strongly impact outcomes observed by friends later in the protocol. However, all this is irrelevant, because in those cases the later outcomes of the friends are not used in the proof of the LF no-go theorem. Indeed, one could avoid these subtleties simply by stipulating that when $x_i=1$, Alice not only reveals the two immediately prior measurement outcomes to the world but also terminates the protocol.} When $x_1=x_2=x_3=2$, the $f_4$ and $f_5$ outcomes are revealed to the world at the end of the experiment.  For the full description of the protocol on Alice side, see Appendix~\ref{app:5Alice}. 

On Bob's side, Bob has a binary choice $y$. If $y=1$, he makes a measurement $B_{0/1}$ on system $T$, which is a computational basis measurement;
if $y=2$, he makes a measurement $B_{\pm}$ on system $T$, which is a $\pm$ basis measurement, i.e, $\{ \ketbra{+}, \ketbra{-}  \}$.
We denote Bob's absolute outcome by $b$.

We assume that all operations on Alice's side are done at space-like separation from Bob's measurement setting and outcome.

\begin{flushleft}  
\begin{figure}[]
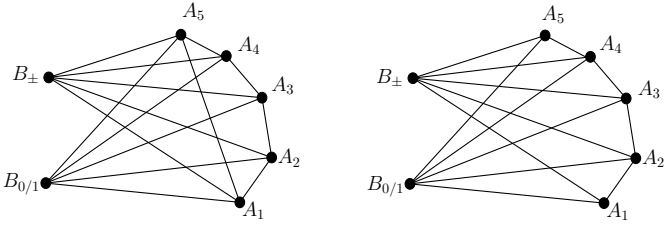

\includestandalone[width=0.5\textwidth]{figures/5-cycle_comp_graph_LF}  \caption{
Left: A compatibility graph for the measurements in the LF 5-cycle protocol of \cref{sec:5cycleprotocol}. Right: The 
 constraints from LF assumptions do not impose KSNC relative to the compatibility graph on the left, but they do effectively impose the same constraints as KSNC relative to the subgraph obtained when the compatibility relation between $A_1$ and $A_5$ is absent, as shown on the right. (They may impose further constraints in addition to these.)} \label{fig:LF_5-cycle_graphs} 
\end{figure}
\end{flushleft}

One can naturally associate a compatibility graph with this 5-cycle LF scenario: namely the one shown in the left subfigure of \cref{fig:LF_5-cycle_graphs}. The Local Friendliness assumptions on this LF 5-cycle scenario \emph{do not} imply KSNC on this compatibility graph, nor on the compatibility graph for the 5-cycle scenario (shown in \cref{fig:NC_5cycle}), since the joint distribution over $A_1$ and $A_5$ measurement outcomes is not observable.  Rather, the LF assumptions imply KSNC constraints on the graph in the right figure of \cref{fig:LF_5-cycle_graphs} (where the edge between $A_1$ and $A_5$ is absent). Nevertheless, in \cref{sec:5cycle_proof} we show that the possibilistic proof of the failure of KSNC in the 5-cycle scenario can still be mathematically turned into a proof of Local Friendliness no-go theorem. 

\subsubsection{The no-go theorem}
\label{sec:5cycle_proof}

\begin{figure}[htb!]
\centering
\includegraphics[width=0.4\textwidth]{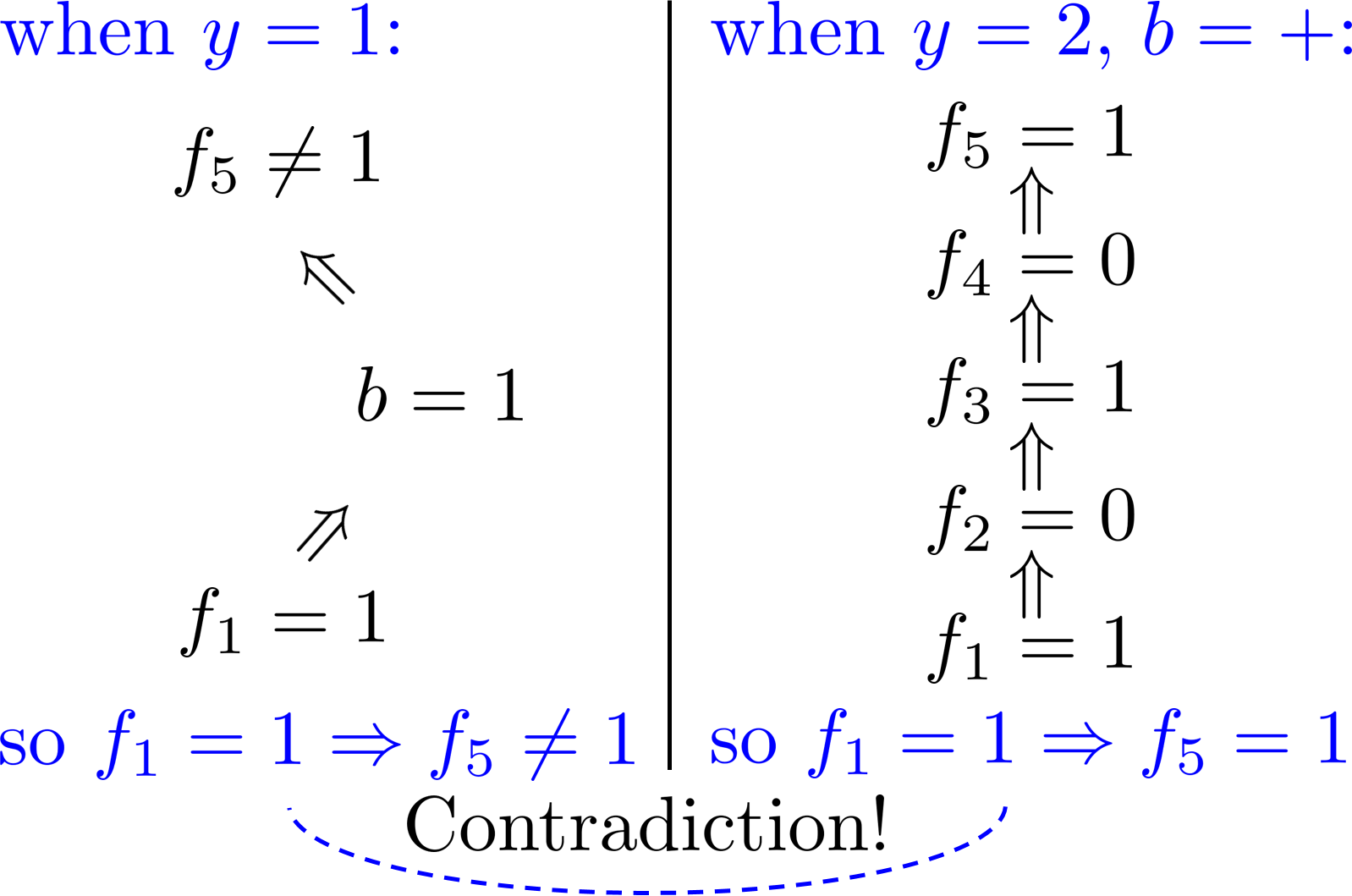}
\caption{The schematic proof using the LF 5-cycle protocol for the LF no-go theorem.  The key steps in the proof are as follows. When Bob measures his qubit in the computational basis (i.e., when $y=1$, as in the left side of the figure),  his outcome $b$ must be the same as the $f_1$ outcome, which in turn constrains the correlation over $f_1$ and $f_5$. When Bob measures his qubit in the $\pm$ basis (i.e., when $y=2$) and obtains outcome $b=+$ (as in the right side of the figure), the effective state on Alice and her friends' side becomes $\ket{\eta}_{S}$, i.e., the state needed in the 5-cycle proof of KSNC theorem, and one obtains a contradictory constraint on the correlation between $f_1$ and $f_5$. 
}
\label{fig_summary5cycle}
\end{figure}

A schematic of the proof is given in Figure~\ref{fig_summary5cycle}. The proof begins by noting that AOE demands the existence of the conditional probability distribution
\begin{align}
\label{eq:aoe5cycle}
P(f_1,f_2,f_3,f_4,f_5,b|x_1,x_2,x_3,y),
\end{align}
exists (although of course, it is not accessible to any agent). Next, one considers the predictions made by quantum theory for various \emph{empirical} correlations in the LF 5-cycle protocol, and then applies Local Agency to them to obtain that
\begin{subequations} 
\begin{align}
& P(f_1=1,f_2=1|  x_{1-3}=2, y=2,b=+) = 0, \label{eq:ma12}\\
&P(f_2=0,f_3=0| x_{1-3}=2, y=2,b=+) = 0, \label{eq:ma23}\\
&P (f_3=1,f_4=1|  x_{1-3}=2, y=2,b=+) = 0, \label{eq:ma34} \\
&P (f_4=0,f_5=0| x_{1-3}=2, y=2,b=+)=0. \label{eq:ma45}
\end{align}
\end{subequations} 
and 
\begin{align}
\label{eq:ma15}
P(f_1=1,f_5=1|x_{1-3}=2,y=2,b=+)=0.
\end{align}
These five equations constitute constraints on five marginals of the joint distribution 
\begin{align}
    P(f_1,f_2,f_3,f_4,f_5|x_{1-3}=2,y=2,b=+),
\end{align}
which is well-defined due to the existence of \cref{eq:aoe5cycle} and the Born rule prediction that 
\begin{align}\label{eq:b+}
    \wp (b=+|x_{1-3}=2,y=2)>0
\end{align}
(a prediction that is empirically accessible).

However, these five constraints, i.e., \cref{eq:ma12,eq:ma23,eq:ma34,eq:ma45,eq:ma15},  cannot be simultaneously satisfied. This fact is the basis of the KSNC proof in the 5-cycle scenario; one can see it most easily by noting that these equations imply (respectively) that when $x_{1-3}=2, y=2,b=+$, we have 
\begin{align}
\label{eq:imp5}
    f_1=1\Rightarrow f_2=0 \Rightarrow f_3=1 \Rightarrow f_4=0 \Rightarrow f_5=1.
\end{align} 
and
\begin{align}
    f_1=1 \Rightarrow f_5\neq 1,
\end{align}
which are in contradiction. Therefore there can be no joint distribution of the form in Eq.~\eqref{eq:ma15} consistent with the empirical predictions given by the Born rule and Local Agency. Hence the no-go theorem.

It remains only to establish Eqs.~\eqref{eq:ma12}-\eqref{eq:ma15}. 

We begin with \cref{eq:ma12,eq:ma23,eq:ma34,eq:ma45}.
Consider the following empirical predictions given by the Born rule for the cases where Bob measures in $\pm$ basis and obtains outcome $+$:
\begin{subequations} 
\begin{align}
&\wp (f_1=1,f_2=1| x_{1}=1,x_2,x_3, y=2,b=+) = 0, \label{eq:em12}\\
&\wp(f_2=0,f_3=0| x_{1}=2,x_2=1,x_3, y=2,b=+) = 0, \label{eq:em23}\\
&\wp (f_3=1,f_4=1| x_{1-2}=2,x_3=1, y=2,b=+) = 0, \label{eq:em34} \\
&\wp(f_4=0,f_5=0| x_{1-3}=2, y=2,b=+)=0. \label{eq:em45}
\end{align}
\end{subequations} 
These are analogous to \cref{eq:5_cycle_12,eq:5_cycle_23,eq:5_cycle_34,eq:5_cycle_45}; indeed, a quick way to see that these are the Born rule predictions in the present scenario is to realize that when $y=2$ and $b=+$, the state $S$ on Alice' side is effectively $\ket{\eta}$ defined in \cref{eq:eta}.

For \cref{eq:em12}, since both $f_1$ and $f_2$ are in the past light cone of all Alice's choices, i.e., they are not in the future light cone of $x_1,x_2,x_3$, according to Local Agency, we have
\begin{align}
   & \wp (f_1,f_2| x_{1}=1,x_2,x_3, y=2,b=+) \nonumber \\
   = & P(f_1,f_2| x_{1-3}=2,y=2,b=+) 
\end{align}

Similarly, for \cref{eq:em23,eq:em34}, Local Agency demands that 
\begin{align}
    & \wp (f_2,f_3| x_{1}=2,x_2=1,x_3, y=2,b=+) \nonumber \\
   = & P(f_2,f_3| x_{1-3}=2,y=2,b=+),\\
    & \wp (f_3,f_4| x_{1-2}=2,x_3=1, y=2,b=+) \nonumber \\
   = & P(f_3,f_4| x_{1-3}=2,y=2,b=+)
\end{align}

Finally, \cref{eq:em45} is identical to Eq.~\eqref{eq:ma45}.

Therefore, together with \cref{eq:em12,eq:em23,eq:em34,eq:em45}, one obtains \cref{eq:ma12,eq:ma23,eq:ma34,eq:ma45}.

Now, we prove Eq.~\eqref{eq:ma15}.
To do so, we start with the following empirical predictions by the Born rule
\begin{align}
   & \wp(f_1=1,b=0|x_{1}=1,x_2, x_3,y=1)=0, \label{eq:ema1b}\\
   &\wp(f_5=1,b=1|x_{1-3}=2, y=1)=0. \label{eq:ema5b}
\end{align}
\cref{eq:ema1b} can be seen by noticing the perfect correlation between the outcome of $A_1$ on system $S$ and the outcome of computational basis measurement on system $T$ given by the initial state $\ket{\Psi_{ST}}$. \cref{eq:ema1b} can be seen by noticing that when the outcome on $T$ corresponds to $\ket{1}_T$, the state of $S$ is effectively $\ket{v_1}_S$, which is orthogonal to $\ket{v_5}_S$.

To relate these two empirical correlations, we use Local Agency again: Since neither $f_1$ nor $b$ is in the future light cone of $x_1,x_2,x_3$, Local Agency demands that \cref{eq:ema1b} satisfies
\begin{align}
    &\wp(f_1=1,b=0|x_{1}=1,x_2, x_3,y=1) \nonumber \\
    =& P(f_1=1,b=0|x_{1-3}=2,y=1)=0. \label{eq:a1b}
\end{align}

Now notice that both $f_1$ and $f_5$ are space-like separated from $y$; thus, according to Local Agency,
\begin{align}
    &P(f_1=1,f_5=1|x_{1-3}=2,y=1) \nonumber\\
    = &P(f_1=1,f_5=1|x_{1-3}=2,y=2)=0.
\end{align}
This implies that 
\begin{align}
    P(f_1=1,f_5=1,b=+|x_{1-3}=2,y=2)=0.
\end{align}
Using \cref{eq:b+}, we arrive at
\begin{align}
P(f_1=1,f_5=1|x_{1-3}=2,y=2,b=+)=0,
\end{align}
as claimed.

\subsection{The Peres-Mermin example} \label{sec:Peres-Mermin}

\subsubsection{The Peres-Mermin noncontextuality argument}
\label{sec:recap_PM}

We briefly recap the Peres--Mermin magic square \cite{mermin1990simple,mermin1993hidden,peres1997quantum} proof of (state-independent) contextuality.
Consider the following square of quantum observables for a 4-dimensional qudit: \begin{equation} \label{eq:Peres_Mermin_array}
\left[\begin{array}{lll}
A & a & \alpha \\
B & b & \beta \\
C & c & \gamma
\end{array}\right]=\left[\begin{array}{ccc}
\sigma_z \otimes \mathbf{1} & \mathbf{1} \otimes \sigma_z & \sigma_z \otimes \sigma_z \\
1 \otimes \sigma_x & \sigma_x \otimes \mathbf{1} & \sigma_x \otimes \sigma_x \\
\sigma_z \otimes \sigma_x & \sigma_x \otimes \sigma_z & \sigma_y \otimes \sigma_y
\end{array}\right] .
\end{equation} 
Crucially, the triples of observables in a given column or row are all compatible with each other, and thus can be jointly measured. Thus the scenario consists of six possible joint measurements, called measurement contexts---one for each row and one for each column---and each observable appears in exactly two contexts (namely in one row and in one column). Moreover, the product of all three observables in any given context is  $\pm\mathbf{1}$:
\begin{equation} \label{eq:column_row_PM}
    \begin{split}
        A a \alpha = \mathbf{1}, \quad  B b \beta  = \mathbf{1}, \quad  C c \gamma = \mathbf{1}, \\
         ABC  = \mathbf{1}, \quad  abc  = \mathbf{1}, \quad \alpha \beta \gamma  = -\mathbf{1}.
    \end{split}
\end{equation} 

By KSNC, the values assigned to mutually commuting observables must obey equalities satisfied by the observables themselves, and so the product of the values assigned to the three observables in each row or column must be $\pm1$. More specifically, denoting the assignment to observable $O$ by $v_O\in \{-1,1\}$, one must have 
\begin{equation} 
    \begin{split}
         v_{A} v_{a} v_{\alpha}  = 1, \quad  v_{B} v_{b} v_{\beta}  = 1, \quad v_{C} v_{c} v_{\gamma} = 1, \\
          v_{A}v_{B}v_{C}  = 1, \quad  v_{a}v_{b}v_{c}  = 1, \quad v_{\alpha} v_{\beta} v_{\gamma} = -1.
    \end{split}
\end{equation} 
But these six equations cannot be satisfied by any nine $\pm1$-valued assignments, since the first three equalities imply that
\begin{equation}
v_{A}v_{a}v_{\alpha}v_{B}v_{b}v_{\beta}v_{C}v_{c}v_{\gamma} = 1,
\end{equation} 
while the last three imply that 
\begin{equation}
v_{A}v_{a}v_{\alpha}v_{B}v_{b}v_{\beta}v_{C}v_{c}v_{\gamma}=-1.
\end{equation}

The above arguments also establish a fact that we will utilize later, which is that the following 6 probability distributions 
\begin{subequations}
\label{eq:PMdistr}
\begin{align}
    P(v_Av_av_{\alpha}=-1)=0,\\ 
    P(v_Bv_bv_{\beta}=-1)=0, \\
    P(v_Cv_cv_{\gamma}=-1)=0, \\
    P(v_Av_Bv_C=-1)=0,\\
    P(v_av_bv_{c}=-1)=0, \\
    P(v_{\alpha}v_{\beta}v_{\gamma}=1)=0
\end{align}
\end{subequations}
cannot simultaneously be marginals of a joint probability distribution
\begin{align}
    P(v_A,v_a,v_{\alpha},v_B,v_b,v_{\beta},v_C,v_c,v_{\gamma}).
\end{align} 
(Indeed, every KSNC argument can be cast as a marginal problem of this sort.~\cite{abramsky2011sheaf}.)

\subsubsection{The LF Peres-Mermin protocol} \label{sec:PeresMerminprotocol}

We now leverage the Peres-Mermin KSNC argument to construct a novel scenario in which one can prove the LF no-go theorem.\footnote{Another extended Wigner's friend argument based on the Peres-Mermin KSNC argument was introduced in Ref.~\cite{szangolies2020quantum}; see also Ref.~\cite{walleghem2023extended}, where we argue that an additional assumption that we call Commutation Irrelevance is required there.    
Furthermore, unlike the LF Peres–Mermin protocol we develop here, the protocol in 
 \cite{szangolies2020quantum} cannot be used to prove a LF no-go theorem. The crucial difference is that our protocol includes an additional group of agents, namely, the Bobs, and allows the superobserver Alice to have measurement choices.} 
This LF Peres-Mermin scenario involves 8 agents. The first two agents are Alice (denoted $A$) and her Friend (denoted $F$), of which Alice is the superobserver of the Friend. The rest 6 agents are Bobs, denoted by $B_A$, $B_a$, $B_{\alpha}$, $B_{B}$, $B_{b}$, $B_{\beta}$ respectively. See \Cref{fig_cartoon_LF_PM} for a cartoon representation of the set-up. 
In the following, we will describe the quantum protocol that will be used to prove the no-go theorem and use the assumption of Universality of Unitarity and the Born rule for empirical correlations (cf. Appendix~\ref{app:compare}); however, the scenario could also be described in a theory-independent way if one wished to derive robust inequalities for empirical data rather than simply a no-go theorem.

\begin{figure}[]
\centering
\includegraphics[width=0.4\textwidth]{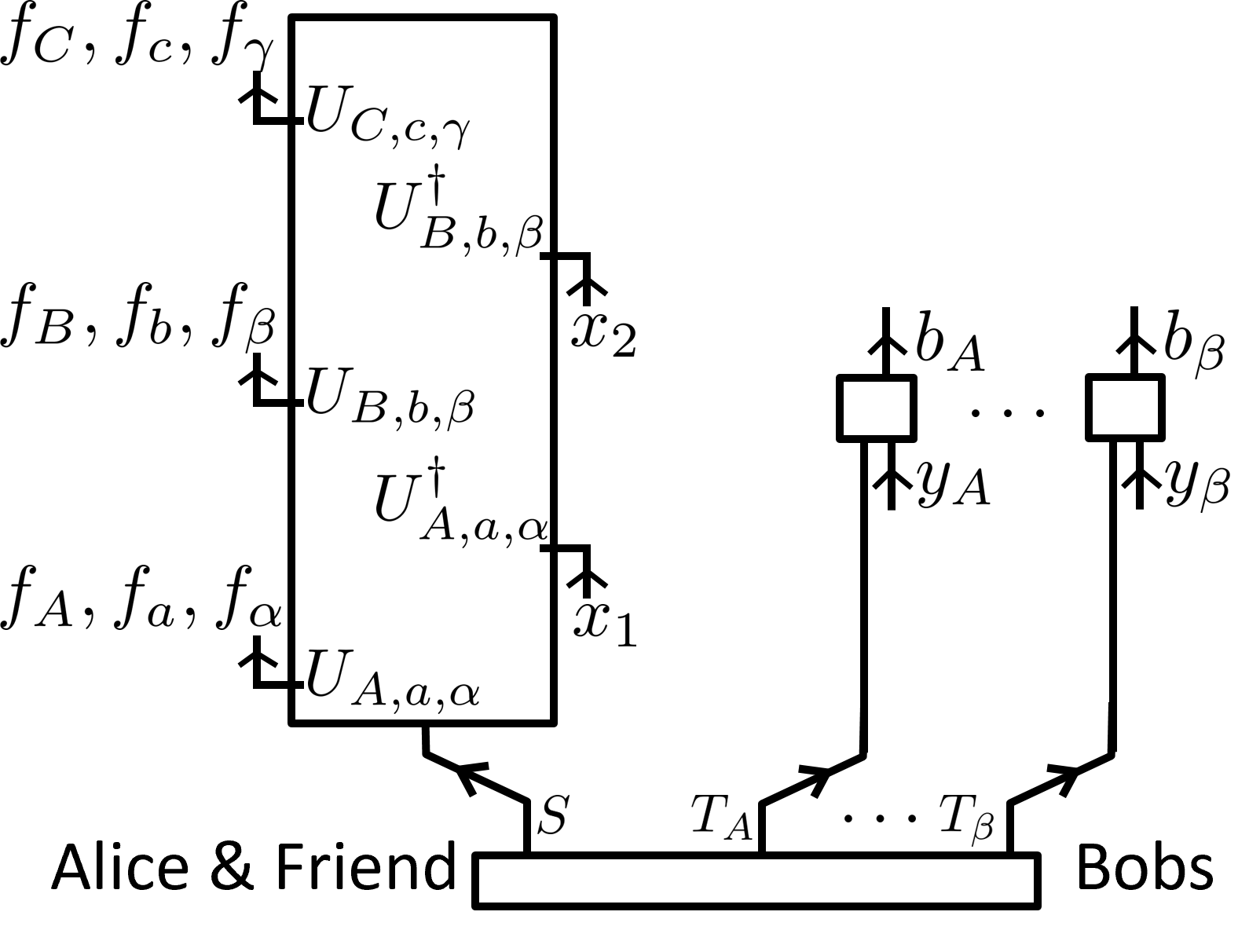}
\caption{Schematic diagram of the LF Peres-Mermin scenario, see also \Cref{fig_cartoon_LF_PM}. All operations performed by Alice and Friend are space-like separated from all operations performed by the six Bobs; all operations performed by the six Bobs are further mutually space-like separated.
}
\label{fig_setuspPM}
\end{figure}

At the beginning of the experiment, $F$ and the six Bobs share a 7-partite system $ST$ where $S$ is a qudit and $T$ consists of 6 subsystems (denoted $T_j$ for $j\in \mathbf{j}\coloneqq\{A,a,\alpha, B,b,\beta\}$),
each of which is a qubit. The joint system $ST$ is prepared in the following entangled state:
\footnote{This initial entangled state, i.e., \cref{eq:initialgiant}, can be constructed by
\begin{align}
\label{eb_construini0}
   \ket{\psi_0}_{ST} = \left( \prod_{j=1}^6 {\rm CNOT}_{ ST_j} \right) \ket{\psi} _{S} \bigotimes_{j=1}^{6} \ket{0}_{T_j}
\end{align}
where 
\begin{align}
\label{eb_creatST0}
   & {\rm CNOT}_{ ST_j} \coloneqq (\Pi_{j^{\perp}})_S\otimes\mathbf{1}_{T_j}+ (\Pi_j)_{ S}\otimes ( \ketbra{1}{0} + \ketbra{0}{1})_{T_j}. 
\end{align}
Similar to \cref{ft:mp}, this entangled state can be constructed in a measure-prepare fashion around which the thought experiment could alternatively be constructed (although some extra complexity is introduced by the fact that since, as illustrated in \cref{fig_setuspPM}, the $B,b,\beta$ measurements are not done at the initial step of the protocol, and thus, more care needs to be taken regarding the space-time arrangement of the procedures).
} 
\begin{align}
\label{eq:initialgiant}
   &\ket{\psi_0}_{ST} \coloneqq \\
   & \sum_{k_j= 0, 1, \forall j \in \mathbf{j}}
   \left( \prod_{j\in\mathbf{j}}\Pi_{{j}^{k_j}} 
\ket{\psi} \right)_{S} \bigotimes_{j\in\mathbf{j}} \ket{k_j}_{T_j}, \nonumber
\end{align}
where, $\Pi_{j^{0}}=\Pi_{j^{\perp}}$ is the projector into the subspace of the qudit $S$ that yields outcome $-1$ for the observable $j\in \mathbf{j}$, while $\Pi_{j^{1}}=\Pi_{j}=\mathbf{1}-\Pi_{j^{\perp}}$, and $\ket{\psi}$ can be \emph{any} 4-qubit quantum state. 

System $S$ is sent to the Friend while each $T_j$ is sent to the corresponding $B_j$. On Alice's side, the Friend will measure the 9 observables $\{i\}_{i\in\mathbf{i}\coloneqq\{A,a,\alpha, B,b,\beta\,C,c,\gamma\}}$ on system $S$. Specifically, she will first measure $\{A,a,\alpha\}$, then $\{B,b,\beta\}$, and finally $\{C,c,\gamma\}$. Whether she implements \emph{each} triple of observables using a joint measurement (since they commute) or in a minimally disturbing sequential manner is not important for our proof of the LF no-go theorem. In either case, from AOE, the Friend will obtain an outcome, denoted $f_i$, for each observable $i\in\mathbf{i}$. (In the case where each triple of observables is measured using a single joint measurement, the $f_i$s are the corresponding coarse-grained outcomes of the joint measurement.)  We label the $-1$ outcome for the observable $i$ by $f_i=-1$, and the other outcome by $f_i=1$. 

Under the assumption of Universality of Unitarity, we use $U_{A,a,\alpha}$, $U_{B,b,\beta}$, and $U_{C,c,\gamma}$ to denote the respective unitaries modeling the measurements for each triple of observables. The superobserver Alice is assumed to have perfect quantum control over the joint system of $F$ and $S$. She will implement the inverse unitaries $U_{A,a,\alpha}^{\dagger}$ and $U_{B,b,\beta}^{\dagger}$ at specific points in time during the protocol.

The exact sequence of these quantum operations is shown in the schematic representation of the protocol in Fig.~\ref{fig_setuspPM}.
Specifically, first, $F$ measures $A$, $a$, and $\alpha$ on $S$, after which the superobserver Alice undoes these measurements by applying $U_{A,a,\alpha}^{\dagger}$. Then $F$ measures $B$, $b$ and $\beta$ on $S$, followed by the undoing $U_{B,b,\beta}^{\dagger}$. Finally, $F$ measures $C$, $c$, and $\gamma$ on $S$.

\begin{figure}[]
\centering
\includegraphics[width=0.46\textwidth]{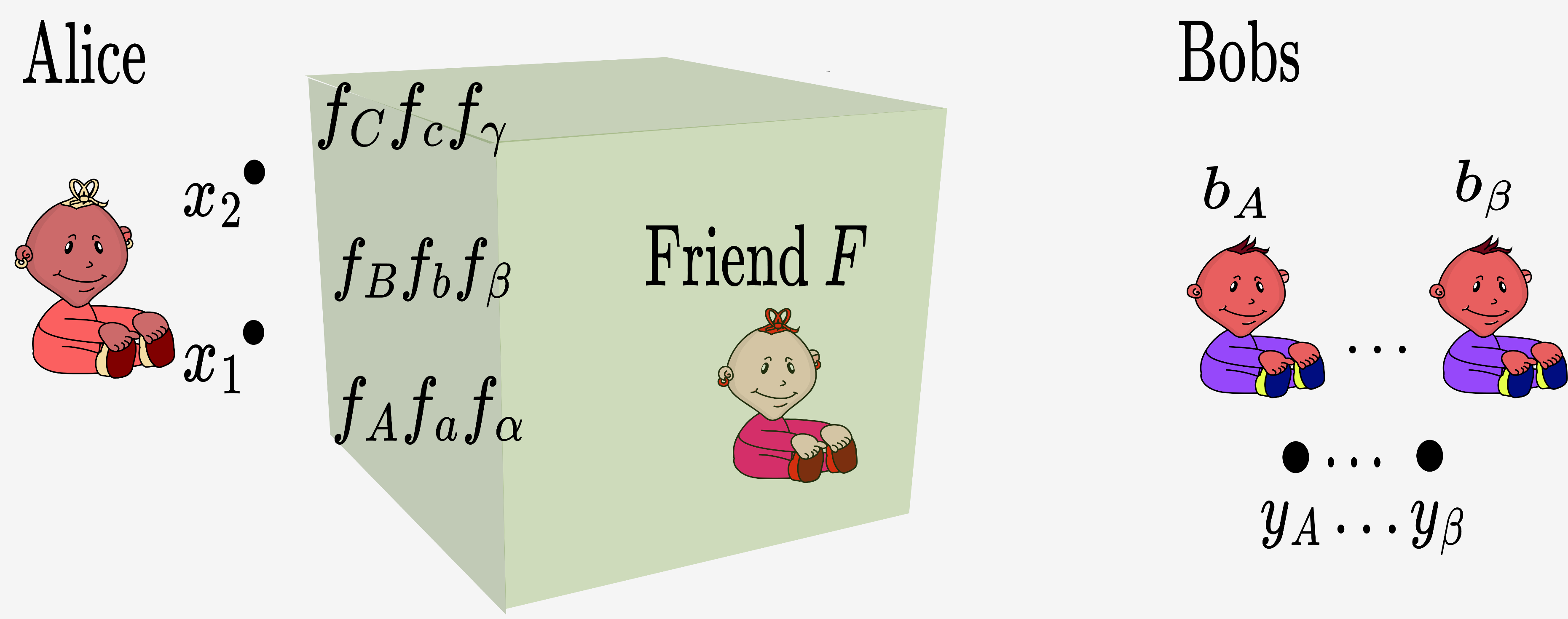}
\caption{Cartoon representation of the LF Peres-Mermin scenario, see also \Cref{fig_setuspPM}. Alice is a superobserver for her Friend, making choices $x_1,x_2$ whether to reveal the measurement outcomes obtained by the Friend, $f_A,f_a,f_\alpha$ and $f_B,f_b,f_\beta$, respectively. The spacelike separated Bobs make measurement choices $y_A,\ldots,y_\beta$ and obtain outcomes $b_A,\ldots,b_\beta$.}
\label{fig_cartoon_LF_PM}
\end{figure}

In addition, the superobserver Alice has two binary choices $x_1$ and $x_2$. The $x_1$ choice is made right before she implements the inverse unitary $U_{A,a,{\alpha}}^{\dagger}$, and the $x_2$ choice is made right before she implements the inverse unitary $U_{B,b,{\beta}}^{\dagger}$. If $x_1=1$, Alice reveals measurement outcomes $f_A$, $f_a$ and $f_{\alpha}$ to the world; if $x_1=2$, she simply continues the protocol.  Similarly, when $x_2=1$, Alice reveals measurement outcomes $f_B$, $f_b$ and $f_{\beta}$ to the world; if $x_2=2$, she simply continues the protocol. 
If $x_1=x_2=2$, the $f_C$, $f_c$ and $f_{\gamma}$ outcomes are revealed to the world at the end of the experiment.\footnote{Similar to the LF 5-cycle case, any measurement outcome obtained after $x_n=2$ is not used for the proof of the LF no-go theorem, so equivalently, whenever $x_n=2$, Alice could also terminate the protocol on her side.} 

For each Bob $B_j$ with any $j\in\mathbf{j}$, he has a binary measurement choice $y_j$. If $y_j=1$, he measures system $T_j$ in the computational basis; if $y_j=2$, he measures system $T$ in the $\pm$ basis, i.e.
\begin{align}
    \{ \ketbra{+}, \ketbra{-}  \}.
\end{align}
We denote $B_j$'s respective absolute outcome by $b_j$. 

We assume that all operations on Alice's side (including Alice and her Friend) are done at space-like separation from all Bobs' operations; furthermore, all of Bobs' operations are mutually space-like separated from each other.

Similar to the LF 5-cycle scenario, one can show that the Local Friendliness constraints on the LF Peres-Mermin scenario do not imply KSNC constraints on the compatibility graph associated with either the LF Peres-Mermin scenario or the KSNC Peres-Mermin scenario. Nevertheless, now we will show in \cref{sec:proof_LFPM} that the possibilistic proof of the failure of KSNC in the Peres-Mermin scenario can be mathematically turned
into a proof of Local Friendliness no-go theorem.

\subsubsection{The no-go theorem}
\label{sec:proof_LFPM}

In this protocol, AOE demands the existence of the following conditional probability distribution:
\begin{align}
\label{eq:PMaoe}
    P(&f_{\mathbf{i}}, b_{\mathbf{j}}|x_1,x_2, y_{\mathbf{j}}) \coloneqq \\
    P(&f_A, f_a, f_{\alpha}, f_B, f_b, f_{\beta}, f_C, f_c, f_{\gamma}, b_A, b_a, b_{\alpha}, b_B, b_b, b_{\beta} \nonumber \\
    &|x_1,x_2, y_A, y_a, y_{\alpha}, y_B, y_b, y_{\beta}) \nonumber
\end{align}
where $f_{\mathbf{i}}$,  $b_{\mathbf{j}}$, and $y_{\mathbf{j}}$ are the respective short-hand notations for \emph{all} elements in $\{f_i\}_{i\in\mathbf{i}}$, $\{b_j\}_{j\in\mathbf{j}}$, and that in $\{y_i\}_{j\in\mathbf{j}}$. 

Below, we will consider the predictions made by quantum theory for various \emph{empirical} correlations in the Peres-Mermin LF protocol, and we will then apply Local Agency to them to show that
\begin{subequations}
    \begin{align}
        P(f_A  f_{a} f_{\alpha} = -1  &|x_1=x_2=2, y_{\mathbf{j}}=2, b_{\mathbf{j}}=+)=0, \label{eq:Aaalpha}\\
        P(f_B  f_b f_{\beta} = -1  &|x_1=x_2=2, y_{\mathbf{j}}=2, b_{\mathbf{j}}=+)=0, \label{eq:Bbbeta} \\
        P(f_C  f_c f_{\gamma} = -1  &|x_1=x_2=2, y_{\mathbf{j}}=2, b_{\mathbf{j}}=+)=0, \label{eq:Ccgamma}\\
        P(f_A  f_B f_C = -1  &|x_1=x_2=2, y_{\mathbf{j}}=2, b_{\mathbf{j}}=+)=0, \label{eq:ABC} \\
        P(f_a  f_b f_c = -1 &|x_1=x_2=2, y_{\mathbf{j}}=2, b_{\mathbf{j}}=+)=0, \label{eq:abc}\\
        P(f_{\alpha}  f_{\beta} f_{\gamma} = 1  &|x_1=x_2=2, y_{\mathbf{j}}=2, b_{\mathbf{j}}=+)=0, \label{eq:alphabetagamma}
    \end{align}
\end{subequations}
where $y_{\mathbf{j}}=2$ is the short-hand notation for $y_j=2, \forall j \in \mathbf{j}$, and similarly for $b_{\mathbf{j}}=+$. 

These six equations constitute constraints on six marginals of the joint distribution
\begin{align}
\label{eq:x1y1qplus}
    P(&f_{\mathbf{i}}|x_1=x_2=2, y_{\mathbf{j}}=2, b_{\mathbf{j}}=+),
\end{align}
which is well-defined due to the existence of \cref{eq:PMaoe} together with the empirical predictions given by the Born rule that
\begin{align}
\label{eq:qplus}
    \wp(b_{\mathbf{j}}=+|x_1=x_2=2, y_{\mathbf{j}}=2 )>0.
\end{align}

However, from the Peres-Mermin KSNC proof in \cref{sec:recap_PM}, we know that these six constraints, i.e., \cref{eq:Aaalpha,eq:Bbbeta,eq:Ccgamma,eq:ABC,eq:abc,eq:alphabetagamma}, cannot be simultaneously satisfied (as they are analogous to Eqs~\eqref{eq:PMdistr}).
 Therefore, there can be no joint distribution of the form in \cref{eq:x1y1qplus} that is consistent with the empirical predictions given by the Born rule and Local Agency. Hence the no-go theorem.

It remains only to establish \cref{eq:Aaalpha,eq:Bbbeta,eq:Ccgamma,eq:ABC,eq:abc,eq:alphabetagamma}. The steps for proving \cref{eq:Aaalpha,eq:Bbbeta,eq:Ccgamma} are similar to those for proving \cref{eq:ma12,eq:ma23,eq:ma34,eq:ma45} in \cref{sec:5cycle_proof}, and the steps for proving \cref{eq:ABC,eq:ABC,eq:alphabetagamma} are similar to those for proving \cref{eq:ma15} in \cref{sec:5cycle_proof}. For completeness, a proof can be found in \Cref{app:proof_PeresMermin_eqs}.

\subsection{Translating any possibilistic KSNC argument into a LF argument } \label{sec:Commutation_Friendlinessmain}

We have now seen two examples of how possibilistic proofs of the failure of KSNC can be translated into LF arguments. More generally, one has the following:

\begin{restatable}{theorem}{LFpossKSNC}
    \label{th:LFpossKSNC}
Every quantum possibilistic proof of the failure of Kochen-Specker noncontextuality can be mathematically translated into a proof of the failure of Local Friendliness in a corresponding extended Wigner's friend scenario.
\end{restatable} 

We prove this in Appendix~\ref{app:proof_LFpossKSNC}, using a construction analogous to the above constructions (where all Bobs only need to perform either the computational basis measurement or the complementary $\pm$ measurement). 

Note that in these LF constructions, one introduces extra measurements beyond those in the contextuality scenarios\footnote{It is commonplace for translations between different scenarios to involve slightly different states and measurements; see for example Ref.~\cite{KunjwalKS2015,PhysRevLett.129.240401,cabello2021converting} for other examples.};
in particular, the measurements made by Bobs on qubits $T_j$. 
However, one can avoid the addition of these additional measurements (and any choices of Alice or Bobs') if one appeals to the assumption of \textit{Commutation Irrelevance} \cite{walleghem2023extended} rather than Local Agency. 
We discuss this in Appendix~\ref{app:CI}. 

One could alternatively provide a different sort of translation from KSNC arguments to LF arguments by first translating the KS construction into a Bell scenario using known procedures~\cite{cabello2021converting,cabello2010proposal,aolita2012fully,cabello2011proposal}, 
and then using the mapping from Bell scenarios to LF scenarios given by reversing the translation given in Section~\ref{sec:generalisation_LF_Bell}. (Note that the translation from a KSNC argument to a Bell argument also involves adding additional measurements in the case where the initial KSNC argument is state-dependent~\cite{cabello2021converting}.)

\section{Discussion} \label{sec:discussion}

We close with a few further comments and open questions.

In Section~\ref{sec:LFasKSNC}, we found a wide class of scenarios wherein the LF and Bell polytopes coincide. 
Within the context of the initial KSNC scenario, the LF assumptions do not imply any notable constraints (for example, because KSNC scenarios often do not even involve more than one system). But after mapping the KSNC scenario to a LF scenario as in \cref{sec:KSNC_to_LF}, the LF assumptions are sufficient to obtain a proof of the LF theorem. 
In \cref{sec:KSNC_to_LF}, we give techniques for deriving LF inequalities; naturally, then it would be interesting in future work to apply the techniques we used in the first half of the paper to more thoroughly characterize the LF constraints in these new LF scenarios using the relation between the LF constraints and KSNC constraints on subgraphs of the compatibility graph for the LF scenario.

In this work, we have only explored the connection between Local Friendliness and Kochen-Specker noncontextuality. But the primary (and arguably only) motivation for Kochen-Specker noncontextuality is in fact the generalized notion of noncontextuality~\cite{spekkens2005contextuality,spekkens2019ontologicalidentityempiricalindiscernibles}, and it would be interesting in the future to explore the connection between generalized noncontextuality and EWF no-go theorems.

\section*{Acknowledgements}
Many of the initial ideas of this paper came from discussions LW had with Rui Soares Barbosa. We thank Howard Wiseman, Eric Cavalcanti, Marwan Haddara, Emmanuel Zambrini Cruzeiro, Yeong-Cherng Liang, Ravi Kunjwal, Ricardo Faleiro, Ernest Tan, 
Marina Maciel Ansanelli for useful discussions. 
LW also thanks the International Iberian Nanotechnology Laboratory--INL in Braga, Portugal and the Quantum and Linear-Optical Computation (QLOC) group for the kind hospitality. YY was supported by Perimeter Institute for Theoretical Physics. Research at Perimeter Institute is supported in part by the Government of Canada through the Department of Innovation, Science and Economic Development and by the Province of Ontario through the Ministry of Colleges and Universities. YY was also supported by the Natural Sciences and Engineering Research Council of Canada (Grant No. RGPIN-2024-04419). 
RW acknowledges support from FCT–-Fundação para a
Ciência e a Tecnologia (Portugal) through PhD Grant SFRH/BD/151199/2021 and from the European Research Council (ERC) under the European Union’s Horizon 2020 research and innovation programme (grant agreement No.856432, HyperQ). LW acknowledges support from the United Kingdom Engineering and Physical Sciences Research Council (EPSRC) DTP Studentship (grant number EP/W524657/1).
DS was supported by the Foundation for Polish Science (IRAP project, ICTQT, contract no. MAB/2018/5, co-financed by EU within Smart Growth Operational Programme). 
While finalizing this manuscript, we became aware of concurrent works in Ref.~\cite{haddara2024local,haddaraSecond} that also include investigations of the relationship between the LF and Bell polytopes. 
We thank the authors of these works, Eric Cavalcanti, Marwan Haddara and Howard Wiseman, for feedback on our manuscript.

\vspace{1cm}

\vspace{1cm}

\bibliographystyle{unsrtnat}

%

\newpage
\appendix

\section{More on the LF scenario and theorem}
\label{app:compare}

Compared to the usual presentations of the Local Friendliness scenario~\cite{bong2020strong,wiseman2023thoughtful,haddara2022possibilistic}, where superobservers reverse their friend’s measurement and perform new measurements themselves, the superobservers (namely, Alice and Bob) in our description do not perform these new measurements themselves; instead, the new measurements are carried out by Charlie and Debbie (after the reversals). These two descriptions are equivalent, since a measurement outcome does not depend on who performs the measurement. Furthermore, as noted in the literature (e.g., in~\cite[Sec. II.E]{schmid2023review}), the quantum predictions in such an experiment are the same as the predictions in experiments where superobservers measure the entire labs of their respective friends (such as in the Frauchiger–Renner scenario~\cite{frauchiger2018quantum}) in appropriate bases instead of reversing and remeasuring. \

Now we provide more information on the proposed quantum experiments~\cite{bong2020strong,wiseman2023thoughtful}  where the LF inequalities are violated. In such proposals, Charlie and Debbie are required to perform their measurements in their respective laboratories in a sufficiently isolated way so that each of their laboratories can be viewed as a closed system, while Alice and Bob are superobservers who possess the extreme technological abilities to perform (close enough to) arbitrary quantum operations on an observer's laboratory. This can be motivated by \emph{Universality of Unitarity}: 
that all possible dynamics of a closed system can be described as some unitary map acting on the Hilbert space associated to that closed system (even the dynamics of macroscopic systems including observers performing measurements).

Then, the evolution of Charlie's lab during Charlie's measurement can be modeled by a unitary $U_C$ and the evolution of Debbie's lab during Debbie's measurement can be modeled by a unitary $U_D$. Consequently, Alice can reverse the evolution of Charlie's lab by applying $U_C^{\dagger}$, and Bob can reverse the evolution of Debbie's lab by applying $U_D^{\dagger}$. The Born rule predicts that the LF inequalities can be violated in such quantum LF experiments, and hence one can prove the following no-go theorem~\cite{bong2020strong}. (Here and as usual, we assume the operational adequacy of quantum theory. That is, any operationally \emph{accessible} correlation obeys the Born rule.) We refer the readers to Ref.~\cite{schmid2023review} for a simple proof of the LF theorem.

Note that the Local Friendliness scenario can be presented in an entirely theory-independent manner, without any mention of quantum theory. Similarly, the derivation of LF inequalities is also theory-independent, relying solely on the Local Friendliness assumptions applied to the LF scenarios.
Violations of the LF inequalities in an LF scenario then signal the failure of the LF assumptions in and of themselves. Quantum theory predicts the violation of these inequalities, provided one assumes that the friends' measurements can be modeled approximately as unitaries and that advanced enough observers can perform quantum operations on their friends and their environment.
For simplicity, our generalizations of LF scenarios and the new LF scenarios inspired by KSNC arguments are not presented in a theory-independent way, but they can nevertheless be rephrased to be theory-independent. 

\

\section{More on the gluing trick}
\label{app:glu}

First, we supplement Sec.~\ref{sec:32} with the proof for $P_{22,32}$ to recover all six pairwise correlations.
\begin{itemize}
    \item Marginalizing over $a_3$ gives
    \begin{align}
    \label{eq:321}
    &\sum_{a_3} P_{22,32} \nonumber\\
    = &\frac{P_{22}(a_1,a_2,b_1,b_2)\sum_{a_3} P_{32}(a_1,a_3,b_1,b_2)}{\sum_{a_3} P_{32}(a_1,a_3,b_1,b_2)}\nonumber\\
    = &P_{22}(a_1,a_2,b_1,b_2),
    \end{align}
    which recovers $\wp(a_2,b_2)$, $\wp(a_1,b_2)$, $\wp(a_2,b_1)$ and $\wp(a_1,b_1)$ as marginals;
    \item From \cref{eq:bcd2232}, the denominator in $ P_{22,32}$ is the same as $\sum_{a_2} P_{22}(a_1,a_2,b_1,b_2)$ and thus,
    \begin{align}
    \label{eq:322}
    & P_{22,32} \\
    =  & \frac{P_{22}(a_1,a_2,b_1,b_2)P_{32}(a_1,a_3,b_1,b_2)}{\sum_{a_2} P_{22}(a_1,a_2,b_1,b_2)}.\nonumber
    \end{align}
    Marginalizing over $a_2$ gives
    \begin{align}
    \label{eq:323}
    &\sum_{a_2} P_{22,32} \nonumber\\
    = &\frac{\bigl(\sum_{a_2} P_{22}(a_1,a_2,b_1,b_2)\bigr)P_{32}(a_1,a_3,b_1,b_2)}{\sum_{a_2} P_{22}(a_1,a_2,b_1,b_2)}\nonumber\\
    =&P_{32}(a_1,a_3,b_1,b_2),
    \end{align}
    which recovers $\wp(a_3,b_2)$ and $\wp(a_3,b_1)$ (and $\wp(a_1,b_2)$, and $\wp(a_1,b_1)$) as marginals.
\end{itemize}

Secondly, we supplement Sec.~\ref{sec:minimal} with the explicit proof of how the LF polytope coincides with the corresponding Bell polytope.

Analogous to \cref{eq:LAij}, we have
\begin{align}
    \sum_{a_2} P(a_1a_2b_1|x=2,y=1)= \wp(a_1,b_1),\label{eq:22em1}\\
    \sum_{a_1} P(a_1a_2b_1|x=2,y=1)= \wp(a_2,b_1),\label{eq:22em2}
\end{align}
and
\begin{align}
    \sum_{a_2} P(a_1a_2b_2|x=2,y=2)= \wp(a_1,b_2),\label{eq:22em3}\\
    \sum_{a_1} P(a_1a_2b_2|x=2,y=2)= \wp(a_2,b_2).\label{eq:22em4}
\end{align}
Analogous to \cref{eq:bcd2232}, $P(abc|x=2,y=1)$ and $P(abc|x=2,y=2)$ are further related by 
\begin{equation}
\label{eq:2122}
    \sum_{b_1} P(a_1a_2b_1|x=2,y=1)= \sum_{b_2} P(a_1a_2b_2|x=2,y=2),
\end{equation}
due to Local Agency. Using \cref{eq:2122}, one can compute that the following joint distribution (which is the analog of \cref{eq:2232}) has all four pairwise correlations $\wp(a_x,b_y)$ for $x=1,2$ and $y=1,2$ as its marginals:
\begin{equation}
\frac{P(a_1a_2b_1|x=2,y=1)P(a_1a_2b_2|x=2,y=2)}{\sum_{b_2} P(a_1a_2b_2|x=2,y=2)}.
\end{equation}
Thus, these four pairwise correlations satisfy the KSNC constraints (or equivalently, the CHSH inequalities) on the compatibility graph in \cref{fig:a1aib1bj} where $i=2$ and $j=2$. 

Furthermore, similar to \cref{sec:equi}, we can show that the existence of $P(a_1a_2b_1|x=2,y=1)$ and $P(a_1a_2b_2|x=2,y=2)$ and that they reproduce empirical correlations as in \cref{eq:22em1,eq:22em2,eq:22em3,eq:22em4} while agreeing on their overlaps as in \cref{eq:2122} are all the LF constraints in this minimal scenario. 
So the Local Friendliness polytope coincides with the corresponding Bell polytope in the minimal scenario.

Thirdly, we show how one can gain some intuition for why the LF polytope for the 3,3-setting LF scenario is strictly bigger than the corresponding Bell polytope (cf. Sec.~\ref{sec:33}) by considering how the gluing trick introduced in \cref{sec:glue} works. The building blocks of our gluing trick are distributions over outcomes of four measurements that are vertices of a compatibility graph in the form of \cref{fig:NaBa}, where each vertex has two edges. As such, when we glue two of these building-blocks together, such as when gluing $P_{22}$ and $P_{32}$ in \cref{eq:2232}, the new distribution $P_{22,32}$ has one more variable (compared to $P_{22}$), namely $a_3$ from $P_{32}$; moreover, it has as its marginals two more empirical pairwise correlations, namely the ones that are marginals of $P_{32}$ involving $a_3$, i.e., $\wp(a_3,b_1)$ and $\wp(a_3,b_2)$. Equivalently, compared to the compatibility graph \cref{fig:a12b120} for $P_{22}$, the compatibility graph \cref{fig:a123b12m} for $P_{22,32}$ has one more vertex, namely $A_3$, coming from the compatibility graph \cref{fig:a13b120} for $P_{32}$, and two more edges, namely the ones connected with $A_3$ in \cref{fig:a13b120}---i.e., the $A_3$-$B_1$ edge and the $A_2$-$B_2$ edge. That is, when we glue one building block to another, adding one more vertex comes with exactly two more edges. Furthermore, the gluing operation cannot add any edges without adding vertices. Therefore, it is impossible to obtain the full compatibility graph of the 3,3-setting LF scenario via gluing since the difference between that and any building block in \cref{fig:4pair} includes 2 vertices but 5 edges, while the gluing operations can only get 4 more edges when adding 2 vertices. Therefore, we could only get $P_{\lnot a_3b_3}$, $P_{\lnot a_2b_2}$, $P_{\lnot a_2b_3}$ and $P_{\lnot a_3b_2}$ and we cannot get KSNC constraints on the full compatibility graphs of the 3,3-setting LF scenario via gluing. As mentioned earlier, the Local Friendliness assumptions in fact do not require the KSNC constraints on the compatibility graph for this scenario, and thus, the LF constraints in this scenario are strictly weaker than the KSNC constraints.

\section{Connecting our results in Sec.~\ref{sec:33} to Ref.~\cite{bong2020strong}}

\label{app:connect}
Table 1 of Ref.~\cite{bong2020strong} listed all of the LF inequalities for the 3,3-setting case with binary outcomes; these are tight LF inequalities and were obtained directly via a linear program. We now compare our results of \Cref{sec:33} to these inequalities.

The first row of \cite[Table 1]{bong2020strong} indicates that the CHSH inequalities on any set of four pairwise correlations, $\bigl\{\wp(a_x,b_y)|x\in\{1,i\}, y\in\{1,j\}\bigr\}$, with $i,j\in\{2,3\}$,
are also LF inequalities. This is a special case of what we proved in item I.1-I.4 when one focuses on binary outcomes (since the KSNC constraints are known to be CHSH inequalities in that case). 

The second row of \cite[Table 1]{bong2020strong} indicates that the CHSH inequalities on any set of four pairwise correlations $\bigl\{\wp(a_x,b_y)|x\in\{1,i\}, y\in\{2,3\}\bigr\}$ or $\bigl\{\wp(a_x,b_y)|x\in\{2,3\}, y\in\{1,i\}\bigr\}$ with $i\in\{2,3\}$, are also LF inequalities. This is a special case of what we proved in II.1-II.2 when one focuses on binary outcomes, since \cref{fig:noa3b3} has \cref{fig:a123b23} and \cref{fig:a23b12} as subgraphs, and \cref{fig:noa2b2} has \cref{fig:a13b23} and \cref{fig:a23b13} as subgraphs. Consequently, II.1 implies that the four pairwise correlations $\wp(a_x,b_y)$ with $x=1,2$ and $y=2,3$ (or with  $x=2,3$ and $y=1,2$) must obey KSNC constraints on the compatibility graph in \cref{fig:a123b23} (or respectively, \cref{fig:a23b12}), while II.2 implies that the four pairwise correlations $\wp(a_x,b_y)$ with $x=1,3$ and $y=2,3$ (or with  $x=2,3$ and $y=1,3$) must obey KSNC constraints on the compatibility graph in \cref{fig:a13b23} (or respectively, \cref{fig:a23b13}). And for each of these graphs, restricted to measurements with binary outcomes, the KSNC inequalities are known to be CHSH inequalities. 

\begin{figure}[!h]
\captionsetup[subfigure]{aboveskip=-2pt,belowskip=-1pt}
    \centering
    \begin{subfigure}[b]{0.11\textwidth}
         \centering
         \includegraphics[width=\textwidth]{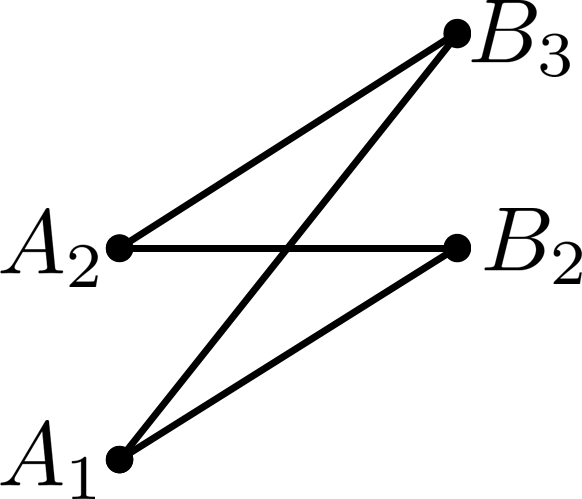}
         \caption{}
    \label{fig:a123b23}
     \end{subfigure}
    \hspace{0.4mm}
    \begin{subfigure}[b]{0.11\textwidth}
         \centering
         \includegraphics[width=\textwidth]{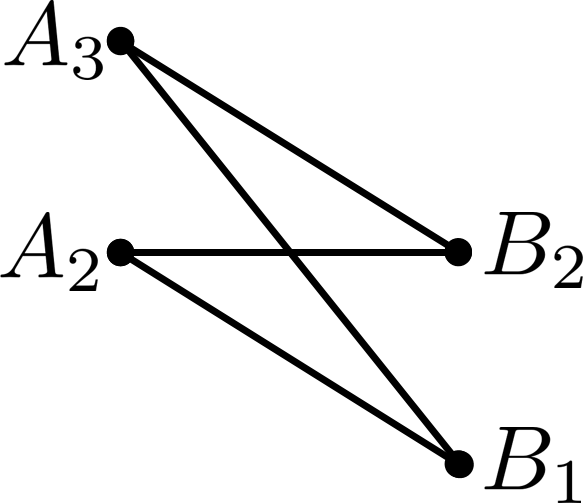}
         \caption{}
         \label{fig:a23b12}
     \end{subfigure}
     \hspace{0.4mm}
    \begin{subfigure}[b]{0.11\textwidth}
         \centering
         \includegraphics[width=\textwidth]{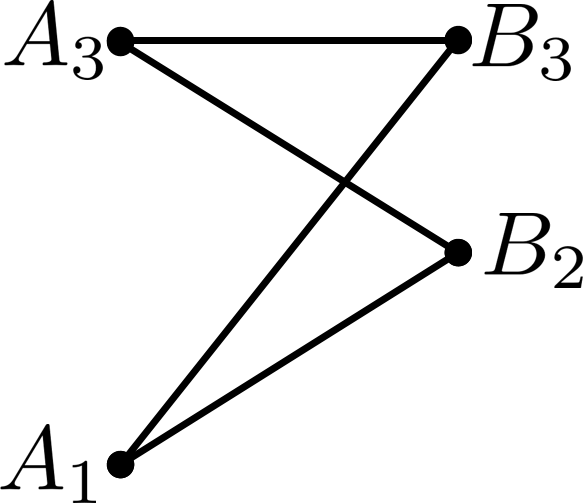}
         \caption{}
         \label{fig:a13b23}
     \end{subfigure}
    \hspace{0.4mm}
    \begin{subfigure}[b]{0.11\textwidth}
         \centering
         \includegraphics[width=\textwidth]{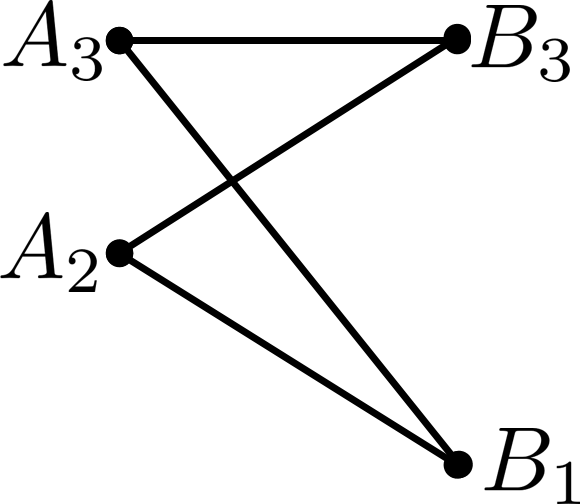}
         \caption{}
         \label{fig:a23b13}
     \end{subfigure}
    \caption{Subgraphs of the compatibility graphs in \cref{fig:8pair} with four edges that are different from the ones in \cref{fig:4pair}.}
    \label{fig:new4pair}
\end{figure}

The third row of Ref.~\cite[Table 1]{bong2020strong} indicates that the CHSH inequalities on the four pairwise correlations $\wp(a_x,b_y)$ for $x=2,3$ and $y=2,3$ are \emph{not} LF inequalities.  This can be seen, for example, by noting that the PR box correlations~\cite{Popescu1994} (which maximally violates CHSH inequalities) do not violate \cref{eq:LA}, and hence do not violate Local Agency.  Note that none of the compatibility graphs in \cref{fig:8pair} (corresponding to our item II.1-II.4) includes a subgraph representing the Bell scenario for $\{A_2,A_3,B_2,B_3\}$, thus, we have not yet derived CHSH inequalities on  $\wp(a_x,b_y)$ for $x=2,3$ and $y=2,3$.

Having seen the pattern so far, one might have expected that every Local Friendliness inequality can be recovered as a KSNC inequality for {\em some} compatibility graph---e.g., for some particular subgraph of the full compatibility graph for the measurements in the scenario. However, this is not true. 
One can see this by considering the facets of the LF polytope in Equations~(13), (14) and (16) of the Appendix of Ref.~\cite{bong2020strong}; they contain the four pairwise correlations $P(a_2, b_2)$, $P(a_2, b_3)$, $P(a_3, b_2)$, and $P(a_3, b_3)$. These four correlations could only appear together in a KSNC inequality in which all four of them are marginals of a single joint probability distribution~\cite{Kunjwal2015}. But if so,
by this fact, Fine's theorem implies that they must satisfy the CHSH inequalities, and, as discussed just above, these are {\em not} Local Friendliness inequalities. 

Each of the KSNC compatibility graphs in \cref{fig:8pair} gives rise to a KSNC polytope that is defined as the set of correlations satisfying the KSNC constraints on the corresponding graph. By taking the intersection\footnote{ More precisely, we first express each of the four KSNC polytopes in the full 36 dimensional behavior space via an H-representation (meaning that the polytope is expressed in terms of its facet-defining inequalities), and then take their intersections by considering all facet-defining inequalities as part of a single H-representation.} of the four KSNC polytopes from the four subfigures in \cref{fig:8pair}, we obtain a polytope of LF inequalities. From the argument above, this polytope strictly contains the LF polytope defined in Ref.~\cite{bong2020strong} for the 3,3-setting LF scenario. That is, our approach reproduces some, but not all, of the LF inequalities characterizing the LF polytope, computed explicitly in Ref.~\cite{bong2020strong}. We have also explicitly checked this fact using the ``traf'' command of the PORTA package. We found that the LF polytope for the 3,3-setting LF scenario is indeed strictly contained in the intersection of the four KSNC polytopes from the four subfigures in \cref{fig:8pair}. Both polytopes are defined in a $36$ dimensional space of behaviors, but the dimensions of their convex polytopes (i.e., the dimensions of their respective affine hull of points) are both $15$.  We also find that every extremal point of the LF polytope is \emph{also} an extremal point of the intersection of KSNC polytopes (but not vice versa). Furthermore, we found that the LF inequality in E   q.~(15) in Ref.~\cite{bong2020strong}, a so-called  $I_{3322}$ Bell inequality that does not involve $A_3$-$B_3$ correlations, is also a facet inequality of the intersection of the four KSNC polytopes; in fact, it is a facet inequality of the KSNC polytope for the compatibility subgraph in \cref{fig:noa3b3}. 

However, our approach can generate LF inequalities beyond those in prior work. For instance, it generates an infinite family of novel LF inequalities arising when one takes the cardinality of the outcomes to be greater than 2, and generates LF inequalities in scenarios with higher cardinality setting variables. Our approach can also be straightforwardly generalized to LF scenarios with sequential measurements as in Ref.~\cite{utreras2023allowing}, but with arbitrary cardinality of setting and outcome variables.

\section{Proof for the quantum
violations of LF inequalities derived in Sec. \ref{sec:LFasKSNC}}
\label{app:violate}

Every LF inequality we derive by connecting with KSNC in \cref{sec:LFasKSNC} is also a Bell inequality (which is not necessarily tight). Here, we prove that, if such a Bell inequality can be violated in a quantum realization of a corresponding Bell scenario, then it can also be violated in the respective LF scenario.

Consider a quantum realization of the $N_A,N_B$-setting Bell scenario that violates this Bell inequality, where Alice's measurements are denoted $A'_{x'}$ for her respective setting $x'=1,2,\dots, N_A $, and Bob's measurements are denoted $B'_{y'}$ for his respective setting $y'=1,2,\dots, N_B$. We further denote the bipartite quantum state shared between Alice and Bob as $\ket{\psi}$. 

Consider the corresponding $N_A,N_B$-setting LF scenario defined in \cref{sec:recap_LF_scenario}, in which we now construct a quantum violation of the same Bell inequality (now viewed as a LF inequality).
We take the bipartite system being measured to also be a quantum system in the state $\ket{\psi}$. 
We let the measurements performed by Charlie satisfy  $A_i=A'_i$, $\forall i\in \{1,2,\dots,N_A\}$ and similarly measurements performed by Debbie satisfy $B_j=B'_j$ for any $j\in\{1,2,\dots,N_B\}$. 

Moreover, we let Charlie perform his measurement in a closed lab such that according to unitary quantum theory, any evolution of his labs, including his measurement process, is unitary. Then, if Alice further has the extreme technological capability to perform unitary operations on Charlie's lab, she could, in principle, apply the inverse of the unitary corresponding to Charlie's measurement so that the state of the lab, including Charlie and his share of the bipartite system, is restored to the initial state prior to his measurement. In particular, we let Alice reverse Charlie's measurement process by applying this inverse unitary when $x\neq 1$. 

Similarly, we let Debbie perform her measurement in a closed lab, and we assume Bob has the extreme technological capability to reverse Debbie's measurement process and he does so by applying the inverse of the corresponding unitary when $y\neq1$.

As such, 
the pairwise correlations $\wp(a_x,b_y)$ for $x=1,2,\dots,N_A$ and $y=1,2,\dots,N_B$ in this LF scenario are all correlations between outcomes of measurements $A'_x$ and $B'_y$ performed on $\ket{\psi}$, and consequently are the same as what they were in the Bell scenario, and so violate the same inequality.

\section{General classes of LF scenarios where all LF constraints are also KSNC constraints}
\label{app:generalisation_LF_Bell}

In this section, we consider a variety of LF scenarios that are generalizations of the ones we have discussed so far, and we show the condition under which all the LF constraints in a given scenario can also be viewed as KSNC constraints, and in particular, Bell inequalities.

\vspace{0.5cm}

\paragraph{Multipartite sequential LF scenario with binary settings} \label{par:LF_construction}

In the bipartite LF scenario defined in \cref{sec:recap_LF_scenario}, there are only two superobserver-friend pairs (namely, the Alice-Charlie pair and the Bob-Debbie pair). We now consider a class of multipartite LF scenarios where there are $\chi$ superobserver-friend pairs for any integer $\chi>1$. The superobservers are called $A,B,C,\dots$, and their respective friends are called $F_A, F_B, F_C,\dots$. For each superobserver $\Omega$ with $\Omega\in \{A,B,C,\dots\}$, their friend $F_{\Omega}$ will make various measurements from the set $\{\Omega_1,\Omega_2,\dots,\Omega_{N_{\Omega}}\}$, one after another in sequence, but with the superobserver $\Omega$ undoing each of their measurements in between.  Furthermore, after each of the friend's measurements, the superobserver has a binary choice to decide whether to reveal the friend's outcome (and stop their portion of the experiment) or reverse their measurement. In total, each superobserver has $N_{\Omega}-1$ binary measurement choices.  
Such a multipartite sequential LF scenario is a generalization (to the multipartite case) of the sequential LF scenario introduced in Ref.~\cite{utreras2023allowing}. Now we will describe the protocol in detail.

At the start of the experiment, a $\chi$-partite state is distributed between the $\chi$ friends, who perform their first respective measurement, $A_1,B_1,\ldots$, on their share of the system, obtaining outcomes $a_1,b_1,c_1,\dots$, respectively. 

Then, each superobserver makes a binary choice. 
These choices are denoted as $x_{A_1},x_{B_1},x_{C_1},\dots\in \{1,2\}$. When $x_{A_1}=1$, the superobserver $A$ reveals the outcome $a_1$ of her friend $F_A$ to the world, and agents $A$ and $F_A$ halt their portion of the experiment.
When $x_{A_1}=2$, the superobserver $A$ reverses the evolution of $F_A$'s lab due to $F_A$'s measurement, and then 
instructs $F_A$ to do another measurement $A_2$, the outcome of which is denoted $a_2$. Similarly, for the other superobservers, when their respective setting choice is 1, they reveal the corresponding friend's outcome, and that superobserver and their friend halt their portion of the experiment; when their respective setting choice is 2, they reverse their friend's measurement and ask them to make another measurement, denoted $B_2,C_2,\dots$, respectively, with the corresponding outcome $b_2,c_2,\dots$.

For each superobserver whose portion of the experiment has not stopped, another binary choice is then made, denoted $x_{A_2},x_{B_2},x_{C_2},\dots\in \{1,2\}$. When a superobserver's choice is 1, they reveal their friend's outcome (which is $a_2, b_2$ or $c_2$, etc.)  and their portion of the experiment stops; when their choice is 2, they reverse their friend's measurement and ask them to make another measurement, denoted $A_3,B_3,C_3,\dots$, respectively, with the corresponding outcome $a_3,b_3,c_3,\dots$.

For each superobserver-friend $\Omega$-$F_{\Omega}$ pair (for any $\Omega=A,B,C,\dots$), this process is repeated until either the friend's measurement outcome is revealed, or until the superobserver $\Omega$  reaches their $(N_{\Omega}-1)$th binary choice, i.e., the final choice. In the latter case, when, $x_{{\Omega}_{N_{\Omega}-1}}=2$, the superobserver $\Omega$ reverses her friend $F_{\Omega}$'s measurement, asking them to perform measurement $\Omega_{N_{\Omega}}$ and then reveals the corresponding outcome. 

The operations done by each superobserver-friend pair are mutually space-like separated.

At the end of the experiment, we collect statistics for whichever outcomes were revealed given the setting choices.

For a given multipartite sequential LF scenario of this sort, one can consider the corresponding multipartite Bell scenario with $\chi$ agents $A,B,\ldots$ whose respective measurement choices are $x,y,\dots$, corresponding to measurements $A_x,B_y,\ldots$  (so that when  $x=1$, $A$ performs the measurement $A_1$, and so on.) The range of their choices are $x\in \{1,2,\dots,N_A\}$, $y\in \{1,2,\dots,N_B\}$ and etcetera.

\begin{restatable}{lemma}{lembinary}\label{lem:LF_binary}
   The LF constraints on the empirical correlations obtained in the multipartite sequential LF scenario with binary choices (as defined just above) are equivalent to the constraints from KSNC in the corresponding multipartite Bell scenario for any integer $\chi,N_A,N_B,\dots,>1$. That is, the LF polytope equals the Bell polytope in these scenarios.
\end{restatable}

The proof of \cref{lem:LF_binary} is in Appendix~\ref{app:LF_binary}. 
It is analogous to that in \cref{sec:22} for the 2,2-setting bipartite LF scenario, and we sketch the proof now.
There are, in total, $\sum_{\Omega}N_{\Omega}$ measurements involved in such a multipartite sequential LF scenario with binary choices, namely $A_1, A_2, \dots, A_{N_A}$, $B_1, B_2, \dots, B_{N_B}$, etcetera. In the runs where all the superobservers have all of their respective choices being 1 except that their final choice is 2, all the $\sum_{\Omega}N_{\Omega}$ measurements are performed by the friends. According to AOE, for those runs, there exists a  distribution
\begin{align}
P&(a_1,\dots,a_{N_A},b_1,\dots,b_{N_B},\dots \\
&|x_{A_1}=\dots =x_{A_{N_A-1}}=x_{B_1}=\dots=x_{{N_B-1}}=\dots=2). \nonumber
\end{align}
Then, using Local Agency, one can show that this joint distribution has as its marginals all $\chi$-way empirical correlations such as $\wp(a_1,b_1,\dots|x_{A_1}=x_{B_1}=\dots=1)$, $\wp(a_2,b_1,\dots|x_{A_2}=x_{B_1}=\dots=1,x_{A_1}=2)$, etcetera.
Thus, LF demands the KSNC constraints on the multipartite sequential LF scenario with binary choices, which are, equivalently, Bell inequalities in the corresponding Bell scenario. Following an analogous proof as in \Cref{sec:Bell_inside_LF}, one can further show that the Bell polytope formed by these Bell inequalities must be contained inside the LF polytope for the LF scenario. As such, the LF polytope here equals the corresponding Bell polytope.

\vspace{0.5cm}

\paragraph{Further generalization with one superobserver having a high-cardinality final choice and/or losing their friend}
\label{par:further}

Now, recall that earlier in \cref{sec:32} for the 3,2-setting bipartite LF scenario and for the minimal LF scenario, all the LF constraints are {\em also} equivalent to the constraints from KSNC in a corresponding Bell scenario. While the multipartite sequential LF scenario with binary choices is analogous to (and can be seen as a generalization of) the 2,2-setting bipartite LF scenario, can we further generalize it to obtain an analog to the 3,2-setting bipartite LF scenario or the minimal scenario, and in this case, are the LF constraints equivalent to KSNC constraints from some corresponding Bell scenario? 
The answer is affirmative.

First, let us define the scenario as a generalization of that given in the previous section. We now allow one of the superobservers to have their \emph{final} choice of measurement to have arbitrary cardinality (as opposed to being binary). 
Without loss of generality, let us assume that it is the superobserver $A$ who has the high-cardinality final choice, i.e., $x_{A_{N_A}}\in \{1,2,\dots,\xi\}$ for some integer $\xi>2$. 
As before, if for this final choice the superobserver $A$ chooses $x_{A_{N_A}}=1$, the superobserver $A$ reveals the outcome of her friend $F_A$'s last measurement. Now if for this final choice $x_{A_{N_A}}$ the superobserver chooses $\alpha \in \{2,\ldots,\xi\}$, the superobserver $A$ reverses her friend $F_A$'s last measurement and tells $F_A$ to perform another measurement corresponding to the choice $\alpha$, after which this outcome is revealed.

Such an LF scenario can be viewed as containing $\xi-1$ subscenarios, each of which is a multipartite sequential LF scenario where all superobservers only have binary choices as described in Sec. \ref{par:LF_construction} and where the final choice of the superobserver $A$ is either $1$ or $\alpha$. According to \cref{lem:LF_binary}, the LF constraints in each of the subscenarios are equivalent to KSNC constraints in the corresponding multipartite Bell scenario. In each of the sub-scenarios, the joint distribution that recovers the corresponding empirically observable $\chi$-way correlations is
\begin{align}
P_{\alpha}&\coloneqq  \nonumber\\
P(&a_1,\dots,a_{N_A-1}, a_{N_{A}-2+\alpha},b_1,\dots,b_{N_B},\dots \nonumber \\
&|x_{A_1}=\dots =x_{A_{N_A-1}}=x_{B_1}=\dots=x_{{N_B-1}}=\dots=1, \nonumber\\
&x_{A_{N_A}}=\alpha, x_{{N_B}}=\dots=2).
\end{align}
Just like in the proof of \cref{sec:32} for the (3,2)-setting bipartite LF scenario, we can here define a large joint distribution (by gluing all the $P_\alpha$'s together)
\begin{align}
    P_{\rm all}\coloneqq 
    \frac{P_{2}P_3\dots P_{\xi}}{\sum_{a_{N_A+1}}P_3\sum_{a_{N_A+2}}P_4\dots\sum_{a_{N_{A}-2+\xi}}P_{\xi}}.
\end{align}
To prove that $P_{\rm all}$ recovers all empirical $\chi$-way correlations in this scenario, it suffices to notice that, except for $a_{N_A-2+\alpha}$, none of the outcomes obtained by any friend is in the future light cone of $x_{A_{N_A}}$ and thus, according to Local Agency, all $P_{\alpha}$s agree on their overlap, namely
\begin{align}
\sum_{a_{N_A}}P_2=\sum_{a_{N_A+1}}P_3=\dots=\sum_{a_{N_{A}-2+\xi}}P_{\xi},
\end{align}
which is analogous to \cref{eq:bcd2232} in \cref{sec:32}. Then the rest of the proof is analogous to \cref{eq:321,eq:322,eq:323} in \cref{sec:32}. 

Furthermore, following an analogous proof as in \cref{sec:equi}, we can show that the existence of the $P_{\alpha}$s and that they recover empirically observable $\chi$-way correlations as their marginals while agreeing on their overlaps are all the LF constraints on this multipartite sequential LF scenario. For such a multipartite sequential LF scenario where one of the superobservers, e.g., $\Omega$, has a final choice with cardinality $\xi$, one can again define a corresponding multipartite Bell scenario with $\chi$ agents $A,B,\ldots$ whose respective setting choices are $x,y,\dots$, corresponding to all measurements performed on the respective superobservers $A,B,\ldots$ in the LF scenario. Then, we reach the conclusion that the LF polytope is the same as the Bell polytope in the corresponding Bell scenario here.

Analogously to the minimal LF scenario discussed in \cref{sec:minimal}, we can also consider the special case where one superobserver has no friend and only has a choice (not necessarily binary) out of several measurements that she herself can perform. This scenario is equivalent to the special case of the scenario described above, where the friend of superobserver $A$ performs no measurements at all before superobserver $A$ makes her final choice (with potentially many options).
The proof that the LF constraints in this scenario where one of the superobservers has no friend are equivalent to the Bell inequalities on the corresponding Bell scenario is analogous to the one in \cref{sec:minimal} for the minimal scenario, and we will not repeat here.

Thus, we can have one of the superobservers in the multipartite sequential LF scenario defined in Sec. \ref{par:LF_construction} having a high-cardinality choice for their final choice and potentially losing their friend, and still have the equivalence between the LF constraints and the KSNC constraints in the corresponding Bell scenario. Let us call these generalizations, where at most one of the superobservers has a high-cardinality final choice and/or has no friend the \emph{multipartite sequential LF scenario with binary choices and one maverick}. We arrive at the following theorem:

\begin{theorem}
\label{thm:maverick}
    The LF constraints on the empirical correlations obtained in a multipartite sequential LF scenario with binary choices and one maverick, as defined above, are equivalent to the constraints from KSNC in the corresponding multipartite Bell scenario for any integer $\chi,N_A,N_B,\dots,>1$. That is, the LF polytope equals the Bell polytope in these scenarios.
\end{theorem}

However, we \emph{cannot} have more than one maverick, i.e., more than one superobserver losing their friends or having high-cardinality final choices. In such a scenario, one will arrive at an analogous situation to the 3,3-setting bipartite LF scenario, where the LF polytope is strictly bigger than the corresponding Bell polytope. 

A special case of Theorem~\ref{thm:maverick} (where one considers multipartite but not sequential LF scenarios) was previously proven in Ref.~\cite{haddara2024local}.

\section{Details for the operations on Alice's side in the LF 5-cycle protocol} 
\label{app:5Alice}

Here we provide the details for the operations on Alice's side in the LF 5-cycle protocol, which is the same as the one in Ref.~\cite{walleghem2023extended}.

By the assumption of Universality of Unitarity (cf. Appendix~\ref{app:compare}), each friend of Alice can be modeled as a quantum system and indeed, as a qubit, since only two orthogonal (coarse-grained) states of the friend are relevant to the argument. If each measurement $A_i$ is performed in a sufficiently isolated environment, Universality of Unitarity further implies that it can be described by a unitary. In particular, if the measurements are performed in a manner that has minimal disturbance on system $S$, then each measurement $A_i$ can be modeled by a CNOT gate $U_{i}$ between the system $S$ and the agent $F_i$, namely,
\begin{align}
\label{eb_uni}
   &U_{i}:={\rm CNOT}_{ SF_i} \\ \nonumber
   =& \ketbra{v_i^{\perp}}_{ S}\otimes\mathbf{1}_{ F_i}+ \ketbra{v_i}_{ S}\otimes ( \ketbra{1}{0} + \ketbra{0}{1})_{ F_i},
\end{align}
where $\ket{0}_{F_i}$ is the coarse-grained state of $F_i$ having observed outcome 0 and $\ket{1}_{F_i}$ is the coarse-grained state of $F_i$ having observed outcome 1. We follow a standard abuse of notation by letting $\ket{0}_{F_i}$ also denote the coarse-grained `ready' state of the observer---that is, the initial state of each observer (prior to the measurement).

From AOE, each friend $F_i$ observes an absolute outcome, denoted $f_i$, during their measurement $A_i$. We label the $\ketbra{v_i}$ outcome by $f_i=1$ and the other outcome by $f_i=0$. 

The superobserver Alice is assumed to have perfect quantum control over the joint system of $F_1$, $F_2$, $F_3$ and $S$. In particular, we assume that she has the (extreme) technological capabilities to implement the inverse of the unitary operations constituting the first three friends' measurements. 

The exact sequence of these quantum operations is shown in the schematic representation of the protocol in Fig.~\ref{fig_setusp5cycle}.
Specifically, $A_1$ is the first measurement on $S$, followed by $A_2$, after which the superobserver undoes the first measurement by applying the unitary $U_{{1}}^{\dagger}$. Then $A_3$ is performed, followed by the undoing $U_{{2}}^{\dagger}$, followed by $A_4$, followed by the undoing $U_{{3}}^{\dagger}$, followed by $A_5$. (So far, the protocol on Alice's side is exactly the same as the protocol in Ref.~\cite{walleghem2023extended}.) 

In addition, the superobserver Alice has a binary choice $x_i$ right before she implements the inverse unitary $U^{\dagger}_i$ for $i=1,2,3$. If $x_i=1$, Alice opens the door to $F_i$'s and $F_{i+1}$'s labs and reveals their respective measurement outcome $f_i$ and $f_{i+1}$ to the world; if $x_i=2$, she simply continues the protocol. When $x_1=x_2=x_3=2$, the $f_4$ and $f_5$ outcomes are revealed to the world at the end of the experiment.

\section{Proof for the empirical correlations in the LF Peres-Mermin protocol}
\label{app:PM}

The initial entangled state defined in \cref{eq:initialgiant} can be constructed by
\begin{align}
\label{eb_construini}
   \ket{\psi_0}_{ST} = \left( \prod_{j=1}^6 {\rm CNOT}_{ ST_j} \right) \ket{\psi} _{S} \bigotimes_{j=1}^{6} \ket{0}_{T_j}
\end{align}
where 
\begin{align}
\label{eb_creatST}
   & {\rm CNOT}_{ ST_j} \coloneqq  \\&(\Pi_{j^{\perp}})_S\otimes\mathbf{1}_{T_j}+ (\Pi_j)_{ S}\otimes ( \ketbra{1}{0} + \ketbra{0}{1})_{T_j}, \nonumber
\end{align}
such that 
\begin{align}
    & {\rm CNOT}_{ ST_j} \ket{\psi} _{S} \ket{0}_{T_j} \\
    =& (\Pi_{j^{\perp}})_S \ket{\psi} _{S}\otimes\ket{0}_{T_j}+ (\Pi_{j})_S\ket{\psi} _{S}\otimes \ket{1}_{T_j} \nonumber \\
    = & \sum_{k_j=0,1} \left(\Pi_{j^{k_j}}  \ket{\psi} \right)_{S}\otimes\ket{k_j}_{T_j},
\end{align}
where, $\Pi_{j^{0}}=\Pi_{j^{\perp}}$ and $\Pi_{j^{1}}=\Pi_{j}$.

We have the following useful mathematical facts: for any state $\ket{\psi}$,
\begin{subequations}
    \begin{align}
        &\bra{+}_{T_j}{\rm CNOT}_{ ST_j}\ket{\psi} _{S}  \ket{0}_{T_j} = \frac{1}{\sqrt{2}} \ket{\psi} _{S}, \label{eb_+cnot}\\
        &\bra{k}_{T_j}{\rm CNOT}_{ ST_j}\ket{\psi} _{S} \ket{0}_{T_j}=\left(\Pi_{j^{k}}\right)_{ S}\ket{\psi} _{S}, \forall k=0,1 \label{eb_1cnot}
    \end{align}
\end{subequations}
which leads to  
    \begin{align}
        &\bra{+}_{T_j}\ket{\psi_0}_{ST} \label{eb_+cnot0} 
        =\frac{1}{\sqrt{2}} \left( \prod_{i\neq j} {\rm CNOT}_{ ST_i} \right) \ket{\psi} _{S} \bigotimes_{i\neq j} \ket{0}_{T_i}, 
\end{align}
and
\begin{align}
        \bra{k}_{T_j}\ket{\psi_0}_{ST}   \label{eb_1cnot0} 
        =  \left( \prod_{i\neq j} {\rm CNOT}_{ ST_i} \right)  \left(\Pi_{j} \right.&\left.\ket{\psi} \right)_{S} \bigotimes_{i\neq j} \ket{0}_{T_i}, \nonumber\\ 
        &\forall k=0,1 
\end{align}
where we used the property that any $\bra{\phi}_{T_j}$ commutes with any ${\rm CNOT}_{ ST_i},\forall i\neq j$.

\medskip

Now, let us prove \cref{eq:empAaalpha,eq:empBbbeta,eq:empCcgamma}. Use \cref{eq:empAaalpha} as an example,  i.e., 
\begin{align}
    \wp(f_A  f_{a} f_{\alpha} = -1  &|x_1=1, x_2=2, y_{\mathbf{j}}=2, b_{\mathbf{j}}=+)=0. \nonumber
\end{align}

Define
\begin{align}
    &\Pi_{-1}\coloneqq \\
    &\Pi_{A^{\perp}}\Pi_{a}\Pi_{\alpha}+ \Pi_{A}\Pi_{a^{\perp}}\Pi_{\alpha} +\Pi_{A}\Pi_{a}\Pi_{\alpha^{\perp}} +\Pi_{A^{\perp}}\Pi_{a^{\perp}}\Pi_{\alpha^{\perp}}. \nonumber
\end{align}
From \cref{sec:recap_PM}, we know that the Born rule predicts that 
\begin{align}
    \Pi_{-1} \ket{\phi}=0, \forall \ket{\phi}. \label{eq:minus1}
\end{align}

The left-hand side of \cref{eq:empAaalpha} is
\begin{align}
    \Tr [ \left( \Pi_{-1}\right) _S \bigotimes_{j=1}^{6} \ketbra{+}_{T_j}  \ketbra{\psi_0}_{ST}].
\end{align}
Using the cyclicity of trace, and that $\ket{+}_{T_j}$ commutes with $(\Pi_{-1})_S$, this equals 
\begin{align}
    \Tr [    \left( \Pi_{-1}\right)_S \bigotimes_{j=1}^{6}  \bra{+}_{T_j}  \ket{\psi_0}_{ST}\bra{\psi_0}_{ST} \ket{+}_{T_j}].
\end{align}
Using \cref{eb_+cnot0}, this further equals
\begin{align}
    \frac{1}{2^6} \Tr [ \left( \Pi_{-1} \right) _S   \ketbra{\psi} _{S} ] = 0,
\end{align}
according to \cref{eq:minus1}. Thus, we arrive at \cref{eq:empAaalpha}. The proof for \cref{eq:empBbbeta,eq:empCcgamma} follows analogously.

\medskip

Next, we will prove \cref{eq:ABC_step1,eq:empabc,eq:empalphabetagamma}. Use \cref{eq:ABC_step1} as an example, i.e., 
\begin{align}
    \wp(b_A  b_B f_C = -1 ,&b_{\mathbf{j}-\{A,B\}}=+ \nonumber \\
    |x_1=x_2=2, &y_A=y_B=1, y_{\mathbf{j}-\{A,B\}}=2)=0.     \nonumber
\end{align} 

Define 
\begin{align}
    k_{-1}=\{\{0,0,0\},\{1,1,0\},\{0,1,1\},\{1,0,1\}\}.
\end{align}
Then the left hand side of \cref{eq:ABC_step1} is
\begin{align}
   \sum_{\{k_C, k_B, k_A\} \in k_{-1}} \mkern-25mu & \mathrm{Tr}   \Big[ \big(\Pi_{C^{k_C}} 
     V_{\beta}^{\dagger}V_{b}^{\dagger}V_{B}^{\dagger} V_{\beta}V_{b}V_{B} V_{\alpha}^{\dagger}V_{a}^{\dagger}V_{A}^{\dagger} \nonumber \\ &V_{\alpha}V_{a}V_{A}
    \big)_S \ketbra{k_A}_{T_A} \ketbra{k_B}_{T_B}  \nonumber \\
    &\bigotimes_{j\neq A,B} \ketbra{+}_{T_j}  \ketbra{\psi_0}_{ST} \Big]  \nonumber \\
    = \sum_{\{k_C, k_B, k_A\} \in k_{-1}} \mkern-25mu  &\mathrm{Tr}  \Big[ \left(\Pi_{C^{k_C}} \right)_S 
    \ketbra{k_B}_{T_B} \ketbra{k_A}_{T_A} \nonumber \\
    & \bigotimes_{j\neq A,B} \ketbra{+}_{T_j}  \ketbra{\psi_0}_{ST} \Big] 
\end{align}
where $V_i$ is the isometry modelling friend $i$'s measurement, which corresponds to the $U_i$ defined in the main text, while $V_i^{\dagger}$ is defined to satisfy $V_i^{\dagger}V_i=\mathbf{1}$. 

Using the cyclicity of trace, and that for any $j$, we have $\ket{+}_{T_j}$ and $\ket{k_j}_{T_j}$ commute with $(\Pi_C)_S$, the above expression equals
\begin{align}
\sum_{\{k_C, k_B, k_A\} \in k_{-1}} \mkern-25mu & \mathrm{Tr}
     \Big[  \left(\Pi_{C^{k_C}} \right)_S 
     \bigotimes_{j\neq A,B}  \bra{k_B}_{T_B} \bra{k_A}_{T_A} \bra{+}_{T_j}  \nonumber \\
     & \ketbra{\psi_0}_{ST} \ket{k_B}_{T_B} \ket{k_A}_{T_A} \ket{+}_{T_j}\Big] 
\end{align}

Using \cref{eb_+cnot0} and \cref{eb_1cnot0}, this further equals
\begin{align}
     \!\!\!\!\!\!\!\!\!\!  \sum_{\{k_C, k_B, k_A\}\in k_{-1}} \!\!\!\!\!\!\! \frac{1}{2^4}  \Tr \big[ \left( \Pi_{C^{k_C}} \Pi_{B^{k_B}} \Pi_{A^{k_A}}\ketbra{\psi}\right) _S \big] = 0, 
\end{align}
according to the Born rule prediction given in \cref{sec:recap_PM}. Thus, we arrive at \cref{eq:empAaalpha}. The proof for \cref{eq:empabc,eq:empalphabetagamma} follows analogously. 

\medskip

Finally, we prove \cref{eq:equal1} and \cref{eq:equal2}. Use \cref{eq:equal1} as an example, i.e., 
\begin{align}
    \text{for } j\in\{A,a,\alpha\}, \, \wp(f_j \neq b_j| x_1=1,y_j=1,x_2, y_{\mathbf{j}-\{j\}})=0. \nonumber
\end{align}

Consider the case where $j=A$ in the above equation. The left-hand side of the equation is then
\begin{align}
    \Tr [ \left( \Pi_{A^k}\right) _S \ketbra{k'}_{T_A}\bigotimes_{j \neq A} \ketbra{+}_{T_j}  \ketbra{\psi_0}_{ST}],
\end{align}
where $k\neq k'$.

Using the cyclicity of trace, and that for any $j$, we have $\ket{+}_{T_j}$ and $\ket{k'}_{T_j}$ commute with $(\Pi_A)_S$, the above expression equals
\begin{align}
    \Tr [    \left( \Pi_{A^k}\right)_S \bra{k'}_{T_A} \bigotimes_{j\neq A}  \bra{+}_{T_j}  \ket{\psi_0}_{ST}\bra{\psi_0}_{ST} \ket{+}_{T_j}\ket{k'}_{T_A}].
\end{align}
Using \cref{eb_+cnot0} and \cref{eb_1cnot0}, this further equals
\begin{align}
    \Tr \big[    \left( \Pi_{A^k}\Pi_{A^{k'}}\ketbra{\psi}\right)_S \big] = 0,
\end{align}
since $k\neq k'$. The proof for \cref{eq:equal1} with $j=a$ or $j={\alpha}$ follows analogously, Similarly, we can proof \cref{eq:equal2}.

\section{Proof of Lemma \ref{lem:LF_binary}}
\label{app:LF_binary}

We restate the lemma here:
\lembinary*

\begin{proof}
As mentioned in Appendix \ref{par:LF_construction},
in the runs of the multipartite sequential LF scenario with
binary choices where all the superobservers have all of their choices being 1 except that their final choice is 2, all the $\sum_{\Omega}N_{\Omega}$ measurements in this scenario are performed by the friends. According to AOE, for those runs, there exists a distribution
\begin{align}
P_{\rm all} \coloneqq P(&a_1,\dots,a_{N_A},b_1,\dots,b_{N_B},\dots \\
&|x_{A_1}{=}\dots {=}x_{A_{N_A}-1}{=}x_{B_1}{=}\dots{=}x_{{N_B}-1}{=}2). \nonumber 
\end{align}

For different choices of the superobservers, one can obtain different $\chi$-way empirical correlations. In particular, a 
$\chi$-way empirical correlations $\wp(a_{k_A},b_{k_B},\dots)$ (where $k_{\Omega}\in \{1,2,\dots,k_{\Omega}\}$ for $\Omega=A,B,\dots$)  is obtained during the runs when for any superobserver $\Omega$, all of their choices before the $k_{\Omega}$-th choice is 2 (i.e., $x_{\Omega_1}=x_{\Omega_2}=\dots=x_{\Omega_{k_{\Omega}-1}}=2$) while their $k_{\Omega}$-th choice is 1 if $k_{\Omega}<N_{\Omega}$ (else, they don't have the $k_{\Omega}$-th choice). Thus, in terms of choices, the only differences between the runs for $\wp(a_{k_A},b_{k_B},\dots)$ and the runs for $P_{\rm all}$ are the ones after the $(k_{\Omega}-1)$-th choice of any superobserver $\Omega$ when $k_{\Omega}<N_{\Omega}$.

The outcome $\omega_{k_{\Omega}}$ (for the $\Omega_{k_{\Omega}}$ measurement) is not in the future light cone of any choice of $\Omega$ occurring after the $(k_\Omega-1)$-th choice; moreover, it is not in the future light cone of any other superobserver's any choice (since they are space-like separated). Therefore, according to Local Agency, we have
\begin{align}
    &\wp(a_{k_A},b_{k_B},\dots)  \\
    =& 
    P(a_{k_A},b_{k_B},\dots \nonumber \\
    &| x_{A_1}=\dots =x_{A_{N_A}-1}=x_{B_1}=\dots=x_{{N_B}-1}=2), \nonumber 
\end{align}
meaning that $\wp(a_{k_A},b_{k_B},\dots)$ is a marginal of $P_{\rm all}$.

Thus, according to Fine's theorem (and the fact that no-signaling is satisfied in this scenario)~\cite{fine1982hidden,abramsky2011sheaf}, the set of empirical correlations that can be obtained in this scenario, i.e.,$\{\wp(a_{k_A},b_{k_B},\dots)\}_{k_{\Omega}\in \{1,2,\dots,N_\Omega\}, \forall \Omega=A,B,\dots}$ must satisfy all the KSNC constraints on the corresponding Bell scenario. Furthermore, following \cref{sec:Bell_inside_LF}, we can show that the LF polytope here must also contain the Bell polytope in the corresponding Bell scenario.  Thus, the LF polytope coincides with the Bell polytope here.
\end{proof}

\section{Proof of \cref{eq:Aaalpha,eq:Bbbeta,eq:Ccgamma,eq:ABC,eq:abc,eq:alphabetagamma} in Section \ref{sec:Peres-Mermin}} \label{app:proof_PeresMermin_eqs}
We begin with \cref{eq:Aaalpha,eq:Bbbeta,eq:Ccgamma}.
Consider the following empirical predictions given by the Born rule for the cases where Bob measures all $T_j$ systems in the $\pm$ basis and obtains outcomes $+$:
\begin{subequations}
\begin{align}
    \wp(f_A  f_{a} f_{\alpha} = -1  &|x_1=1, \label{eq:empAaalpha} \\
    &x_2=2, y_{\mathbf{j}}=2, b_{\mathbf{j}}=+)=0, \nonumber\\
    \wp(f_B  f_b f_{\beta} = -1  &|x_2=1, \label{eq:empBbbeta} \\
    & x_1=2, y_{\mathbf{j}}=2, b_{\mathbf{j}}=+)=0,  \nonumber\\
    \wp(f_C  f_c f_{\gamma} = -1  &|x_1=x_2=2, y_{\mathbf{j}}=2, b_{\mathbf{j}}=+)=0.\label{eq:empCcgamma}
\end{align}
\end{subequations} 
The detailed calculation using the Born rule showing these three empirical correlations are provided in Appendix~\ref{app:PM}.

\cref{eq:empCcgamma} gives \cref{eq:Ccgamma} immediately. Furthermore, since $\{f_A  f_{a} f_{\alpha}\}$ is not in the future light cone of $x_1$, $\{f_B  f_b f_{\beta}\}$ applying Local Agency to \cref{eq:empAaalpha}, we arrive at \cref{eq:Aaalpha}; similarly, since $\{f_B  f_b f_{\beta}\}$ is not in the future light cone of $x_2$, applying Local Agency to \cref{eq:empBbbeta}, we arrive at \cref{eq:Bbbeta}.

\medskip

Now we turn to prove \cref{eq:ABC,eq:ABC,eq:alphabetagamma}.

First, we note that for any $j\in\mathbf{j}$, whenever Alice reveals $f_j$ and Bob measures $T_j$ in the computational basis, we must empirically have $f_j=b_j$ according to the Born rule. That is, 
\begin{subequations}
   \begin{align}
    \text{for } &j\in\{A,a,\alpha\}, \nonumber \\
    &\wp(f_j \neq b_j| x_1=1,y_j=1,x_2, y_{\mathbf{j}-\{j\}})=0,  \label{eq:equal1}\\
    \text{for } &j\in\{B,b,\beta\}, \nonumber \\
    &\wp(f_j \neq b_j| x_2=1,y_j=1,x_1=2, y_{\mathbf{j}-\{j\}})=0,  \label{eq:equal2}
\end{align} 
\end{subequations}
where $\mathbf{j}-\{j\}$ denotes set subtraction of $\{j\}$ from $\mathbf{j}$. The proof is given in Appendix~\ref{app:PM}.

Since $x_1$ is in the future light cone of $f_A,f_a,f_{\alpha}$,  $x_2$ is in the future light cone of $f_B,f_b,f_{\beta}$, both $x_1$ and $x_2$ are space-like separated from $b_j$, according to Local Agency, we further have
\begin{align}
\label{eb_encode}
    P(f_j \neq b_j| x_1=x_2=2, y_j=1,y_{\mathbf{j}-\{j\}})=0.
\end{align}

Next, consider the runs of the experiment where the superobserver Alice has $x_1=x_2=2$, while Bob measures $T_A$ and $T_B$ in the computational basis and measures the rest in the $\pm$ basis, obtaining outcome $+$, i.e., $y_A=y_B=1$ while $y_{\mathbf{j}-\{A,B\}}=2$ and $b_{\mathbf{j}-\{A,B\}}=+$.

As shown in Appendix~\ref{app:PM}, the Born rule gives the empirical prediction
\begin{align}
        \wp(b_A  b_B f_C = -1 ,&b_{\mathbf{j}-\{A,B\}}=+ |x_1=x_2=2, \nonumber \\ &y_A=y_B=1, y_{\mathbf{j}-\{A,B\}}=2)=0.     \label{eq:ABC_step1}
\end{align} 
Together with \cref{eb_encode}, we can replace $b_A,b_B$ in the above equation with $f_A,f_B$ and have
\begin{align}
        \!\! P(f_A  f_B f_C = -1, &b_{\mathbf{j}-\{A,B\}}=+ |x_1=x_2=2, \nonumber \\ &y_A=y_B=1, y_{\mathbf{j}-\{A,B\}}=2)=0.     \label{eq:ABC_step2}
\end{align} 

Since none of $f_A,f_B,f_C$ and $b_{I-A-B}$ is in the future light cone of $y_A,y_B$, according to Local Agency, \cref{eq:ABC_step2} gives
\begin{equation} \begin{split}
    \label{eq:ABC_step3}
      \!\!  P(f_A  f_B f_C = -1, b_{\mathbf{j}-\{A,B\}}=+  |x_1=x_2=2, y_{\mathbf{j}}=2)=0,
\end{split} 
\end{equation}
which implies 
\begin{equation} \begin{split}
    \label{eq:ABC_step4}
       P(f_A  f_B f_C = -1, b_{\mathbf{j}}=+  |x_1=x_2=2, y_{\mathbf{j}}=2)=0.
\end{split} 
\end{equation}
Together with \cref{eq:x1y1qplus}, we then arrive at \cref{eq:ABC}.

Similarly,  by starting from the following empirical predictions given by the Born rule prediction (shown in Appendix~\ref{app:PM}).
\begin{subequations}
\begin{align}
        \wp(b_a  b_b f_c = -1 ,&b_{\mathbf{j}-\{a,b\}}=+ |x_1=x_2=2, \nonumber \\ 
        &y_a=y_b=1, y_{\mathbf{j}-\{a,b\}}=2)=0, \label{eq:empabc}\\
        \wp(b_{\alpha}  b_{\beta} f_{\gamma} = 1 ,&b_{\mathbf{j}-\{\alpha,\beta\}}=+ |x_1=x_2=2,\nonumber\\ &y_{\alpha}=y_{\beta}=1, y_{\mathbf{j}-\{\alpha,\beta\}}=2)=0, \label{eq:empalphabetagamma}
\end{align}
\end{subequations}
together with \cref{eb_encode}, \cref{eq:x1y1qplus} and Local Agency, we arrive at \cref{eq:abc,eq:alphabetagamma}, and thus complete the proof.

\section{Proof of Theorem \ref{th:LFpossKSNC}} \label{app:proof_LFpossKSNC}

We restate the theorem here:
\LFpossKSNC*

We prove the theorem by explicitly constructing an LF scenario for any given quantum possibilistic KS contextual model and then use that scenario to prove the LF no-go theorem.

\vspace{0.3cm}
\begin{center}
\textbf{The LF protocol from a possibilistic KS contextual model} \label{app:poss_KSC_to_LF}
\end{center}
\medskip

Consider any given quantum possibilistic KS contextual model with $n$ measurements $M_1,\ldots,M_n$ on system $S$ (potentially in a specific state $\ket{\psi}_{\!S}$) and $K$ maximal contexts $C_1,\ldots,C_K$.\footnote{The system $S$ would be the joint system of all systems that may be measured by some measurements in the model. For example, in Hardy's proof of KSNC~\cite{hardy1993nonlocality}, we take the system $S$ to be a 4-dimensional qudit, representing a bipartite qubit state.} 
The empirical correlations for this KS model are $\{\wp(m_{C_k}|C_k)\}_{k\in\{1,2,\dots,K\}}$ where $m_{C_k}$ is an outcome vector of all measurements in context $C_k$.  Being a possibilistic KS contextual model, we can assume that they satisfy
\begin{align}
        \wp(m_{C_1} \in \alpha_{C_1} | C_1) &> 0, \label{eq:app_KSC_postselection} \\
        \wp(m_{C_2} \in \alpha_{C_2} | C_2) &= 0, \label{eq:app_KSC_impossible_C2} \\
        \wp(m_{C_3} \in \alpha_{C_3} | C_3) &= 0, \label{eq:app_KSC_impossible_C3} \\
        &\vdots \\
        \wp(m_{C_K} \in \alpha_{C_K} | C_K) &= 0 \label{eq:app_KSC_impossible_CN},
    \end{align} for some outcome vectors \emph{sets} $\alpha_{C_k}$ for each context $C_k$. For example, in the Peres-Mermin model, $\alpha_{A,B,C}$ could be the set of outcome vectors where the product of the outcomes of each vector is -1.

The LF protocol that we will introduce is analogous to the LF 5-cycle protocol in \cref{sec:5cycleprotocol} and the LF Peres-Mermin protocol in \cref{sec:PeresMerminprotocol}. For simplicity, we assume that all measurements $M_1,\dots,M_n$ in the KS contextual model have binary outcomes 0 and 1. The protocol can be generalized to have measurements with higher cardinality outcomes by replacing the qubits for Bobs in the protocol below to higher-dimensional qudits.

In the LF protocol for the possibilistic KS contextual model, we have $n$ Bobs, denoted $B_1,B_2,\dots,B_n$, respectively,\footnote{Some Bobs are in fact redundant, but for the simplicity in describing the protocol, we keep all of them.} and one Alice, who has $n$ friends. They share a joint system $ST$ where $S$ is the system in the KS model and $T$ includes $n$ qubits (labeled $T_1,T_2,\dots,T_n$, respectively). Alice and her friends have $S$, while each Bob has one of the qubits in $T$.

We let friends $F_1,\ldots,F_n$ perform measurements $M_1,\ldots,M_n$ sequentially on $S$ at certain points in the protocol, obtaining outcomes $f_1,\dots,f_n$, respectively. We denote the set of friends performing the measurements of context $C_k$ by $F_{C_k}$, who obtain the outcome vector $f_{C_k}$ from their measurements. We assume that each measurement is done in a minimally disturbing way in an isolated lab. Under the assumption of the Universality of Unitary Dynamics, these measurements can be modeled as unitaries $U_1,U_2,\dots,U_n$, respectively.

The superobserver Alice undoes some of the friends' measurements at certain points. She also has binary choices at various points to decide whether or not to reveal the friends' outcomes; they are denoted $x_1,\ldots,x_n$ for deciding whether to reveal the respective outcomes $f_1,\dots,f_n$. 
The vector of Alice's choices to be set to 1 to reveal outcomes $f_{C_k}$ on Alice's side is denoted by $x_{C_k}$, and the vector $x_{\rm rest}$ denotes all other choices on Alice's side. The set of all choices (o.e., the choice vector) on Alice's side is denoted by $x_\text{All}$.

Each Bob $B_1,\ldots,B_n$ has a binary measurement choice. Their choices are labeled as $y_1,y_2,\dots,y_n$ and their outcomes are labeled as $b_1,b_2,\dots,b_n$, respectively. For each Bob, when his choice is 1, he measures his qubit in the computational basis; when his choice is 2, he measures his qubit in the $\pm$ basis. The set of all choices (i.e., the choice vector) on Bob's side is denoted by $y_\text{All}$. Similarly, $b_\text{All}$ denotes the outcome vector of all measurements on the Bobs side. 

The detailed steps of the protocol are as follows.

\vspace{0.5cm}

\textit{1. Initialization.} Prepare the joint system $ST$ in the following state\footnote{Alternatively, this entangled state can be constructed in a measure-prepare fashion as in Ref.~\cite{frauchiger2018quantum} around which the thought experiment could alternatively be constructed. }
\begin{align}
\label{eb_construinige}
   \ket{\psi_0}_{ST} = \left( \prod_{j=1}^n {\rm CNOT}_{ ST_j} \right) \ket{\psi} _{S} \bigotimes_{j=1}^{n} \ket{0}_{T_j}
\end{align}
with 
\begin{align}
   & {\rm CNOT}_{ ST_j} \coloneqq (\Pi_{j^{\perp}})_S\otimes\mathbf{1}_{T_j}+ (\Pi_j)_{ S}\otimes ( \ketbra{1}{0} + \ketbra{0}{1})_{T_j}, \nonumber
\end{align}
where $\Pi_{j^{\perp}}$ is the projector corresponding to the outcome of measurement $M_j$ being 0 and $\Pi_j$ is the one for outcome 1.

When the KS contextual model is state-independent, $\ket{\psi}$ can be any arbitrary state; else $\ket{\psi}$ is the initial state of the quantum realization of the KS contextual model we use. 

The system $S$ is sent to Alice and her friends while $T_1,T_2,\dots,T_n$ are sent to the Bobs $B_1,B_2,\dots,B_n$ respectively. 

\textit{2. Context $C_1$.}
We assume without loss of generality that $C_1,C_2,\ldots,C_K$ are ordered such that $C_k$ and $C_{k+1}$ have maximal number of overlap measurements $C_k \cap C_{k+1}$ for any $k\in \{1,2,\dots,K\}$. (By ordering the contexts in this way, the LF protocol we construct here is more analogous to the LF 5-cycle protocol than the LF Peres-Mermin protocol. If we order the contexts in the opposite way such that $C_k$ and $C_{k+1}$ have the minimal number of overlap measurements, then the protocol becomes more analogous to the LF Peres-Mermin protocol.) 

The friends in $F_{C_1}$ perform the measurements on $S$ in $C_1$, obtaining outcome $f_{C_1}$. 

\textit{3. Context $C_2$.} 
Next, we consider the maximal context $C_2$. If a measurement in $C_2$ equals a measurement in $C_1$, this measurement is not performed anymore. If not, it will be performed. 

Before performing those measurements, the superobserver Alice first makes the choices $x_{C_1-(C_1\cap C_2)}$. That is, for each measurement that is in $C_1$ but not in $C_2$ (which is any measurement in $C_1$ not compatible with at least one measurement in $C_2$ since $C_1$ and $C_2$ are maximum contexts), she makes a binary choice. If the choice is 1, she reveals the outcome while otherwise, she does not reveal it. 

After Alice has made all the choices, she reverses all the measurements in $C_1-(C_1\cap C_2)$  by applying the respective inverse unitaries. 

Then, those friends in $F_{C_2}$ but not in $F_{C_{1}}$ perform the measurements in $C_2-(C_1\cap C_2)$, i.e., the ones in $C_1$ that have not been performed earlier.

\textit{4. Contexts $C_3,\ldots,C_K$.} 
Next, context $C_3$ is considered. The procedure for the subsequent contexts $C_4,\ldots,C_K$ is analogous, so assume now we have just followed the procedure for context $C_{k-1}$.

Before proceeding to context $C_k$ for $k\in\{3,4,\dots,K\}$, for each measurement that has been performed so far but has not been reversed yet, if it is not in $C_k$, Alice will make a binary choice for it. We denote the set of all these measurements as $C'_{k-1}$ since it must be a subset of $C_{k-1}$.
For each measurement in $C'_{k-1}$, if her choice is 0, she reveals the outcome of that measurement; otherwise, she does not reveal it. 
Then, for any measurements in $C'_{k-1}$, she reverses it by applying the inverse of the respective unitary.\footnote{Note that if Alice chooses to reveal a measurement $i$ and then applies $U_i^\dagger$, she will actually not reverse the measurement, i.e. not reverse the states of the involved systems and labs back to their initial state.}

Afterward, for any measurement in $C_k$ that has not been performed already, the respective friend will perform that measurement and obtain the corresponding outcome.

Such a procedure continues until all measurements in all contexts are performed. Then, Alice reveals all outcomes of the measurements that have not been reversed.

\textit{5. Bob's measurements} 
As mentioned earlier, each Bob has a binary measurement choice to measure their qubit in either the computational basis or $\pm$ basis. 

Each Bob's operation is spacelike separated from all operations of Alice and her friends, and is space-like separated from any other Bob's operation.

\vspace{0.3cm}
\begin{center}
\textbf{Proving the no-go theorem}
\end{center}

\medskip

We now prove that using the LF scenario defined above, we can translate each quantum possibilistic proof of the KSNC theorem to a proof of LF no-go theorem. The essential idea of the proof is that the correlations for all contexts $\{C_1,C_2,\dots,C_K\}$ will be encoded in a single distribution that is demanded to exist by AOE in the runs where all choices of Alice's are 1 (denoted as $x_{\rm All}=1$) so that all measurements in the KS contextual models are performed and all choices of the Bobs are 1 (denoted as $y_{\rm All}=1$) with all outcomes being + (denoted as $b_{\rm All}=1$). We use a similar notation below (which has also been used in \cref{sec:proof_LFPM}) 
 where ``a vector = a number'' is a shorthand notation for each component in the vector being equal to that number.

\vspace{0.5cm}

First, note that by AOE, the following joint probability distribution exists in the runs where all Alice's choices are 2 and all Bob's choices are also 2:
\begin{align}
    P(f_{\rm All}, b_{\rm All} |x_\text{All}=2,y_\text{All}=2).
\end{align}
Furthermore, since 
\begin{align}
    \wp(b_{\rm All}=+ |x_\text{All}=2,y_\text{All}=2)\neq 0,
\end{align}
the following distribution is also well-defined.
\begin{align}
    P(f_{\rm All}|x_\text{All}=2,y_\text{All}=2,b_\text{All}{=}+)
\end{align}

    Now, let us consider each context in the KS contextual model.
    
    Starting with $C_1$, we know that all measurement outcomes in $C_1$ can be revealed by Alice when $x_{C_1}=0$. By \cref{eb_+cnot0} and \cref{eq:app_KSC_postselection}, we have     \begin{equation}
        \wp(f_{C_1} \in \alpha_{C_1} | x_{C_1}{=}1,x_{\rm rest}{=}2, y_{\text{All}}=2,b_\text{All}{=}+)>  0,
    \end{equation}
    where $x_{\rm rest}$ refers to all choices of Alice's that have not been specified to be 0.   Since none of the measurements in $C_1$ performed by agents in $F_{C_1}$ takes place in the future light cone of any choices in $x_{C_1}$, we obtain by Local Agency that
    \begin{align}
        &P(f_{C_1}\in \alpha_{C_1}| x_{{\text{All}}}{=}2, y_\text{All}=2,b_\text{All}{=}+) \nonumber \\
        =&\wp(f_{C_1} \in \alpha_{C_1} | x_{C_1}{=}1,x_{\rm rest}{=}2, y_{\text{All}}=2,b_\text{All}{=}+)  >  0. \label{eq:app_LF_postselection}
    \end{align}
    
    Next, we turn to $C_2$. 
    If a measurement in $C_2$ is already performed by an agent in $F_{C_1}$, then it is in $C_1 \cap C_2$, and is not reversed until all measurements in $C_2$ have been performed.
    Else, it is performed by its respective agent in $F_{C_2}$, and is also available (i.e. not reversed yet) when all of $C_2$ have been performed. 
    Therefore, in both cases, the choice to reveal is made by Alice \emph{after} all of $C_2$ have been performed.
    Thus, none of the measurements in $C_2$ performed by agents in $F_{C_2}$ takes place in the future light cone of any choices in $x_{C_2}$. As such, from  \cref{eb_+cnot0} the empirical KS correlation, \cref{eq:app_KSC_impossible_C2} is obtained and by Local Agency, we obtain
    \begin{align}
        &P(f_{C_2}\in \alpha_{C_2}| x_{{\text{All}}}{=}2,y_\text{All}=2,b_\text{All} {=} +) \nonumber \\
        =& \wp(f_{C_2} \in \alpha_{C_2} | x_{C_2}{=}1,x_{\rm rest}{=}2,y_{\text{All}}=2)
        =   0. \label{eq:app_LF_C2.2}
        \end{align}
    
    Next, we consider context $C_3$. 
    
    If all measurements in $C_3$ are performed by friends in either $F_{C_2}$ or $F_{C_3}$, i.e., if $ (C_3 \cap C_1)\subseteq C_2$, we can use the same arguments as the ones used for context $C_2$ and derive
    \begin{align}
    \label{eq:c3_1}
        P(f_{C_3}\in \alpha_3| x_{{\text{All}}}{=}2,y_\text{All}=2,b_\text{All}{=}+)=0,
    \end{align}
    by \cref{eq:app_KSC_impossible_C3} and Local Agency because no outcome in $f_{C_3}$ is in the future light cone of any choice in $x_{C_3}$. 

    However, if instead there exists a nonempty set of measurements $\mathbf{M'}$ in $C_3$ such that $\mathbf{M'}\subseteq C_1$ but $\mathbf{M'} \cap C_2=\emptyset$, i.e., each measurement in $\mathbf{M'}$ is also in $C_1$ but not in $C_2$, then these measurements in $\mathbf{M'}$ have been undone before the rest of measurements in $C_3$ are performed. Nevertheless, we can still prove \cref{eq:c3_1}.
    In particular, consider the empirical correlation in the runs where Alice reveals the outcomes of measurements in $C_3-\mathbf{M'}$  and does not reveal any outcomes for the measurements in $\mathbf{M'}$, while any Bob $B_i$ whose index $i$ corresponds to the index of measurement in $\mathbf{M'}$ chooses to measure his qubit in the computational basis (we denote the set of choices of these Bobs as $y_{\mathbf{M'}}$ and the set of their outcomes as $b_{\mathbf{M'}}$). By \eqref{eb_1cnot}, the outcome vector $(f_{C_3-\mathbf{M'}},b_{\mathbf{M'}})$ follows the empirical KS correlation \eqref{eq:app_KSC_impossible_C3}:
    \begin{align}
        \wp(&(f_{C_3-\mathbf{M'}},b_{\mathbf{M'}})\in \alpha_3  \\
        |& x_{C_3-{\mathbf{M'}}}{=}1, y_{\mathbf{M'}}=1,x_{\text{rest}}{=}2, y_\text{rest}=2,b_{\text{rest}}{=}+)=0. \nonumber
    \end{align} 
    None of $f_{C_3}$ and $b_{\mathbf{M'}}$ is in the future light cone of any choice in $x_{C_3-\mathbf{M'}}$, so by Local Agency 
    \begin{align}
        P\big(&(f_{C_3-\mathbf{M'}},b_{\mathbf{M'}})\in \alpha_3 \\ & | x_\text{All}{=}2,y_{\mathbf{M'}}{=}1,y_\text{rest}{=}2,b_{\rm rest}{=}+\big)  \nonumber\\                
        =\wp(&(f_{C_3-\mathbf{M'}},b_{\mathbf{M'}})\in \alpha_3 \nonumber \\
        &| x_{C_3-{\mathbf{M'}}}{=}1, y_{\mathbf{M'}}=1,x_{\text{rest}}{=}2, y_\text{rest}=2,b_{\text{rest}}{=}+)
\nonumber \\ =0 \nonumber
    \end{align}\label{eq:app_proof_C3_1}
    
    Following the proof earlier for \cref{eq:equal1}, we know that for any $i$,
    \begin{equation}
        \wp(f_i \neq b_i | x_i {=} 1 ,y_i=1) = 0, 
    \end{equation} 
    and by Local Agency, we then have 
    \begin{equation}
        P(f_i \neq b_i |x_i{=}2,y_i{=}1) = 0.
    \end{equation} 
    Using this in \cref{eq:app_proof_C3_1} we obtain \begin{equation}
        \label{eq:app_proof_C3_2}
         P(f_{C_3} \in \alpha_{C_3} | x_\text{All}{=}2,y_{\mathbf{M'}}{=}1,y_{\rm rest}{=}2,b_{\rm rest}{=}+) =0, 
    \end{equation} for which using Local Agency again we obtain \begin{equation}
        \label{eq:app_proof_C3_3}
            P(f_{C_3} \in \alpha_{C_3} | x_\text{All}{=}2,y_\text{All}{=}2,b_{\rm rest}{=}+) =0, 
    \end{equation} and thus also (as $b_{\mathbf{M'}}=+$ has nonzero (empirical) probability for the choices on the right hand side) \begin{equation}
        \label{eq:app_proof_C3_4} 
           P(f_{C_3} \in \alpha_{C_3} | x_\text{All}{=}2,y_\text{All}{=}2,b_\text{All}{=}+) =0
    \end{equation} as desired.
    
    Analogous reasoning yields \begin{equation}
        P(f_{C_k} \in \alpha_k | x_\text{All}=2,y_\text{All}=2,b_\text{All}{=}+) = 0
    \end{equation} for any $k\in\{2,3,\ldots,K\}$. Thus, we obtain indeed that \begin{equation}
    \begin{split}
         P(f_{C_1}\in \alpha_{C_1} | x_\text{All}=2,y_\text{All}=2,b_\text{All}{=}+) &> 0, \\
        P(f_{C_2}\in\alpha_{C_2} |x_\text{All}=2,y_\text{All}=2,b_\text{All}{=}+) &= 0, \\
        P(f_{C_3}\in\alpha_{C_3} |x_\text{All}=2,y_\text{All}=2,b_\text{All}{=}+) &= 0, \\
        &\vdots \\
        P(f_{C_K}\in\alpha_{C_K} |x_\text{All}=2,y_\text{All}=2,b_\text{All}{=}+) &= 0,
    \end{split}
    \end{equation} which cannot occur as these correlations constitute a possibilistic contextual model for which no global distribution $P(f_{\text{All}}|x_\text{All}=2,y_\text{All}=2,b_\text{All}{=}+)$ can have all of them as its marginals.

\section{No-go theorems with Commutation Irrelevance} \label{app:CI}

In \cref{sec:5-cycle,sec:Peres-Mermin}, we have turned the $5$-cycle proof and the Peres--Mermin proof of KSNC theorem to proofs of LF theorem with new LF scenarios. 
Compared to the measurements present in the KS contextuality scenarios, our corresponding LF scenarios have additional measurements (the ones by Bobs), and additional settings (Alice's choices and Bob's choices), complicating the scenario somewhat. 
If we did not have these additional measurements and settings, and only the measurements of the contextuality scenarios, one cannot construct a proof of LF theorem, but can one obtain an extended Wigner's friend no-go theorem with different assumptions?
The answer is affirmative. Specifically, instead of Local Agency, we can prove the no-go theorem with the assumption of \textit{Commutation Irrelevance} introduced in Ref.~\cite{walleghem2023extended}. This is an assumption arguably as well-motivated as the analogous assumption used in the Frauchiger-Renner no-go theorem~\cite{walleghem2023extended,schmid2023review}. 

\begin{assump}[Commutation Irrelevance] \label{assumption:Commutation_Irrelevance}
Let $M_1,M_2$ denote compatible measurements with outcomes $m_1,m_2$, where the measurements are modeled unitarily as $U_1,U_2$. If the commutation relation $U_{2}UU_{1}=U U_{2}U_{1}$ is satisfied, then Commutation Irrelevance requires that $P_{U_{2}UU_{1}}(m_1,m_2)=P_{UU_2 U_1}(m_1,m_2)$, where the former denotes the joint distribution when $U$ is performed between $M_1$ and $M_2$, and the latter denotes the joint distribution when $U$ is performed after $M_1$ and $M_2$. Similarly, for any set of measurements that are performed sequentially with unitary operations in between, the correlation over the measurement outcomes is the same as if there were no such unitaries, provided that the operations commute with the measurements in such a way that they can be moved to the future of all the measurements. For example, for three measurements $M_1,M_2,M_3$, represented by three respective unitaries $U_1,U_2,U_3$, and two operations $U$ and $U'$, if $U_1UU_2U'U_3=UU'U_1U_2U_3$, then
$P_{U_1UU_2U'U_3}(m_1,m_2,m_3)=P_{UU'U_1U_2U_3}(m_1,m_2,m_3)$.
\end{assump}

This assumption is motivated by the fact that when $P_{U_{2}UU_{1}}(m_1,m_2)$ and $P_{UU_2 U_1}(m_1,m_2)$ are both empirical correlations, operational quantum theory do predict that they are the same.

\begin{theorem}[A no-go theorem with Commutation Irrelevance ]
If a superobserver can perform arbitrary quantum operations on an observer's laboratory, quantum theory with Universality of Unitarity predicts that the assumption of Absoluteness of Observed Events and the assumption of Commutation Irrelevance cannot be simultaneously true.
\end{theorem}

In Ref.~\cite{walleghem2023extended}, we have already shown how to construct a proof of the no-go theorem with Commutation Irrelevance based on the 5-cycle scenario. In the following, we generalize the constructions there such that one can translate any possibilistic KS contextual model into a proof of the no-go theorem, which can be straightforwardly further generalized to any probabilistic KS contextual model as well. 

\begin{theorem} \label{th:CF_KSNC_equiv}
Every possibilistic proof of the failure of Kochen-Specker noncontextuality can be mathematically translated into a proof of the failure of the conjunction of Absoluteness of Observed Events and Commutation Irrelevance in a corresponding Extended Wigner's friend scenario.  
\end{theorem}

\begin{proof}
    
Given a quantum possibilistic KS contextual model, the corresponding CI protocol of it only has the following differences compared to the LF protocol in Appendix~\ref{app:proof_LFpossKSNC}: 
\begin{enumerate}
    \item There are no Bobs and system $T$;
    \item The system $S$ is prepared in the state $\ket{\psi}$ as in the KS contextual model (or an arbitrary state in case of a state-independent KS contextual model);
    \item The superobserver Alice has no choices and does not reveal any outcomes until the end of the experiment.\footnote{
One disadvantage of this protocol compared to the LF protocol in Appendix~\ref{app:proof_LFpossKSNC} is that one will not be able to derive constraints on empirical correlations that are analogous to LF inequalities. This is because the only empirical correlation in the CI protocol is over some of the outcomes in the final context $C_K$. If one wishes to derive inequalities, one could keep the superobserver Alice's choices and use the assumption of No-Superdeterminism defined in Ref.~\cite{bong2020strong} in addition to Commutation Irrelevance and Absoluteness of Observed Events. } 
\end{enumerate}

First, note that by AOE, the joint probability distribution $P(f_{\rm All})$ exists for all runs of the experiment.

    Now, let us consider each context in the KS contextual model.
    
    Starting with $C_1$. Since all measurement outcomes coexist before any of Alice's reversal operations, they are accessible to all friends in $F_{C_1}$ and need to agree with the Born rule prediction. (We here and as usual, always assume the operational adequacy of quantum theory. That is, any operationally accessible correlation can be predicted by the Born rule.)
    Thus, by
    \cref{eq:app_KSC_postselection}, $P(f_{C_1})$, which is a marginal of $P(f_{\rm All})$, satisfies 
    \begin{equation}
        P(f_{C_1} \in \alpha_{C_1})>  0.
    \end{equation}

    Next, we turn to $C_2$. 
    A measurement in $C_2$ performed by an agent in $F_{C_1}$ is not reversed by Alice when the rest of the measurements in $C_2$ are performed. As such all outcomes in $C_2$ coexist at that point right before any measurements in $C_2$ is reversed and hence, the correlation over them is accessible. Then, by \cref{eq:app_KSC_impossible_C2}, $ P(f_{C_2})$, which is a marginal of $P(f_{\rm All})$, satisfies
    \begin{equation}
        P(f_{C_2}\in \alpha_{C_2})=  0. \label{eq:app_CI_C2.2}
        \end{equation}
    
    Next, for context $C_3$. 
    
    If all measurements in $C_3$ are performed by friends in either $F_{C_2}$ or $F_{C_3}$, i.e., if $ (C_3 \cap C_1)\subseteq C_2$, we can use the same arguments as the ones used for context $C_2$ and derive
    \begin{align}
    \label{eq:c3}
        P(f_{C_3}\in \alpha_3)=0,
    \end{align}
    since all of their outcomes coexist at the point right before any measurement in $C_3$ is reversed.

    If instead, there exists a nonempty set of measurements $\mathbf{M'}$ in $C_3$ such that $\mathbf{M'}\subseteq C_1$ but $\mathbf{M'} \cap C_2=\emptyset$, i.e., each measurement in $\mathbf{M'}$ is also in $C_1$ but not in $C_2$, then these measurements in $\mathbf{M'}$ has been undone before the rest of measurements in $C_3$ are performed. Nevertheless, we can still prove \cref{eq:c3} with Commutation Irrelevance.
    This is because all operations after in between all $C_3$ measurements are effectively only the inverses of the unitaries for measurements in $\mathbf{M'}$ (all the rest cancels out each other due to $U^{\dagger}U=\mathbf{1}$). 
    Thus by Commutation Irrelevance, 
    $P(f_{C_3})$ (which is a marginal or $P(f_{\rm All})$) must be the same as the correlation over $f_{C_3}$ in an alternative protocol where $f_{C_3}$ are sequentially measured without any operations in-between and the inverses of the unitaries for measurements in $\mathbf{M'}$ are done only after all measurements in $f_{C_3}$ are performed. Then, by \eqref{eq:app_KSC_impossible_C3}, we obtain \cref{eq:c3}.

    \begin{figure}[h]
    \centering
    \includestandalone[width=0.5\textwidth]{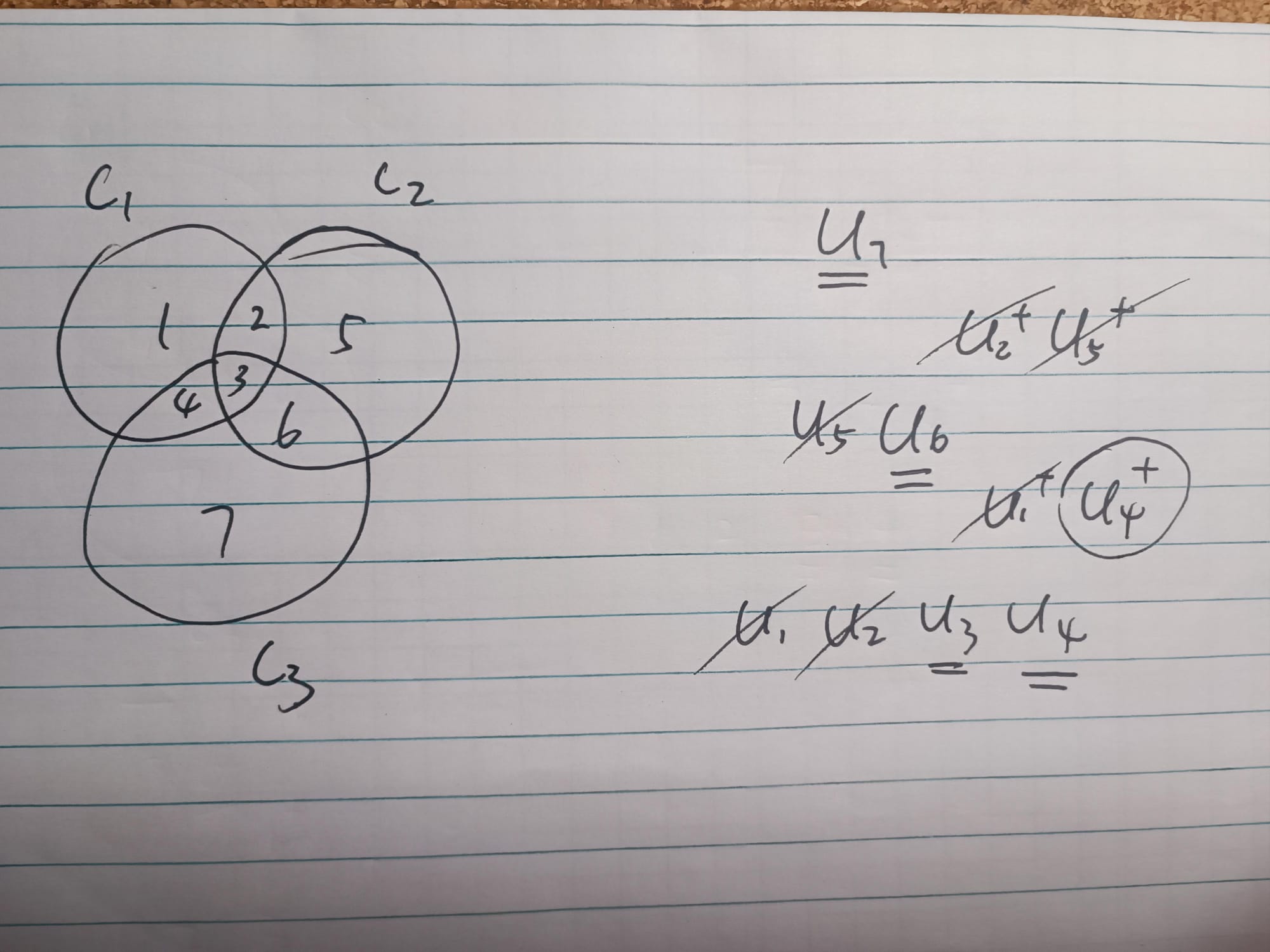}
    \caption{Left: Venn diagrams of three contexts; Right: the corresponding CI protocol.}
    \label{fig_vennCI}
    \end{figure}
    (To see the above use of Commutation Irrelevance more explicitly, it is sufficient to consider the Venn diagram in the left  of \cref{fig_vennCI} and the corresponding protocol in the right of \cref{fig_vennCI}. Here $\mathbf{M'}$ has a single element $M_4$. Since $U_7 U_2^{\dagger} U_5^{\dagger}U_5 U_6 U_1^{\dagger}U_4^{\dagger} U_1U_2U_3U_4=U_4^{\dagger}U_7 U_6U_3U_4$, the correlation over all outcomes in $f_{C_3}$ is the same as the correlation in the alternative protocol where $U_4^{\dagger}$ is performed after all measurements in $C_3$.  

    Analogous reasoning yields \begin{equation}
        P(f_{C_k} \in \alpha_k) = 0
    \end{equation} for any $k\in\{2,3,\ldots,K\}$. Thus, we obtain indeed that \begin{equation}
    \begin{split}
         P(f_{C_1}\in \alpha_{C_1} ) &> 0, \\
        P(f_{C_2}\in\alpha_{C_2} ) &= 0, \\
        P(f_{C_3}\in\alpha_{C_3}) &= 0, \\
        &\vdots \\
        P(f_{C_K}\in\alpha_{C_K}) &= 0,
    \end{split}
    \end{equation} which cannot occur as these correlations constitute a possibilistic contextual model for which no global distribution $P(f_{\rm All})$ can have all of them as its marginals.
    
\end{proof}


\end{document}

%% file: figures/5-cycle_LF_vs_noncontextuality.tex
\tikzset{every picture/.style={line width=0.75pt}} 

\begin{tikzpicture}[x=0.75pt,y=0.75pt,yscale=-1,xscale=1]

\draw  [fill={rgb, 255:red, 0; green, 0; blue, 0 }  ,fill opacity=1 ] (292.37,199.13) .. controls (292.37,196.37) and (294.61,194.13) .. (297.37,194.13) .. controls (300.14,194.13) and (302.37,196.37) .. (302.37,199.13) .. controls (302.37,201.9) and (300.14,204.13) .. (297.37,204.13) .. controls (294.61,204.13) and (292.37,201.9) .. (292.37,199.13) -- cycle ;
\draw  [fill={rgb, 255:red, 0; green, 0; blue, 0 }  ,fill opacity=1 ] (222.05,176.28) .. controls (222.05,173.52) and (224.29,171.28) .. (227.05,171.28) .. controls (229.81,171.28) and (232.05,173.52) .. (232.05,176.28) .. controls (232.05,179.05) and (229.81,181.28) .. (227.05,181.28) .. controls (224.29,181.28) and (222.05,179.05) .. (222.05,176.28) -- cycle ;
\draw  [fill={rgb, 255:red, 0; green, 0; blue, 0 }  ,fill opacity=1 ] (222.05,102.34) .. controls (222.05,99.58) and (224.29,97.34) .. (227.05,97.34) .. controls (229.81,97.34) and (232.05,99.58) .. (232.05,102.34) .. controls (232.05,105.1) and (229.81,107.34) .. (227.05,107.34) .. controls (224.29,107.34) and (222.05,105.1) .. (222.05,102.34) -- cycle ;
\draw  [fill={rgb, 255:red, 0; green, 0; blue, 0 }  ,fill opacity=1 ] (292.37,79.49) .. controls (292.37,76.73) and (294.61,74.49) .. (297.37,74.49) .. controls (300.14,74.49) and (302.37,76.73) .. (302.37,79.49) .. controls (302.37,82.25) and (300.14,84.49) .. (297.37,84.49) .. controls (294.61,84.49) and (292.37,82.25) .. (292.37,79.49) -- cycle ;
\draw  [fill={rgb, 255:red, 0; green, 0; blue, 0 }  ,fill opacity=1 ] (335.84,139.31) .. controls (335.84,136.55) and (338.08,134.31) .. (340.84,134.31) .. controls (343.6,134.31) and (345.84,136.55) .. (345.84,139.31) .. controls (345.84,142.07) and (343.6,144.31) .. (340.84,144.31) .. controls (338.08,144.31) and (335.84,142.07) .. (335.84,139.31) -- cycle ;
\draw   (340.84,139.31) -- (297.37,199.13) -- (227.05,176.28) -- (227.05,102.34) -- (297.37,79.49) -- cycle ;

\draw (287.05,206.28) node [anchor=north west][inner sep=0.75pt]   [align=left] {$\displaystyle A_{1}$};
\draw (205.05,177.28) node [anchor=north west][inner sep=0.75pt]   [align=left] {$\displaystyle A_{2}$};
\draw (200.05,99.28) node [anchor=north west][inner sep=0.75pt]   [align=left] {$\displaystyle A_{3}$};
\draw (292.05,55.28) node [anchor=north west][inner sep=0.75pt]   [align=left] {$\displaystyle A_{4}$};
\draw (349.05,128.28) node [anchor=north west][inner sep=0.75pt]   [align=left] {$\displaystyle A_{5}$};

\end{tikzpicture}

%% file: figures/5-cycle_comp_graph_LF.tex
\tikzset{every picture/.style={line width=0.75pt}} 

\begin{tikzpicture}[x=0.75pt,y=0.75pt,yscale=-1,xscale=1]

\draw  [fill={rgb, 255:red, 0; green, 0; blue, 0 }  ,fill opacity=1 ] (251.79,221) .. controls (251.79,218.24) and (253.81,216) .. (256.31,216) .. controls (258.8,216) and (260.82,218.24) .. (260.82,221) .. controls (260.82,223.76) and (258.8,226) .. (256.31,226) .. controls (253.81,226) and (251.79,223.76) .. (251.79,221) -- cycle ;
\draw  [fill={rgb, 255:red, 0; green, 0; blue, 0 }  ,fill opacity=1 ] (281.6,179) .. controls (281.6,176.24) and (283.62,174) .. (286.11,174) .. controls (288.61,174) and (290.63,176.24) .. (290.63,179) .. controls (290.63,181.76) and (288.61,184) .. (286.11,184) .. controls (283.62,184) and (281.6,181.76) .. (281.6,179) -- cycle ;
\draw  [fill={rgb, 255:red, 0; green, 0; blue, 0 }  ,fill opacity=1 ] (272.56,123) .. controls (272.56,120.24) and (274.59,118) .. (277.08,118) .. controls (279.57,118) and (281.6,120.24) .. (281.6,123) .. controls (281.6,125.76) and (279.57,128) .. (277.08,128) .. controls (274.59,128) and (272.56,125.76) .. (272.56,123) -- cycle ;
\draw  [fill={rgb, 255:red, 0; green, 0; blue, 0 }  ,fill opacity=1 ] (239.14,84) .. controls (239.14,81.24) and (241.17,79) .. (243.66,79) .. controls (246.15,79) and (248.18,81.24) .. (248.18,84) .. controls (248.18,86.76) and (246.15,89) .. (243.66,89) .. controls (241.17,89) and (239.14,86.76) .. (239.14,84) -- cycle ;
\draw  [fill={rgb, 255:red, 0; green, 0; blue, 0 }  ,fill opacity=1 ] (196.69,64) .. controls (196.69,61.24) and (198.71,59) .. (201.21,59) .. controls (203.7,59) and (205.72,61.24) .. (205.72,64) .. controls (205.72,66.76) and (203.7,69) .. (201.21,69) .. controls (198.71,69) and (196.69,66.76) .. (196.69,64) -- cycle ;
\draw    (286.11,179) -- (256.31,221) ;
\draw    (277.08,123) -- (286.11,179) ;
\draw    (243.66,84) -- (277.08,123) ;
\draw    (201.21,64) -- (242.76,86) ;
\draw  [fill={rgb, 255:red, 0; green, 0; blue, 0 }  ,fill opacity=1 ] (72.95,104) .. controls (72.95,101.24) and (74.97,99) .. (77.46,99) .. controls (79.96,99) and (81.98,101.24) .. (81.98,104) .. controls (81.98,106.76) and (79.96,109) .. (77.46,109) .. controls (74.97,109) and (72.95,106.76) .. (72.95,104) -- cycle ;
\draw  [fill={rgb, 255:red, 0; green, 0; blue, 0 }  ,fill opacity=1 ] (70.19,203) .. controls (70.19,200.24) and (72.22,198) .. (74.71,198) .. controls (77.2,198) and (79.23,200.24) .. (79.23,203) .. controls (79.23,205.76) and (77.2,208) .. (74.71,208) .. controls (72.22,208) and (70.19,205.76) .. (70.19,203) -- cycle ;
\draw    (201.21,64) -- (256.31,221) ;
\draw    (201.21,64) -- (77.46,104) ;
\draw    (256.31,221) -- (77.46,104) ;
\draw    (201.21,64) -- (74.71,203) ;
\draw    (256.31,221) -- (74.71,203) ;
\draw    (286.11,179) -- (77.46,104) ;
\draw    (277.08,123) -- (77.46,104) ;
\draw    (243.66,84) -- (77.46,104) ;
\draw    (286.11,179) -- (74.71,203) ;
\draw    (277.08,123) -- (74.71,203) ;
\draw    (243.66,84) -- (74.71,203) ;
\draw  [fill={rgb, 255:red, 0; green, 0; blue, 0 }  ,fill opacity=1 ] (592.6,221.72) .. controls (592.6,218.95) and (594.62,216.72) .. (597.11,216.72) .. controls (599.61,216.72) and (601.63,218.95) .. (601.63,221.72) .. controls (601.63,224.48) and (599.61,226.72) .. (597.11,226.72) .. controls (594.62,226.72) and (592.6,224.48) .. (592.6,221.72) -- cycle ;
\draw  [fill={rgb, 255:red, 0; green, 0; blue, 0 }  ,fill opacity=1 ] (622.4,179.72) .. controls (622.4,176.95) and (624.42,174.72) .. (626.92,174.72) .. controls (629.41,174.72) and (631.43,176.95) .. (631.43,179.72) .. controls (631.43,182.48) and (629.41,184.72) .. (626.92,184.72) .. controls (624.42,184.72) and (622.4,182.48) .. (622.4,179.72) -- cycle ;
\draw  [fill={rgb, 255:red, 0; green, 0; blue, 0 }  ,fill opacity=1 ] (613.37,123.72) .. controls (613.37,120.95) and (615.39,118.72) .. (617.89,118.72) .. controls (620.38,118.72) and (622.4,120.95) .. (622.4,123.72) .. controls (622.4,126.48) and (620.38,128.72) .. (617.89,128.72) .. controls (615.39,128.72) and (613.37,126.48) .. (613.37,123.72) -- cycle ;
\draw  [fill={rgb, 255:red, 0; green, 0; blue, 0 }  ,fill opacity=1 ] (579.95,84.72) .. controls (579.95,81.95) and (581.97,79.72) .. (584.47,79.72) .. controls (586.96,79.72) and (588.98,81.95) .. (588.98,84.72) .. controls (588.98,87.48) and (586.96,89.72) .. (584.47,89.72) .. controls (581.97,89.72) and (579.95,87.48) .. (579.95,84.72) -- cycle ;
\draw  [fill={rgb, 255:red, 0; green, 0; blue, 0 }  ,fill opacity=1 ] (537.5,64.72) .. controls (537.5,61.95) and (539.52,59.72) .. (542.01,59.72) .. controls (544.51,59.72) and (546.53,61.95) .. (546.53,64.72) .. controls (546.53,67.48) and (544.51,69.72) .. (542.01,69.72) .. controls (539.52,69.72) and (537.5,67.48) .. (537.5,64.72) -- cycle ;
\draw    (626.92,179.72) -- (597.11,221.72) ;
\draw    (617.89,123.72) -- (626.92,179.72) ;
\draw    (584.47,84.72) -- (617.89,123.72) ;
\draw    (542.01,64.72) -- (583.56,86.72) ;
\draw  [fill={rgb, 255:red, 0; green, 0; blue, 0 }  ,fill opacity=1 ] (413.75,104.72) .. controls (413.75,101.95) and (415.78,99.72) .. (418.27,99.72) .. controls (420.77,99.72) and (422.79,101.95) .. (422.79,104.72) .. controls (422.79,107.48) and (420.77,109.72) .. (418.27,109.72) .. controls (415.78,109.72) and (413.75,107.48) .. (413.75,104.72) -- cycle ;
\draw  [fill={rgb, 255:red, 0; green, 0; blue, 0 }  ,fill opacity=1 ] (411,203.72) .. controls (411,200.95) and (413.02,198.72) .. (415.52,198.72) .. controls (418.01,198.72) and (420.03,200.95) .. (420.03,203.72) .. controls (420.03,206.48) and (418.01,208.72) .. (415.52,208.72) .. controls (413.02,208.72) and (411,206.48) .. (411,203.72) -- cycle ;
\draw    (542.01,64.72) -- (418.27,104.72) ;
\draw    (597.11,221.72) -- (418.27,104.72) ;
\draw    (542.01,64.72) -- (415.52,203.72) ;
\draw    (597.11,221.72) -- (415.52,203.72) ;
\draw    (626.92,179.72) -- (418.27,104.72) ;
\draw    (617.89,123.72) -- (418.27,104.72) ;
\draw    (584.47,84.72) -- (418.27,104.72) ;
\draw    (626.92,179.72) -- (415.52,203.72) ;
\draw    (617.89,123.72) -- (415.52,203.72) ;
\draw    (584.47,84.72) -- (415.52,203.72) ;

\draw (257.05,219) node [anchor=north west][inner sep=0.75pt]  [font=\Large] [align=left] {$\displaystyle A_{1}$};
\draw (290.47,170) node [anchor=north west][inner sep=0.75pt]  [font=\Large] [align=left] {$\displaystyle A_{2}$};
\draw (284.53,102) node [anchor=north west][inner sep=0.75pt]  [font=\Large] [align=left] {$\displaystyle A_{3}$};
\draw (252.53,65) node [anchor=north west][inner sep=0.75pt]  [font=\Large] [align=left] {$\displaystyle A_{4}$};
\draw (200.76,33) node [anchor=north west][inner sep=0.75pt]  [font=\Large] [align=left] {$\displaystyle A_{5}$};
\draw (41.75,92) node [anchor=north west][inner sep=0.75pt]  [font=\Large] [align=left] {$\displaystyle B_{\pm }$};
\draw (33,192) node [anchor=north west][inner sep=0.75pt]  [font=\Large] [align=left] {$\displaystyle B_{0/1}$};
\draw (597.85,219.72) node [anchor=north west][inner sep=0.75pt]  [font=\Large] [align=left] {$\displaystyle A_{1}$};
\draw (631.27,170.72) node [anchor=north west][inner sep=0.75pt]  [font=\Large] [align=left] {$\displaystyle A_{2}$};
\draw (627.34,105.72) node [anchor=north west][inner sep=0.75pt]  [font=\Large] [align=left] {$\displaystyle A_{3}$};
\draw (591.34,63.72) node [anchor=north west][inner sep=0.75pt]  [font=\Large] [align=left] {$\displaystyle A_{4}$};
\draw (537.56,35.72) node [anchor=north west][inner sep=0.75pt]  [font=\Large] [align=left] {$\displaystyle A_{5}$};
\draw (382.56,94.72) node [anchor=north west][inner sep=0.75pt]  [font=\Large] [align=left] {$\displaystyle B_{\pm }$};
\draw (373.81,194.72) node [anchor=north west][inner sep=0.75pt]  [font=\Large] [align=left] {$\displaystyle B_{0/1}$};

\end{tikzpicture}

%% file: figures/venn.tex
\tikzset{every picture/.style={line width=0.75pt}} 

\begin{tikzpicture}[x=0.75pt,y=0.75pt,yscale=-1,xscale=1]

\draw    (473.25,30.1) -- (472.25,356.1) ;
\draw  [fill={rgb, 255:red, 255; green, 255; blue, 255 }  ,fill opacity=1 ] (407,293) -- (538.5,293) -- (538.5,332.2) -- (407,332.2) -- cycle ;
\draw  [fill={rgb, 255:red, 255; green, 255; blue, 255 }  ,fill opacity=1 ] (408,232.2) -- (539.5,232.2) -- (539.5,271) -- (408,271) -- cycle ;
\draw  [fill={rgb, 255:red, 255; green, 255; blue, 255 }  ,fill opacity=1 ] (407,173.2) -- (538.5,173.2) -- (538.5,213) -- (407,213) -- cycle ;
\draw  [fill={rgb, 255:red, 255; green, 255; blue, 255 }  ,fill opacity=1 ] (408,112.2) -- (539.5,112.2) -- (539.5,152) -- (408,152) -- cycle ;
\draw  [fill={rgb, 255:red, 255; green, 255; blue, 255 }  ,fill opacity=1 ] (408,52.2) -- (539.5,52.2) -- (539.5,92) -- (408,92) -- cycle ;
\draw   (45,153.76) .. controls (45,112.03) and (79,78.2) .. (120.95,78.2) .. controls (162.89,78.2) and (196.89,112.03) .. (196.89,153.76) .. controls (196.89,195.49) and (162.89,229.31) .. (120.95,229.31) .. controls (79,229.31) and (45,195.49) .. (45,153.76) -- cycle ;
\draw   (142.61,154.75) .. controls (142.61,113.02) and (176.61,79.19) .. (218.55,79.19) .. controls (260.5,79.19) and (294.5,113.02) .. (294.5,154.75) .. controls (294.5,196.48) and (260.5,230.3) .. (218.55,230.3) .. controls (176.61,230.3) and (142.61,196.48) .. (142.61,154.75) -- cycle ;
\draw   (92.81,239.97) .. controls (92.81,198.24) and (126.81,164.41) .. (168.75,164.41) .. controls (210.7,164.41) and (244.7,198.24) .. (244.7,239.97) .. controls (244.7,281.69) and (210.7,315.52) .. (168.75,315.52) .. controls (126.81,315.52) and (92.81,281.69) .. (92.81,239.97) -- cycle ;

\draw (427,304) node [anchor=north west][inner sep=0.75pt]  [font=\Large] [align=center] {$\displaystyle U_{1} U_{2} U_{3} U_{4}$};
\draw (444,237) node [anchor=north west][inner sep=0.75pt]  [font=\Large] [align=center] {$\displaystyle U_{1}^{\dagger } U_{4}^{\dagger }$};
\draw (445,183) node [anchor=north west][inner sep=0.75pt]  [font=\Large] [align=center] {$\displaystyle U_{5} U_{6}$};
\draw (445,118) node [anchor=north west][inner sep=0.75pt]  [font=\Large] [align=center] {$\displaystyle U_{2}^{\dagger } U_{5}^{\dagger }$};
\draw (458,60) node [anchor=north west][inner sep=0.75pt]  [font=\Large] [align=center] {$\displaystyle U_{7}$};
\draw (36,82) node [anchor=north west][inner sep=0.75pt]  [font=\Large] [align=center] {$\displaystyle C_{1}$};
\draw (270,72) node [anchor=north west][inner sep=0.75pt]  [font=\Large] [align=center] {$\displaystyle C_{2}$};
\draw (91.76,294.12) node [anchor=north west][inner sep=0.75pt]  [font=\Large] [align=center] {$\displaystyle C_{3}$};
\draw (96.77,118.75) node [anchor=north west][inner sep=0.75pt]  [font=\Large] [align=center] {1};
\draw (169.48,133.61) node [anchor=north west][inner sep=0.75pt]  [font=\Large] [align=center] {2};
\draw (165.5,174.24) node [anchor=north west][inner sep=0.75pt]  [font=\Large] [align=center] {3};
\draw (130.64,195.05) node [anchor=north west][inner sep=0.75pt]  [font=\Large] [align=center] {4};
\draw (236.21,138.57) node [anchor=north west][inner sep=0.75pt]  [font=\Large] [align=center] {5};
\draw (200.36,200) node [anchor=north west][inner sep=0.75pt]  [font=\Large] [align=center] {6};
\draw (172.47,260.45) node [anchor=north west][inner sep=0.75pt]  [font=\Large] [align=center] {7};

\end{tikzpicture}